\documentclass[reqno,hidelinks]{amsart}

\usepackage{a4,amsmath,amssymb,amsthm}
\usepackage[top=1in, bottom=1in, left=1.25in, right=1.25in]{geometry}
\usepackage{xfrac}

\usepackage{leftidx}
\usepackage{inputenc}

\usepackage{amsaddr}

\usepackage[T1]{fontenc}

\usepackage{leftidx}
\usepackage{tikz}

\usepackage{hyperref}
\newcounter{mnotecount}[section]

\usepackage{color}

\usepackage{url}

\numberwithin{equation}{section}

\usepackage[mathscr]{eucal}

\usepackage{graphicx}
\graphicspath{pic/}
\def\nablaslash{\mbox{$\nabla \mkern -13mu /$ \!}}

\newcommand{\B}{\mathbf}
\newcommand{\half}{\tfrac{1}{2}}         %
\newcommand{\eps}{\epsilon}

\newcommand{\veps}{\varepsilon}

\newcommand{\inttau}{\int_{\mathcal{D}(0,\tau)}}
\newcommand{\R}{r^2+a^2}

\newcommand{\KDelta}{\Delta}

\newcommand{\di}{\mathrm{d}} 
\newcommand{\p}{\partial}

\newcommand{\intMinfty}{\int_{\mathcal{D}(0,\tau)\cap [6M,\infty)}}

\newcommand{\intMcut}{\int_{\mathcal{D}(0,\tau)\cap [5M,6M]}}

\theoremstyle{plain}
\newtheorem{thm}{Theorem}

\newtheorem{lemma}[thm]{Lemma}

\newtheorem{prop}[thm]{Proposition}

\newtheorem{remark}[thm]{Remark}

\title{Uniform Energy Bound and Morawetz Estimate for Extreme Components of Spin Fields in the Exterior of a Slowly Rotating Kerr Black Hole II: Linearized Gravity}

\author{Siyuan Ma}
\address{Laboratoire Jacques-Louis Lions,
Sorbonne Université, Campus Jussieu,
4 place Jussieu 75005 Paris, France.\\
\email{siyuan.ma@sorbonne-universite.fr}\\
Albert Einstein Institute,
Am M\"uhlenberg 1,
D-14476 Potsdam, Germany}

\begin{document}

\allowdisplaybreaks

%
%
%

\begin{abstract}
This second part of the series treats spin $\pm2$ components (or extreme components), that satisfy the Teukolsky master equation, of the linearized gravity
in the exterior of a slowly rotating Kerr black hole.
For each of these two components, after performing a first-order differential operator once and twice, the resulting equations together with the Teukolsky master equation itself constitute a linear spin-weighted wave system. An energy and Morawetz estimate for spin $\pm 2$ components is proved by treating this system. This is a first step in a joint work \cite{andersson2019stability} in addressing the linear stability of slowly rotating Kerr metrics.
\end{abstract}

\maketitle
\section{Introduction}

The nonlinear stability conjecture of Kerr black holes says that metrics of the subextremal Kerr family of spacetimes $(\mathcal{M},g=g_{M,a})$ ($|a|< M$) are (expected to be) stable against small perturbations of initial data as solutions to the vacuum Einstein equations
\begin{align}\label{eq:EinsteinVacuumEq}
\text{\textbf{Ric}}[g]_{\mu\nu}=0,
\end{align}
$\text{\textbf{Ric}}[g]_{\mu\nu}$ being the Ricci curvature tensor of the metric. An important step toward the resolution of this conjecture of nonlinear stability of Kerr metrics is to show the linear stability, i.e., to show the asymptotic decay of linearized gravitational perturbations (also called \textquotedblleft{linearized gravity\textquotedblright}) around Kerr metrics.

As a beginning step in proving linear stability of Kerr metrics, we consider the gauge invariant extreme components of the Weyl tensor which satisfy the well-known Teukolsky master equation \cite{Teu1972PRLseparability} and govern the dynamics of linearized gravity,
and prove both a uniform bound of a positive definite energy and a Morawetz estimate
for these extreme components on a slowly rotating Kerr background (where $|a|/M\ll 1$ is sufficiently small). In a joint work \cite{andersson2019stability}, this basic energy and Morawetz estimate, whereas called \textquotedblleft{Basic decay condition\textquotedblright} or \textquotedblleft{BEAM condition\textquotedblright} in  \cite{andersson2019stability}, is utilized  to obtain
further strong decay estimates for the extreme components and the full linear stability of slowly rotating Kerr metrics.

\subsection{Kerr Metric}
For the purpose that this paper can be read independently, we review in this subsection the setup of the Kerr metric  and notation from the first part \cite{Ma17spin1Kerrlatest} of this series.

A Kerr spacetime $(\mathcal{M},g=g_{M,a})$ \cite{kerr63} has a metric given in Boyer--Lindquist (B--L) coordinates \cite{boyer:lindquist:1967} $(t,r,\theta,\phi)$ by
\begin{align}\label{eq:KerrMetricBoyerLindquistCoord}
g_{M,a}= & -\left(1-\tfrac{2Mr}{\Sigma} \right) \di t^2 -\tfrac{2Mar \sin^2\theta}{\Sigma}(\di t \di \phi + \di \phi \di t) \nonumber\\
 & + \tfrac{\Sigma}{\Delta} \di r^2 + \Sigma \di \theta^2 +\tfrac{\sin^2\theta}{\Sigma} \left[(r^2+a^2)^2 -a^2\Delta \sin^2\theta\right]\di \phi^2
\end{align}
with
\begin{align}\label{eq:kerrfunctions}
\Delta(r)&= r^2 -2Mr +a^2 \   \   \   \   \   \text{    and    }  & \Sigma(r,\theta) = r^2+a^2 \cos^2\theta ,
\end{align}
and describes a rotating, stationary (with $\partial_t$ Killing), axisymmetric (with $\partial_{\phi}$ Killing), asymptotically flat solution to vacuum Einstein equations \eqref{eq:EinsteinVacuumEq}. The Schwarzschild metric \cite{schw1916} is obtained by setting $a=0$.

The region we consider is the domain of outer communication (DOC)
\begin{equation}\label{def:DOC}
\mathcal{D}=\overline{\left\{(t,r,\theta,\phi)\in \mathbb{R}\times (r _+,\infty)\times \mathbb{S}^2\right\}},
\end{equation}
where $r_+=M+\sqrt{M^2-a^2}$ is the value of the larger root of $\Delta(r)=0$ and corresponds to the location of the event horizon. By symmetry (cf. Section \ref{sect:MainTheorems}),
we focus only on the future development with boundary the future event horizon $\mathcal{H}^+$.
In this paper, a slowly rotating Kerr spacetime should always be understood as the DOC of a Kerr spacetime endowed with the Kerr metric $g=g_{M,a}$ with sufficiently small $|a|/M\ll1$.

The tortoise coordinate $r^*$ is defined by:
\begin{align}
\frac{\di r^*}{\di r}=\frac{r^2+a^2}{\Delta},\   \   \   \   \   \    r^*(3M)=0,
\end{align}
and we call $(t,r^*,\theta,\phi)$ the tortoise coordinate system.
However, both the B-L and tortoise coordinate systems fail to extend across the future event horizon $\mathcal{H}^+$ due to the singularity in the metric coefficients. Instead,
an ingoing Kerr coordinate system $(v,r,\theta,\tilde{\phi})$, which is regular on $\mathcal{H}^+$, is defined by:
\begin{equation}\label{def:IngoingEddiFinkerCoord}
\left\{
  \begin{array}{ll}
    \di v=\di t +\di r^*, \\
    \di \tilde{\phi}=\di \phi +a(r^2 +a^2)^{-1}\di r^*.\\
  \end{array}
\right.
\end{equation}
Moreover, via gluing the coordinate system $(\vartheta=v-r, r, \theta,\tilde{\phi})$ near horizon with the B--L coordinate system $(t,r,\theta,\phi)$ away from horizon smoothly, a global Kerr coordinate system $(t^*,r, \theta,\phi^*)$ can be given by
\begin{equation}\label{def:globalkerrcoord}
\left\{
  \begin{array}{ll}
    t^*=t+\chi_{1}(r)\left(r^*(r)-r-r^*(r_0)+r_0\right),   \\
    \phi^*=\phi+\chi_{1}(r)\acute{\phi}(r)\ \ \text{mod}\ 2\pi, \ \ \di \acute{\phi}/\di r=a/\Delta.
  \end{array}
\right.
\end{equation}
The smooth cutoff function $\chi_1(r)$ here equals $1$ in $[r_+,M+r_0/2]$ and identically vanishes for $r\geq r_0$ with $r_0=r_0(M)$ fixed in Section \ref{sect:Redshift}, and is chosen such that on the initial spacelike hypersurface
\begin{equation}\label{def:InitialHypersurface}
\Sigma_{0}=\left\{(t^*,r,\theta,\phi^*)|t^*=0\right\}\cap \mathcal{D},
\end{equation}
there exist two universal positive constants $c(M)$ and $C(M)$ so that
\begin{equation}\label{eq:initialtimehypersurfacegradientmodular}
c(M)\leq-g(\nabla t^*,\nabla t^*)|_{\Sigma_0}\leq C(M).
\end{equation}
Here the initial hypersurface could be taken as $\{t^*=D\}$ hypersurface for any real value $D$, but for convenience, we take it as in \eqref{def:InitialHypersurface}.

In these coordinate systems, it is manifest that
\begin{equation}\label{eq:VFieldTandPhi}
\partial_{t^*}=\partial_t\triangleq T\ \ \text{and}\ \  \partial_{\phi^*}=\partial_{\tilde{\phi}}=\partial_{\phi}.
\end{equation}
Denote $\varphi_{\tau}$ as the $1$-parameter family of diffeomorphisms generated by $T$ and define constant-time spacelike hypersurfaces satisfying \eqref{eq:initialtimehypersurfacegradientmodular} as well:
\begin{equation}\label{def:constanttimeHypersurface}
\Sigma_{\tau}=\varphi_{\tau}\left(\Sigma_0\right)=\left\{(t^*,r,\theta,\phi^*)|t^*=\tau\right\}\cap \mathcal{D}.
\end{equation}
We finally adopt the notations for any $0\leq \tau_1< \tau_2$ that
\begin{align*}
\mathcal{D}(\tau_1, \tau_2)=\bigcup_{\tau\in [\tau_1, \tau_2]}\Sigma_{\tau}, \quad \text{and}\quad
\mathcal{H}^+(\tau_1, \tau_2)=\partial\mathcal{D}(\tau_1, \tau_2) \cap \mathcal{H}^+.
\end{align*}
The reader may refer to the Penrose diagram Fig. \ref{fig:penrosediagram}.
\begin{figure}[!h]
\centering
\raisebox{-0.5\height}{
\includegraphics{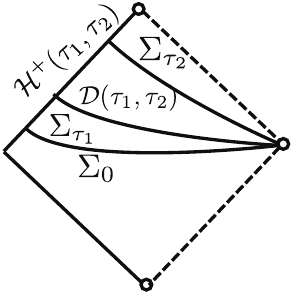}}
\caption{Penrose diagram}
\label{fig:penrosediagram}
\end{figure}

We shall make use of a volume element for the hypersurface $\Sigma_{\tau} (\tau\geq0)$
\begin{equation}\label{def:volumeformhypersurface}
\di \text{Vol}_{\Sigma_{\tau}}=\Sigma \di r \sin \theta \di \theta \di \phi^* \ \  \text{in global Kerr coordinates},
\end{equation}
and the volume form of the manifold is
\begin{align}\label{eq:VolumeForm}
\di {\text{Vol}}_{\mathcal{M}}=
\left\{
  \begin{array}{ll}
    \Sigma \di t\di r\sin\theta \di \theta \di \phi & \text{in B--L coordinates,} \\
\Sigma \di t^*\di r\sin\theta \di \theta \di \phi^* & \text{in global Kerr coordinates}.
  \end{array}
  \right.
\end{align}
Note that $\di \text{Vol}_{\Sigma_{\tau}}$ is a convenient reference volume form in calculations and in stating the integral estimates, but it is not the induced volume form on $\Sigma_\tau$. Unless otherwise specified, we will always suppress these volume forms associated to the integrals in this paper.

\subsection{Linearized Gravity and Teukolsky Master Equation}\label{sect:lingraandTME}
Following Newman--Penrose (N--P) formalim  \cite{newmanpenrose62,newmanpenrose63errata}, we obtain the complete five N--P components
\begin{align}\label{eq:MaxwellNPcomponentswithnosuperscript}
\Phi_0 = &-\B{W}_{l m l m} ,\ \ \ \ \Phi_1=-\B{W}_{l n l m},\ \ \ \ \Phi_2 = -\B{W}_{l m \overline{m} n} ,\notag\\
&\Phi_3=-\B{W}_{l n \overline{m} n},\ \ \ \ \ \Phi_4 = -\B{W}_{n\overline{m}n\overline{m}}
\end{align}
by projecting the Weyl tensor $\B{W}_{\alpha\beta\gamma\delta}$ onto the Kinnersley null tetrad $(l,n,m,\overline{m})$ \cite{Kinnersley1969tetradForTypeD} in B--L coordinates:
\begin{align}\label{eq:Kinnersley tetrad}
l^\mu &= \tfrac{1}{\Delta}(r^2+a^2 , \Delta , 0 , a), \notag\\
n^\nu &= \tfrac{1}{2\Sigma} (r^2+a^2 , - \Delta , 0 , a), \notag\\
m^\mu &= \tfrac{1}{\sqrt{2} \bar{\rho}}\left(i a \sin{\theta},0 , 1, \tfrac{i}{\sin{\theta}}\right),
\end{align}
and $\overline{m}^{\mu}$ and $\bar{\rho}$ being the complex conjugate of $m^{\mu}$ and $\rho = r- i a \cos{\theta}$ respectively.
The full set of N-P equations, comprising the commutation relations, the Ricci identities, the eliminant relations and the Bianchi identities in \cite[Chapter 1.8]{MR1647491},
is then a coupled first-order differential system linking the tetrad, the spin coefficients and these five N--P components.
On a Kerr background,
\begin{align}\label{eq:KerrbackgroundNPcomps}
&\Phi_0=\Phi_1=\Phi_3=\Phi_4=0,&\Phi_2 =-M\bar{\rho}^{-3}.
\end{align}
We perturb in the N--P equations all the tetrad components, the spin coefficients and the five N-P components by $l^T=l +l^P$, $\kappa^T=\kappa+\kappa^P$,\footnote{$\kappa$ is one of the spin coefficients used in \cite[Chapter 1.8]{MR1647491}.} $\Phi_0^T=\Phi_0+\Phi_0^P$, etc, and the complete set of equations for linearized gravity is then obtained from the N--P equations by keeping the perturbation terms (with superscript $P$) only to first order. The perturbed extreme components $\Phi_0^T$ and $\Phi_4^T$ (which are equal to $\Phi_0^P$ and $\Phi_4^P$) for linearized gravity
are the \textquotedblleft{ingoing and outgoing radiative parts,\textquotedblright}
and are invariant under gauge transformations and infinitesimal tetrad rotations. From now on, we will drop the superscript and still denote these perturbed extreme components as $\Phi_0$ and $\Phi_4$.

Teukolsky \cite{Teu1972PRLseparability} derived the decoupled equations on Kerr backgrounds for the spin $s=\pm2$ components
\begin{equation}
\begin{split}
\psi_{[+2]}= \Delta^2 \Phi_0  \ \ \text{ and } \ \ & \psi_{[-2]}=\Delta^{-2}\rho^4\Phi_4,
\end{split}
\label{eq:spinsfields}
\end{equation}
and showed that these decoupled equations are in fact separable and governed by a single master equation--the celebrated \emph{Teukolsky Master Equation} (TME)--given in B--L coordinates by
\begin{equation}\label{eq:TME}
\begin{split}
& -\left[\tfrac{(r^2+a^2)^2}{\Delta} -a^2 \sin^2{\theta} \right] \tfrac{\partial^2 \psi_{[s]}}{\partial t^2} - \tfrac{4Mar}{\Delta} \tfrac{\partial^2 \psi_{[s]}}{\partial t \partial \phi}-\left[\tfrac{a^2}{\Delta} -\tfrac{1}{\sin^2{\theta}} \right] \tfrac{\partial^2 \psi_{[s]}}{\partial \phi^2}   \\
&  +\Delta^{s} \tfrac{\partial}{\partial r} \left( \Delta^{-s+1} \tfrac{\partial \psi_{[s]}}{\partial r} \right) + \tfrac{1}{\sin{\theta}} \tfrac{\partial}{\partial \theta} \left( \sin{\theta} \tfrac{\partial \psi_{[s]}}{\partial \theta}\right) +2s \left[ \tfrac{a(r-M)}{\Delta} + \tfrac{i \cos{\theta}}{\sin^2{\theta} } \right] \tfrac{\partial \psi_{[s]}}{\partial \phi} \\
&  +2s\left[ \tfrac{M(r^2-a^2)}{\Delta} -r -ia \cos{\theta} \right] \tfrac{\partial \psi_{[s]}}{\partial t}- (s^2 \cot^2{\theta} +s) \psi_{[s]} = 0 .
\end{split}
\end{equation}
In fact, TME is valid for general spin $s$ fields with $s=\frac{n}{2}, n\in \mathbb{Z}$. In particular, the $s=0$ case is the scalar wave equation, and the $s=\pm 1$ cases are the governing equations of spin $\pm 1$ components of the Maxwell field.

The Kinnersley tetrad is, however, singular on $\mathcal{H}^+$ in ingoing Kerr coordinates, suggesting that the perturbed N--P components
are not all regular there. We perform a null rotation by
\begin{equation}\label{null rotation}
\left\{
  \begin{array}{ll}
    l \rightarrow \tilde{l}=\Delta/(2\Sigma)\cdot l, \\
    n \rightarrow \tilde{n}=(2\Sigma)/\Delta \cdot n,\\
    m \rightarrow m,
  \end{array}
\right.
\end{equation}
and the resulting tetrad $(\tilde{l},\tilde{n},m,\overline{m})$, namely the Hawking--Hartle (H--H) tetrad \cite{HHtetrad72}, is in fact regular up to and on $\mathcal{H}^+$ in global Kerr coordinates.
The regular extreme components of linearized gravity
in the regular H--H tetrad are then
\begin{align}\label{def:regularNPComps}
\left\{
  \begin{array}{ll}
    \widetilde{\Phi_0}(\B{W})&=-\B{W}_{\tilde{l} m\tilde{l}m}
    =\tfrac{1}{4\Sigma^2}\psi_{[+2]}, \\
    \widetilde{\Phi_4}(\B{W})&=-\B{W}_{\tilde{n}\overline{m} \tilde{n}\overline{m}}
    =\tfrac{4\Sigma^2}{\rho^4}\psi_{[-2]}.
  \end{array}
\right.
\end{align}
The results in this paper will be with respect to complex scalars $\widetilde{\Phi_0}$ and $\widetilde{\Phi_4}$.

\subsection{Coupled Systems}

Denote the future-directed ingoing and outgoing principal null vector fields in B--L coordinates
\begin{align}\label{def:VectorFieldYandV}
Y&\triangleq \tfrac{(r^2+a^2)\partial_t +a\partial_{\phi}}{\Delta}-\partial_r, \ &\ V&\triangleq \tfrac{(r^2+a^2)\partial_t+
a\partial_{\phi}}{\Delta}+\partial_r.
\end{align}
From TME \eqref{eq:TME}, the equations for $\psi_{[+2]}$ and $\psi_{[-2]}$  are
\begin{subequations}\label{eq:TME0orderPosiandNegaS2}
\begin{align}\label{eq:TME0orderS2}
&\left(\Sigma\Box_g + 4i\left(\tfrac{\cos\theta}{\sin^2 \theta}\partial_{\phi}-a\cos \theta \partial_t\right) - (4\cot^2\theta +2)\right)\psi_{[+2]} =-4Z\psi_{[+2]},\\
\label{eq:TME0orderNegaS2}
&\left(\Sigma\Box_g - 4i\left(\tfrac{\cos\theta}{\sin^2 \theta}\partial_{\phi}-a\cos \theta \partial_t\right) - (4\cot^2\theta -2)\right) \psi_{[-2]} =4Z\psi_{[-2]},
\end{align}
\end{subequations}
with $Z=(r-M)Y-2r\partial_t$.
Construct from $\psi_{[+2]}$ and $\psi_{[-2]}$ the quantities
\begin{subequations}\label{eq:DefOfphi012BothSpinS2}
\begin{align}\label{eq:DefOfphi012PosiSpinS2}
\left\{
  \begin{array}{ll}
    \phi^0_{+2}&=\psi_{[+2]}/r^4;\\
\phi^1_{+2}&=\left(rYr\right)(\phi^0_{+2});\\ \phi^2_{+2}&=\left(rYr\right)
\left(rYr\right)(\phi^0_{+2}),
  \end{array}
\right.
\end{align}
and
\begin{align}\label{eq:DefOfphi012NegaSpinS2}
\left\{
  \begin{array}{ll}
 \phi^0_{-2}&=\Delta^2/r^4\psi_{[-2]};\\
\phi^1_{-2}&=-\left(rVr\right)(\phi^0_{-2});\\
\phi^2_{-2}&=\left(rVr\right)
\left(rVr\right)(\phi^0_{-2}).
  \end{array}
\right.
\end{align}
\end{subequations}
The upper index here denotes the number of times the differential operator $rYr$ or $-(rVr)$ is performed.
Note that though it is not $V$ but rather $\frac{\Delta}{r^2+a^2} V$ which is a regular vector field on $\mathcal{H}^+$, by the second relation in \eqref{def:regularNPComps}, the variables $\left\{\phi_{-2}^i\right\}_{i=0,1,2}$ in \eqref{eq:DefOfphi012NegaSpinS2} are indeed smooth up to and on future horizon if the regular N--P scalar $\widetilde{\Phi_4}$ is. In global Kerr coordinates, the regular vector field $Y$ equals $-\partial_r +\partial_{t^*}$ in $[r_+, M+r_0/2]$ and is $\tfrac{\R}{\Delta}\partial_{t^*}+\tfrac{a}{\Delta}\partial_{\phi^*}-\partial_r$ for $r\geq r_0$.

The rescalings in the definitions of $\phi_{\pm 2}^0$ are such that one can rewrite the first order $Z$-derivative terms in terms of $\phi_{\pm 2}^1$ plus first order derivative terms with $a$-dependent coefficients. From the commutator relations in Appendix \ref{sect:commutatorwaveandYV}, one can derive the equations satisfied by $\phi_{\pm 2}^1$ and $\phi_{\pm 2}^2$. The coupled systems of equations for these quantities are
\begin{subequations}\label{eq:ReggeWheeler Phi^012KerrS2}
\begin{align}
\label{eq:ReggeWheeler Phi^0KerrS2}
\mathbf{L}^0_{+2}\phi^0_{+2}
=F^{0}_{+2}=&\tfrac{4(r^2-3Mr+2a^2)}{r^3}\phi^1_{+2}
-\tfrac{8(a^2\partial_t+a\partial_{\phi})\phi^0_{+2}}{r},\\
\label{eq:ReggeWheeler Phi^1KerrS2}
\mathbf{L}^1_{+2}\phi^1_{+2} =F^1_{+2}=&\tfrac{2(r^2-3Mr+2a^2)}{r^3}\phi^2_{+2}
+\tfrac{6Mr-12a^2}{r}\phi^0_{+2}\notag\\
&-\tfrac{4(a^2\partial_t+a\partial_{\phi})\phi^1_{+2}}{r}
-6(a^2\partial_t+a\partial_{\phi})\phi^0_{+2},\\
\label{eq:ReggeWheeler Phi^2KerrS2}
\mathbf{L}^1_{+2}\phi^2_{+2} =F^2_{+2}
=&-8(a^2\partial_t+a\partial_{\phi})\phi^1_{+2}-12a^2\phi^0_{+2},
\end{align}
\end{subequations}
and
\begin{subequations}\label{eq:ReggeWheeler Phi^012KerrNegaS2}
\begin{align}\label{eq:ReggeWheeler Phi^0KerrNegaS2}
\mathbf{L}^0_{-2}\phi^0_{-2}
=F^{0}_{-2}=&\tfrac{4(r^2-3Mr+2a^2)}{r^3}\phi^1_{-2}
+\tfrac{8(a^2\partial_t+a\partial_{\phi})\phi^0_{-2}}{r},\\
\label{eq:ReggeWheeler Phi^1KerrNegaS2}
\mathbf{L}^1_{-2}\phi^1_{-2} =F^{1}_{-2}=&\tfrac{2(r^2-3Mr+2a^2)}{r^3}\phi^2_{-2}
+\tfrac{6Mr-12a^2}{r}\phi^0_{-2}\notag\\
&+\tfrac{4(a^2\partial_t+a\partial_{\phi})\phi^1_{-2}}{r}
+6(a^2\partial_t+a\partial_{\phi})\phi^0_{-2},\\
\label{eq:ReggeWheeler Phi^2KerrNegaS2}
\mathbf{L}^1_{-2}\phi^2_{-2}=F^{2}_{-2}
=&8(a^2\partial_t+a\partial_{\phi})\phi^1_{-2}-12a^2\phi^0_{-2},
\end{align}
\end{subequations}
respectively.\footnote{This application of the first-order differential operators to the spin $\pm 2$ components is closely related to \emph{Chandrasekhar
transformation} \cite{chandrasekhar1975linearstabSchw}.}
The subscript $+2$ or $-2$ here indicates the spin weight $s=\pm2$, and the operators $\mathbf{L}^0_s$ and $\mathbf{L}^1_s$, given by
\begin{subequations}\label{def:SWRW01operatorS2}
\begin{align}
\label{def:SWRW0operatorS2}
\mathbf{L}^0_s&=\Sigma \Box_g+2is\left(\tfrac{\cos\theta}{\sin^2 \theta}\partial_{\phi}-a\cos \theta \partial_t\right)-s^2\left(\cot^2 \theta+\tfrac{r^2+2Mr-2a^2}{2r^2}\right),\\
\label{def:SWRWoperatorS2}
\mathbf{L}^1_s&=\Sigma \Box_g+2is\left(\tfrac{\cos\theta}{\sin^2 \theta}\partial_{\phi}-a\cos \theta \partial_t\right)-s^2\left(\cot^2 \theta+\tfrac{r^2-2Mr+2a^2}{r^2}\right),
\end{align}
\end{subequations}
are both \emph{spin-weighted wave operators}, but with different potentials. The equations for $\phi^i_s$ in \eqref{eq:ReggeWheeler Phi^012KerrS2} and \eqref{eq:ReggeWheeler Phi^012KerrNegaS2} are in either form of the following equations:
\begin{subequations}
\begin{align}\label{eq:RewrittenFormofISWWEphi0OpeForm}
\mathbf{L}_s^0\psi&=F;\\
\label{eq:RewrittenFormofSWRWEOpeForm}
\mathbf{L}_s^1\psi&=F,
\end{align}
\end{subequations}
both of which are called  in this paper as \emph{inhomogeneous spin-weighted wave equations} (ISWWE).
When there is no confusion of which spin component we are treating, we may suppress the subscript of $\phi^i_s$ and simply write as $\phi^i$.
\begin{remark}
After making the substitutions $\partial_t\leftrightarrow -i\omega$, $\partial_{\phi}\leftrightarrow im$, and separating the operators $\mathbf{L}_s^k$ $(k=0,1)$, the angular parts are the spin-weighted spheroidal harmonic operator of angular Teukolsky equation. The radial operator of $\mathbf{L}_s^1$ is the sum of the radial part of the rescaled scalar wave operator $\Sigma \Box_g$ and a potential $s^2(r^2-\Delta-a^2)/r^2$,
and reduces to the radial operator for Regge--Wheeler equation \cite{ReggeWheeler1957} when on Schwarzschild background $(a=0)$, while the one of $\mathbf{L}_s^0$ is the sum of the radial part of $\Sigma \Box_g$ and another potential $s^2(\Delta+a^2)/(2r^2)$. See more details in Section \ref{sect:decompSchwS2} for Schwarzschild case and Section \ref{sect:SeparateAngAndRadialEqs} for Kerr case.
\end{remark}

\subsection{Main Theorem}\label{sect:MainTheorems}

The TME admits a symmetry that $\Delta^s \psi_{[-s]}(-t,r,\theta,-\phi)$ and $\psi_{[s]}(t,r,\theta,\phi)$ satisfy the same equation; hence we focus only on the future time development in this paper, and one can obtain the analogous estimates in the past time direction.

For any complex-valued smooth function $\psi: \mathcal{M}\rightarrow \mathbb{C}$ with spin weight $s$, we define in global Kerr coordinates that  for any $\tau\geq0$,
\begin{equation}\label{eq:ModuloSquareofDeris}
|\partial\psi(t^*,r,\theta,\phi^*)|^2=|\partial_{t^*}\psi|^2
+|\partial_r\psi|^2+|\nablaslash\psi|^2,
\end{equation}
\begin{equation}
 {E}_{\tau}(\psi)= \int_{\Sigma_{\tau}}|\partial\psi|^2,
\end{equation}
and in ingoing Kerr coordinates that  for any $\tau_2>\tau_1\geq 0$,
\begin{equation}
 {E}_{\mathcal{H}^+(\tau_1,\tau_2)}(\psi)= \int_{\mathcal{H}^+(\tau_1,\tau_2)}(|\partial_v\psi|^2
 +|\nablaslash\psi|^2)r^2\di v\sin\theta \di \theta \di \tilde{\phi}.
\end{equation}
The $\nablaslash$ used here are not the standard rotational angular derivatives $\check{\nablaslash}$ on sphere $\mathbb{S}^2(t^*,r)$, but the spin-weighted version of them, i.e., $\nablaslash$ could be any one of $\nablaslash_j$ $(j=1,2,3)$ defined by
\begin{align}\label{SpinWeightedAngularDerivaBasisOnSphere}
\left\{
  \begin{array}{ll}
    r\nablaslash_1&=r\check{\nablaslash}_1-\tfrac{is\cos\phi}{\sin\theta}
    =(-\sin\phi\partial_{\theta}-
\tfrac{\cos\phi}{\sin\theta}\cos\theta\partial_{\phi^*})
-\tfrac{is\cos\phi}{\sin\theta},\\
r\nablaslash_2&=r\check{\nablaslash}_2-\tfrac{is\sin\phi}{\sin\theta}
=(\cos\phi\partial_{\theta}
-\tfrac{\sin\phi}{\sin\theta}\cos\theta\partial_{\phi^*})
-\tfrac{is\sin\phi}{\sin\theta},\\
r\nablaslash_3&=r\check{\nablaslash}_3=\partial_{\phi^*}.
  \end{array}
\right.
\end{align}
In global Kerr coordinates, we can express the squared absolute value of $\nablaslash\psi$ as
\begin{align}\label{def:nablaslashModuleSquare}
|\nablaslash\psi|^2=\sum_{i=1,2,3}\left|\nablaslash_i \psi\right|^2&=\frac{1}{r^2}\Big(|\partial_{\theta}\psi|^2
+\left|\tfrac{\cos\theta\partial_{\phi^*}\psi+
is\psi}{\sin\theta}\right|^2
+|\partial_{\phi^*}\psi|^2\Big)\notag\\
&=\frac{1}{r^2}\Big(|\partial_{\theta}\psi|^2
+\left|\tfrac{\partial_{\phi^*}\psi+
is\cos\theta\psi}{\sin\theta}\right|^2
+s^2|\psi|^2\Big).
\end{align}
In particular, note from \eqref{def:nablaslashModuleSquare} that $|\nablaslash\psi|^2$, and thus
$|\partial\psi|^2$, already have control over $r^{-2}|\psi|^2$ if the spin weight $s\neq 0$.
The same expressions \eqref{SpinWeightedAngularDerivaBasisOnSphere} and \eqref{def:nablaslashModuleSquare} hold true in B--L coordinates and ingoing Kerr coordinates from \eqref{eq:VFieldTandPhi}. For convenience of calculations, we may always refer to these expressions with $\partial_{\phi}$ in place of $\partial_{\phi^*}$.

For any smooth function $\psi$ with spin weight $s$, we define for any multi-index $\mathrm{k}=(\mathrm{k}_1,\mathrm{k}_2,\mathrm{k}_3,\mathrm{k}_4,\mathrm{k}_5)$ with $\mathrm{k}_i\geq 0$ $(i=1,2,3,4,5)$ that
\begin{equation}\label{def:multiindexderi}
\partial^{\mathrm{k}} \psi=\partial_{t^*}^{\mathrm{k}_1}\partial_r^{\mathrm{k}_2}\nablaslash_{1}^{\mathrm{k}_3}
\nablaslash_{2}^{\mathrm{k}_4}\nablaslash_{3}^{\mathrm{k}_5}\psi.
\end{equation}
Define the length of such a multi-index $\mathrm{k}$ by
\begin{align}
\vert \mathrm{k}\vert ={} \mathrm{k}_1+\mathrm{k}_2+\mathrm{k}_3
+\mathrm{k}_4+\mathrm{k}_5 .
\end{align}
Denote a few Morawetz densities by\footnote{We should distinguish among these different notations that a tilde means that there is no extra $r^{-\delta}$ power in the coefficients of $\partial_r$- and $\partial_{t^*}$-derivatives term and a subscript $\text{deg}$ means there is the trapping degeneracy in the trapped region, and vice versa.}
\begin{subequations}\label{def:mathbbMpsiWholeSet}
\begin{align}
\label{def:mathbbMpsi}
\mathbb{M}_{\text{deg}}(\psi)
={}&r^{-1-\delta}|\partial_{r}\psi|^2
+\chi_{\text{trap}}(r) (r^{-1-\delta}|\partial_{t^*} \psi|^2+r^{-1}|\nablaslash\psi|^2)+r^{-3}|\psi|^2,\\
\label{def:MpsiII}
\mathbb{M}(\psi)={}&r^{-1-\delta}(|\partial_{r}\psi|^2+|\partial_{t^*} \psi|^2)+r^{-1}|\nablaslash\psi|^2+r^{-3}|\psi|^2,\\
\label{def:widetildemathbbMpsi}
\widetilde{\mathbb{M}}_{\text{deg}}(\psi)
={}&r^{-1}|\partial_{r}\psi|^2
+\chi_{\text{trap}}(r)r^{-1}(|\partial_{t^*} \psi|^2+|\nablaslash\psi|^2)+r^{-3}|\psi|^2,\\
\label{def:Mpsi}
\widetilde{\mathbb{M}}(\psi)={}&
r^{-1}|\partial\psi|^2+r^{-3}|\psi|^2.
\end{align}
\end{subequations}
Here, $\chi_{\text{trap}}(r)=1-\eta_{[r^-_{\text{trap}}, r^+_{\text{trap}}]}(r)$, $\eta_{[r^-_{\text{trap}}, r^+_{\text{trap}}]}(r)$ is the indicator function in the radius region bounded by $r^{-}_{\text{trap}}=2.9M<3M<r^{+}_{\text{trap}}=3.1M$, and $\delta\in (0,1/2)$ is an arbitrary constant.

\begin{thm}\label{thm:EneAndMorEstiExtremeCompsNoLossDecayVersion2}
Consider the linearized gravity in the DOC of a slowly rotating Kerr spacetime $(\mathcal{M},g=g_{M,a})$. Given any smooth\footnote{In fact, the N--P components should be viewed as sections of a complex line bundle. Therefore, \textquotedblleft{smooth\textquotedblright}  means that these components and their derivatives to any order with respect to $(\partial_{t^*},\partial_r, \nablaslash_1, \nablaslash_2, \nablaslash_3)$ are continous.} regular extreme components $\widetilde{\Phi_0}$, $\widetilde{\Phi_4}$ as in Section \ref{sect:lingraandTME} which vanish near spatial infinity,
then for any $0<\delta<1/2$ and nonnegative integer $n$, there exist universal constants $\veps_0$, $R=R(M)$ and $C=C(M,\delta,\Sigma_0,n)=C(M,\delta,\Sigma_{\tau},n)$ such that for all $|a|/M\leq \veps_0$ and any $\tau\geq 0$,
we have the following energy and Morawetz estimate for the regular extreme components:
\begin{subequations}
\label{eq:MoraEstiFinal(2)KerrRegularpsiBothSpinComp}
\begin{align}
\label{eq:MoraPosiSpinMain}
&\sum_{\vert \mathrm{k}\vert\leq n}\int_{\mathcal{D}(0,\tau)} \left(\mathbb{M}_{\text{deg}}(\partial^{\mathrm{k}}\Phi^{(2)}_0)
+\widetilde{\mathbb{M}}(\partial^{\mathrm{k}}\Phi^{(1)}_0)
+\widetilde{\mathbb{M}}(\partial^{\mathrm{k}}\Phi^{(0)}_0)\right)\notag\\
&+\sum_{\vert \mathrm{k}\vert\leq n}\sum_{i=0}^2\left({E}_{\tau}(\partial^{\mathrm{k}}\Phi^{(i)}_0)
+{E}_{\mathcal{H}^+(0,\tau)}(\partial^{\mathrm{k}}\Phi^{(i)}_0)\right)
\leq {}
C\sum_{\vert \mathrm{k}\vert\leq n}\sum_{i=0}^2{E}_{0}(\partial^{\mathrm{k}}\Phi^{(i)}_0),\\
\label{eq:MoraNegaSpinMain}
&\sum_{\vert \mathrm{k}\vert\leq n}\int_{\mathcal{D}(0,\tau)} \left(\mathbb{M}_{\text{deg}}(\partial^{\mathrm{k}}\Phi^{(2)}_4)
+\mathbb{M}(\partial^{\mathrm{k}}\Phi^{(1)}_4)
+\mathbb{M}(\partial^{\mathrm{k}}\Phi^{(0)}_4)\right)\notag\\
&
+\sum_{\vert \mathrm{k}\vert\leq n}\sum_{i=0}^2\left({E}_{\tau}(\partial^{\mathrm{k}}\Phi^{(i)}_4)
+{E}_{\mathcal{H}^+(0,\tau)}(\partial^{\mathrm{k}}\Phi^{(i)}_4)\right)
\leq {}
C\sum_{\vert \mathrm{k}\vert\leq n}\sum_{i=0}^2{E}_{0}(\partial^{\mathrm{k}}\Phi^{(i)}_4).
\end{align}
\end{subequations}
Here, the set $(\Phi^{(0)}_j, \Phi^{(1)}_j, \Phi^{(2)}_j)$ for $j=0,4$ takes
\begin{subequations}
\begin{align}
\label{def:varphi012positive}
\Phi^{(0)}_{0}&=r^{4-\delta}\widetilde{\Phi_0},
& \Phi^{(1)}_{0}&=r^{4-\delta}Y\widetilde{\Phi_0},
& \Phi^{(2)}_{0}&=r^4YY\widetilde{\Phi_0}; \\
\label{def:varphi012negative}
\Phi^{(0)}_{4}&=\widetilde{\Phi_4},
& \Phi^{(1)}_{4}&=\tfrac{r\Delta}{\R}V(r\Phi^{(0)}_{4}),
& \Phi^{(2)}_{4}&=\tfrac{r\Delta}{\R}V(r\Phi^{(1)}_{4}).
\end{align}
\end{subequations}
\end{thm}

\begin{remark}
The trapping degeneracy for the Morawetz densities $\mathbb{M}_{\text{deg}}(\partial^{\mathrm{k}}\Phi^{(2)}_0)$ and $\mathbb{M}_{\text{deg}}(\partial^{\mathrm{k}}\Phi^{(2)}_4)$ with $\vert {\mathrm{k}} \vert\leq n-1$ can actually be removed. We shall only focus on proving the $n=0$ case until Section \ref{sect:highorderS2}. As shown in Section \ref{sect:highorderS2}, the general $n\geq 0$ cases follow from an induction in $n$.
\end{remark}
\begin{remark}
As stated above, consider the $n=0$ case.
The energy and Morawetz estimate \eqref{eq:MoraEstiFinal(2)KerrRegularpsiBothSpinComp} is obtained by treating systems \eqref{eq:ReggeWheeler Phi^012KerrS2} and \eqref{eq:ReggeWheeler Phi^012KerrNegaS2} of $\phi^i_s$ and
is a single estimate at three levels of regularity for each extreme component, since $\phi^2_s$ involves at most second order derivatives of $\phi^0_s$. Therefore, in spite of the well-known trapping phenomenon, we prove Morawetz estimates for $\phi^0_s$ and $\phi^1_s$ which are in fact non-degenerate in the trapped region. However, the three levels of regularity must be treated simultaneously. On one hand, to estimate the inhomogeneous terms on the right-hand side of \eqref{eq:ReggeWheeler Phi^012KerrS2} and \eqref{eq:ReggeWheeler Phi^012KerrNegaS2}, it is necessary to eliminate the trapping degeneracy in the Morawetz estimates for $\phi^0_s$ and $\phi^1_s$ by considering one more order of derivative; on the other hand, it is possible to close the three estimates simultaneously, because the right-hand side of \eqref{eq:ReggeWheeler Phi^012KerrS2} and \eqref{eq:ReggeWheeler Phi^012KerrNegaS2} are at two level of regularity at most, involving no derivatives of $\phi^2_s$ and at most one of $\phi^0_s$ and $\phi^1_s$.

Note that systems \eqref{eq:ReggeWheeler Phi^012KerrS2} and \eqref{eq:ReggeWheeler Phi^012KerrNegaS2} are, however, not weakly coupled anymore as in the spin-$1$ case \cite{Ma17spin1Kerrlatest}, a fact caused by the presence of the $\phi^1_{s}$ term in \eqref{eq:ReggeWheeler Phi^0KerrS2} and \eqref{eq:ReggeWheeler Phi^0KerrNegaS2} or the $\phi^0_{s}$ term in \eqref{eq:ReggeWheeler Phi^1KerrS2} and \eqref{eq:ReggeWheeler Phi^1KerrNegaS2}. Take system \eqref{eq:ReggeWheeler Phi^012KerrS2} for $s=+2$ for example. Our approach here relies on an estimate bounding $\phi^1_{+2}$ from $\phi^2_{+2}$ by employing the differential relation \eqref{eq:DefOfphi012PosiSpinS2} between them, which enables us to treat the system in a rough (but accurate in the Schwarzschild case) sense that the error term in the Morawetz estimate for \eqref{eq:ReggeWheeler Phi^0KerrS2} arising from the inhomogeneous term  can be controlled by adding a large amount of Morawetz estimate of \eqref{eq:ReggeWheeler Phi^2KerrS2} to it, cf. Section \ref{sect:outlineproof}.
\end{remark}

\subsection{Previous Results}
We briefly review the existed results in the literature on scalar wave equation and Maxwell equations in the exterior of Schwarzschild and Kerr black holes. The uniform boundedness of scalar wave on a Schwarzschild background is first proved in \cite{kaywald87Schw}, and a robust, powerful tool--Morawetz estimate \cite{morawetz1968time}, or integrated local energy decay estimate--was used first in \cite{blue2003semilinear} and then in some followup works like \cite{bluesoffer09phase,dafrod09red}. For the scalar wave on the Kerr background, the decay results are shown in three different approaches \cite{tataru2011localkerr,larsblue15hidden,dafermos2016decay} where the first two are restricted to slowly rotating Kerr and the last one is for full subextremal Kerr.
Decay behaviours for Maxwell field are proved in \cite{blue08decayMaxSchw} on Schwarzschild, and on some spherically symmetric backgrounds or non-stationary asymptotically flat backgrounds in
\cite{metcalfe2014PWdecayMaxBH,sterbenz2015decayMaxSphSym}. These works start with estimating the middle component from a decoupled, separable Fackerell--Ipser equation \cite{fackerell:ipser:EM}. This approach is further generalized to the slowly rotating Kerr case in \cite{larsblue15Maxwellkerr} to show the uniform boundedness of a positive definite energy and the convergence property of the Maxwell field to a stationary Coulomb field. In contrast, one can treat first the extreme components satisfying the TME by applying some first-order differential operators used here to the extreme components, and the new quantitities satisfy an equation similar to the Fackerell--Ipser equation (or spin-weighted version of Fackerell--Ipser equation). This is carried out in \cite{Fede2016MaxwellSchw} for the Schwarzschild case and in our first part \cite{Ma17spin1Kerrlatest} of this series for the slowly rotating Kerr case.

There are typically two ways of linearizing the vacuum Einstein equations, one via metric perturbations and the other via tetrad perturbations. The linear stability of Schwarzschild metric under metric perturbations was resolved recently in  \cite{DRG16linearstabSchw,Hung2017linearstabSchw}. The former one starts from a Regge--Wheeler \cite{ReggeWheeler1957} type equation satisfied by some scalar constructed by applying a physical-space version of fixed-frequency \emph{Chandrasekhar
transformation} \cite{chandrasekhar1975linearstabSchw}  to some Riemann curvature components (closely related to extreme components in N--P formalism), and the later one carries out a detailed study on the Regge--Wheeler--Zerilli--Moncrief \cite{ReggeWheeler1957,Zerilli1970evenparity,moncrief74gravitational} system. The energy, Morawetz, and pointwise decay estimates for this system are obtained in \cite{Jinhua17LinGraSchw} as well.

On a non-static Kerr background, as mentioned already, if one linearizes the vacuum Einstein equations via tetrad perturbations, the extreme components of the Weyl tensor in Newman--Penrose formalism satisfy a decoupled, separable wave-like equation--the TME \eqref{eq:TME} and govern the dynamics of linearized gravity. These two extreme components are closely related to each other by differential relations: after decomposing spin $\pm2$ components into modes, differential relations between the radial parts of the modes with opposite extreme spin weights, as well as between the angular parts, are derived in \cite{starobinsky1973amplification,TeuPress1974III} and are known as \emph{Teukolsky-Starobinsky Identities}. See the version of these identities in the physical space in \cite{aksteiner2019new}. In \cite{whiting1989mode}, it is shown that the TME admits no mode with frequency having positive imaginary part, or in another way, no exponentially growing mode solution exists, under the assumption of  no incoming radiation condition. This mode stability result is recently generalized in \cite{2015AnHP...16..289S,andersson2017mode} to the case of real frequencies. Note that the mode stability result \cite{2015AnHP...16..289S} for scalar field is indispensable in the work \cite{dafermos2016decay} treating the scalar wave on the full subextremal Kerr backgrounds, and we expect our mode stability result \cite{andersson2017mode} for general spin fields will play an essential role in generalizing the results from the slowly rotating Kerr case considered in this work to the full subextremal Kerr case. Linear stability of slowly rotating Kerr spacetimes is proved in \cite{andersson2019stability,hafner2019linear}, and it is shown in \cite{andersson2019stability} that the linear stability of a subextremal Kerr spacetime can be proved under an assumption of a basic energy and Morawetz estimate in the same form of the main result of this work in the full subextremal Kerr backgrounds. During the submission of this work, the authors in \cite{DHR2019TME} obtain a similar result as that in this work. We notice also the works \cite{finster2016linear} which
discusses the stability problem for each azimuthal mode solution to TME and \cite{klainerman2017global} which proves a first nonlinear stability result for the Schwarzschild metric (though under axially symmetric polarized perturbations).

\subsection{Outline of the Proof}\label{sect:outlineproof}
It is convenient for the latter discussions to introduce the variables which are non-degenerate at $\mathcal{H}^+$
\begin{align}\label{def:widetildephi01KerrNegaS2}
\widetilde{\phi^0_{-2}}&=\Delta^{-2}r^4\phi^0_{-2},&
\widetilde{\phi^1_{-2}}&=\Delta^{-1}r^2\phi^1_{-2},
\end{align}
and we may suppress the subindex and simply write as $\widetilde{\phi^0}$ and $\widetilde{\phi^1}$. Let $\mu_0$, $\mu_1$, $\mu_2$, $\eps_0$, $\eps_1$, $\hat{\eps}_1$,  $A_0$, and $A_2$ be small positive constants to be fixed.
Define two quantities for spin $\pm2$ components respectively that
\begin{subequations}\label{eq:MasterInitialEnergyandError}
\begin{align}\label{eq:MasterInitialEnergyandErrorPosi}
\Xi_{+2}(0,\tau)={}&
{E}_{0}(r^{4-\delta}\phi^0_{+2})
+{E}_{0}(r^{2-\delta}\phi^1_{+2})+
{E}_{0}(\phi^2_{+2})\notag\\
&+{\veps_0^{1/2}}\sum_{\varphi\in
\{r^{4-\delta}\phi^0_{+2},r^{2-\delta}\phi^1_{+2},\phi^2_{+2}\}}
\left(
E_{\tau}(\varphi)
+E_{\mathcal{H}^+(0,\tau)}(\varphi)\right)\notag\\
&
+\veps_0^{1/2}\int_{\mathcal{D}(0,\tau)}
\left(\widetilde{\mathbb{M}}(r^{4-\delta}\phi^0_{+2})
+\widetilde{\mathbb{M}}(r^{2-\delta}\phi^1_{+2})
+\mathbb{M}_{\text{deg}}(\phi^2_{+2})\right),\\
\label{eq:MasterInitialEnergyandErrorNega}
\Xi_{-2}(0,\tau)={}&
{E}_{0}(\widetilde{\phi^0})
+{E}_{0}(\widetilde{\phi^1})+{E}_{0}(\phi^2_{-2})
+\int_{\Sigma_0}r\left(|\nablaslash \widetilde{\phi^0}|^2+|\nablaslash \widetilde{\phi^1}|^2\right)
\notag\\
&+\veps_0^{1/2}
\sum_{\varphi\in\{\widetilde{\phi^0},
\widetilde{\phi^1},\phi_{-2}^2\}}\left(E_{\tau}(\widetilde{\varphi})+
{E}_{\mathcal{H}^+(0,\tau)}(\varphi)\right)
\notag\\
&
+\veps_0^{1/2}\int_{\mathcal{D}(0,\tau)}
\left({\mathbb{M}}(\widetilde{\phi^0})
+{\mathbb{M}}(\widetilde{\phi^1})
+\mathbb{M}_{\text{deg}}(\phi^2_{-2})\right).
\end{align}
\end{subequations}
We say $F_1\lesssim_{\veps_0} F_2$ for two functions $F_1$ and $F_2$ in the region $\mathcal{D}(0,\tau)$ if there exists a universal constant $C$ and a constant $C_1=C_1(\mu_0,\mu_1,\mu_2,\eps_0,\eps_1,\hat{\eps}_1, A_0,A_2)$ \footnote{The dependence of $\hat{\eps}_1$  in $C_1$ is not needed for spin $-2$ component.} such that
\begin{subequations}\label{def:AlesssimaB}
\begin{align}\label{def:AlesssimaBposi}
F_1 \leq {}& CF_2+C_1\Xi_{+2}(0,\tau)
\end{align}
or
\begin{align}\label{def:AlesssimaBnega}
F_1 \leq{}& CF_2+C_1\Xi_{-2}(0,\tau),
\end{align}
\end{subequations}
depending on which spin component we are considering.
We now give the outline of the proof of estimates \eqref{eq:MoraEstiFinal(2)KerrRegularpsiBothSpinComp} for spin $+2$ and $-2$ components separately.

\subsubsection{Spin $+2$ Component}
We will first obtain in Section \ref{sect:pfMainthmSchwS2} and Section \ref{sect:finishpfS2Kerr} the following energy and Morawetz estimates for $\phi^0$, $\phi^1$ and $\phi^2$ defined from the spin $+2$ component:
\begin{subequations}\label{eq:estiphi02hatphi1kerrS2}
\begin{align}
\label{eq:estiphi0kerrS2}
\hspace{6ex}&\hspace{-6ex}{E}_{\tau}(r^{4-\delta}\phi^0)
+{E}_{\mathcal{H}^+(0,\tau)}(r^{4-\delta}\phi^0)
+\int_{\mathcal{D}(0,\tau)} \widetilde{\mathbb{M}}_{\text{deg}}(r^{4-\delta}\phi^0)\notag\\
\lesssim_{\veps_0}{}&
 \mu_0^{-1}\int_{\mathcal{D}(0,\tau)} \left(\eps_0 \widetilde{\mathbb{M}}(r^{4-\delta}\phi^0)
+\frac{\hat{\eps}_1}{\eps_0} \widetilde{\mathbb{M}}(r\phi^1)+\frac{1}{\eps_0\hat{\eps}_1} \mathbb{M}_{\text{deg}}(\phi^2)\right)\notag\\
&+\mu_0\int_{\mathcal{D}(0,\tau)}\left(
\eps_1 \widetilde{\mathbb{M}}(r^{2-\delta}\phi^1)+
\eps_1^{-1}
\widetilde{\mathbb{M}}_{\text{deg}}(r^{4-\delta}\phi^0)
+\eps_1^{-1}\mathbb{M}_{\text{deg}}(\phi^2)\right)\notag\\
&+\mu_0\int_{\mathcal{D}(0,\tau)}\left(
\widetilde{\mathbb{M}}_{\text{deg}}(r^{4-\delta}\phi^0)
+\widetilde{\mathbb{M}}_{\text{deg}}(r^{2-\delta}\phi^1)
+\mathbb{M}_{\text{deg}}(\phi^2)\right),\\
\label{eq:estiphi1kerrS2}
\hspace{6ex}&\hspace{-6ex}{E}_{\tau}(r^{2-\delta}\phi^1)
+{E}_{\mathcal{H}^+(0,\tau)}(r^{2-\delta}\phi^1)
+\int_{\mathcal{D}(0,\tau)} \widetilde{\mathbb{M}}_{\text{deg}}(r^{2-\delta}\phi^1)\notag\\
\lesssim_{\veps_0} {}&\mu_1^{-1}\int_{\mathcal{D}(0,\tau)} \Big(\eps_1 \widetilde{\mathbb{M}}(r^{2-\delta}\phi^1)+
\eps_1^{-1}
\widetilde{\mathbb{M}}_{\text{deg}}(r^{4-\delta}\phi^0)
+\eps_1^{-1}\mathbb{M}_{\text{deg}}(\phi^2)\Big)\notag\\
&+\mu_1\int_{\mathcal{D}(0,\tau)}\left(\eps_0 \widetilde{\mathbb{M}}(r^{4-\delta}\phi^0)
+\frac{\hat{\eps}_1}{\eps_0} \widetilde{\mathbb{M}}(r\phi^1)+\frac{1}{\eps_0\hat{\eps}_1} \mathbb{M}_{\text{deg}}(\phi^2)\right)\notag\\
&+\mu_1\int_{\mathcal{D}(0,\tau)}\left(
\widetilde{\mathbb{M}}_{\text{deg}}(r^{4-\delta}\phi^0)
+\widetilde{\mathbb{M}}_{\text{deg}}(r^{2-\delta}\phi^1)
+\mathbb{M}_{\text{deg}}(\phi^2)\right),\\
\label{eq:estiphi2kerrS2}
\hspace{6ex}&\hspace{-6ex}{E}_{\tau}(\phi^2)+{E}_{\mathcal{H}^+(0,\tau)}(\phi^2)
+\int_{\mathcal{D}(0,\tau)} \mathbb{M}_{\text{deg}}(\phi^2)\notag\\
\lesssim_{\veps_0} {}&
\mu_2\int_{\mathcal{D}(0,\tau)}\left(
\widetilde{\mathbb{M}}_{\text{deg}}(r^{4-\delta}\phi^0)
+\widetilde{\mathbb{M}}_{\text{deg}}(r^{2-\delta}\phi^1)
+\mathbb{M}_{\text{deg}}(\phi^2)\right)\notag\\
&+\mu_2\int_{\mathcal{D}(0,\tau)} \left(\eps_0 \widetilde{\mathbb{M}}(r^{4-\delta}\phi^0)
+\frac{\hat{\eps}_1}{\eps_0} \widetilde{\mathbb{M}}(r\phi^1)+\frac{1}{\eps_0\hat{\eps}_1} \mathbb{M}_{\text{deg}}(\phi^2)\right)\notag\\
&+\mu_2\int_{\mathcal{D}(0,\tau)} \Big(\eps_1 \widetilde{\mathbb{M}}(r^{2-\delta}\phi^1)+
\eps_1^{-1}
\widetilde{\mathbb{M}}_{\text{deg}}(r^{4-\delta}\phi^0)
+\eps_1^{-1}\mathbb{M}_{\text{deg}}(\phi^2)\Big).
\end{align}
\end{subequations}
Here, all the parameters appearing above are small constants to be fixed, and we have assumed that $\veps_0$ is much smaller compared to these parameters.
We add an $A_0$ multiple of estimate \eqref{eq:estiphi0kerrS2} and an $A_2$ multiple of \eqref{eq:estiphi2kerrS2} to estimate \eqref{eq:estiphi1kerrS2}, and find the left-hand side (LHS) of the obtained estimate is larger than
\begin{align}
\label{eq:finalpart:+2:LHS}
&c\inttau \left(A_0 \widetilde{\mathbb{M}}_{\text{deg}}(r^{4-\delta}\phi^0)
+\widetilde{\mathbb{M}}_{\text{deg}}(r^{2-\delta}\phi^1)
+A_2\mathbb{M}_{\text{deg}}(\phi^2)\right)\notag\\
&+c\inttau \left(
\widetilde{\mathbb{M}}(r^{4-\delta}\phi^0)
+\widetilde{\mathbb{M}}(r^{2-\delta}\phi^1)\right).
\end{align}
Here, the bound over the second line comes from the following relations
\begin{subequations}\label{eq:equivalentoftwobulktermsinMorawestiSubeqs}
\begin{align}
\label{eq:equivalentoftwobulktermsinMorawesti}
\hspace{4ex}&\hspace{-4ex}
{E}_{\tau}(r^{4-\delta}\phi^0)
+{E}_{0}(r^{4-\delta}\phi^0)
+\int_{\mathcal{D}(0,\tau)}
\left(\widetilde{\mathbb{M}}_{\text{deg}}(r^{4-\delta}\phi^0)
+\widetilde{\mathbb{M}}_{\text{deg}}(r^{2-\delta}\phi^1)\right)\notag\\
\sim {}&
{E}_{\tau}(r^{4-\delta}\phi^0)
+{E}_{0}(r^{4-\delta}\phi^0)
+\int_{\mathcal{D}(0,\tau)}\left(\widetilde{\mathbb{M}}
(r^{4-\delta}\phi^0)+
\widetilde{\mathbb{M}}_{\text{deg}}
(r^{2-\delta}\phi^1)\right),\\
\label{eq:equivalentoftwobulktermsinMorawesti20}
\hspace{4ex}&\hspace{-4ex}
{E}_{\tau}(r^{4-\delta}\phi^0)
+{E}_{\tau}(r^{2-\delta}\phi^1)
+{E}_{0}(r^{4-\delta}\phi^0)
+{E}_{0}(r^{2-\delta}\phi^1)
\notag\\
+\hspace{4ex}&\hspace{-4ex}\int_{\mathcal{D}(0,\tau)}\left(
\widetilde{\mathbb{M}}_{\text{deg}}(r^{4-\delta}\phi^0)
+\widetilde{\mathbb{M}}_{\text{deg}}(r^{2-\delta}\phi^1)
+\mathbb{M}_{\text{deg}}(\phi^2)\right)\notag\\
\sim &
{E}_{\tau}(r^{4-\delta}\phi^0)
+{E}_{\tau}(r^{2-\delta}\phi^1)
+{E}_{0}(r^{4-\delta}\phi^0)
+{E}_{0}(r^{2-\delta}\phi^1)
\notag\\
&+\int_{\mathcal{D}(0,\tau)}\left(
\widetilde{\mathbb{M}}(r^{4-\delta}\phi^0)
+\widetilde{\mathbb{M}}(r^{2-\delta}\phi^1)
+\mathbb{M}_{\text{deg}}(\phi^2)\right).
\end{align}
\end{subequations}
In the trapped region, $\widetilde{\mathbb{M}}_{\text{deg}}(r^{4-\delta}\phi^0)+
\widetilde{\mathbb{M}}_{\text{deg}}(r^{2-\delta}\phi^1)$ bounds over $|Y\phi^0|^2$, $|\partial_{r^*}\phi^0|^2$  and $|\phi^0|^2$,
and hence over $|\phi^0|^2$ and $|H\phi^0|^2$, $H=\partial_t+a/(r^2+a^2)\partial_{\phi}$ being a globally timelike vector field in the interior of $\mathcal{D}$ with
$-g(H,H)=\Delta\Sigma/(r^2+a^2)^2$.
Away from the horizon, the wave operator can be rewritten as a sum of $H^2$, a second order elliptic operator and up to first order operators. Let $\tilde{\chi}(r)$ be a smooth cutoff function which equals to $1$ in $[2.8M, 3.2M]$ and vanishes for $[r_+,2.7M]\cup[3.3M,+\infty)$. By multiplying the wave equation of $r^{4-\delta}\phi^0$  by $\tilde{\chi}(r)r^{4-\delta}\phi^0$, integrating over $\mathcal{D}(0,\tau)$ and applying integration by parts, one obtains an integral of $|\partial_{r^*}\phi^0|^2+|\nablaslash \phi^0|^2$ over $\mathcal{D}(0,\tau)\cap[2.8M,3.2M]$ is bounded by a constant times a sum of energy fluxes on $\Sigma_{\tau}$ and $\Sigma_0$, an integral of $|H\phi^0|^2$ over $\mathcal{D}(0,\tau)\cap[2.7M,3.3M]$ and an integral of  $\Re(\tilde{\chi}(r) \overline{\phi^0}L(\phi^0, H\phi^0, \partial_{r^*}\phi^0, \nablaslash \phi^0, \phi^1))$ over $\mathcal{D}(0,\tau)\cap[2.7M,3.3M]$, where $L$ is a linear operator in its all arguments and has coefficients depending only on $r, M$ and $a$.
Relation \eqref{eq:equivalentoftwobulktermsinMorawesti} then follows from applying the Cauchy--Schwarz inequality to the last integral.  Relation \eqref{eq:equivalentoftwobulktermsinMorawesti20} can be similarly justified.
We now
fix the parameters one by one to satisfy
\begin{align}
&\mu_1\ll1, \quad \eps_1\ll \mu_1, \quad A_0\gg \mu_1^{-1}\eps_1^{-1},\notag\\
&\mu_0\ll\min\{\eps_1,A_0^{-1}\},\quad
\eps_0\ll A_0^{-1}\mu_0, \quad \hat{\eps}_1\ll \min\{A_0^{-1}\mu_0\eps_0,\eps_0\mu_1^{-1}\},\notag\\
&A_2\gg (A_0\mu_0+\mu_1^{-1})\eps_1^{-1}
+(A_0\mu_0^{-1}+\mu_1)\eps_0^{-1}\hat{\eps}_1^{-1},\notag\\
&\mu_2\ll A_2^{-1}\min\{\eps_1A_0, {\eps_0}{\hat{\eps}_1^{-1}},\eps_1^{-1},\eps_0^{-1}, \eps_1 A_2, \eps_0\hat{\eps}_1A_2, 1\}.
\end{align}
Then for sufficiently small $\veps_0$, all the spacetime integrals on the right-hand side (RHS) of the gained estimate can be absorbed by \eqref{eq:finalpart:+2:LHS}, and the terms with $\sqrt{\veps_0}$-dependent coefficients in $\lesssim_{\veps_0}$ are absorbed by the LHS of the gained estimate. Therefore, we arrive at:
\begin{align}\label{eq:estiMoraphi012Kerrposicomp}
\hspace{4ex}&\hspace{-4ex}\sum_{\varphi\in \{
r^{4-\delta}\phi^0_{+2}, r^{2-\delta}\phi^1_{+2}, \phi^2_{+2}\}} \left(E_{\tau}(\varphi)
+E_{\mathcal{H}^+(0,\tau)}(\varphi)
\right)\notag\\
\hspace{4ex}&\hspace{-4ex}+\int_{\mathcal{D}(0,\tau)}
\left(\widetilde{\mathbb{M}}(r^{4-\delta}\phi^0_{+2})
+\widetilde{\mathbb{M}}(r^{2-\delta}\phi^1_{+2})
+\mathbb{M}_{\text{deg}}(\phi^2_{+2})\right)\notag\\
\lesssim{}&
\sum_{\varphi\in \{
r^{4-\delta}\phi^0_{+2}, r^{2-\delta}\phi^1_{+2}, \phi^2_{+2}\}} E_{0}(\varphi).
 \end{align}
Estimate \eqref{eq:MoraPosiSpinMain} with $n=0$ then follows from \eqref{eq:estiMoraphi012Kerrposicomp}.

\subsubsection{Spin $-2$ Component}

We will prove in Section \ref{sect:pfMainthmSchwS2} and Section \ref{sect:finishpfS2Kerr}
the following energy and Morawetz estimates for $\widetilde{\phi^0}$, $\widetilde{\phi^1}$ and $\phi^2$ constructed from the spin $-2$ component:
\begin{subequations}\label{eq:estiphi02hatphi1kerrS2Nega}
\begin{align}\label{eq:estiphi0kerrS2Nega}
\hspace{5ex}&\hspace{-5ex}{E}_{\tau}(\widetilde{\phi^0})+{E}_{\mathcal{H}^+(0,\tau)}(\widetilde{\phi^0})
+\int_{\mathcal{D}(0,\tau)} \mathbb{M}_{\text{deg}}(\widetilde{\phi^0})\notag\\
\lesssim_{\veps_0} {}&\int_{\mathcal{D}(0,\tau)} \left[\mu_0^{-1}\left(\eps_0{\mathbb{M}}(\widetilde{\phi^0})
+\eps_0^{-1}  \mathbb{M}_{\text{deg}}(\phi^2)\right)
+\mu_0\left(\eps_1 {\mathbb{M}}(\widetilde{\phi^1})+
\eps_1^{-1}
\mathbb{M}_{\text{deg}}(\phi^2)\right)\right]\notag\\
&+\mu_0\inttau \left(
{\mathbb{M}}_{\text{deg}}(\widetilde{\phi^0})
+{\mathbb{M}}_{\text{deg}}(\widetilde{\phi^1})
+\mathbb{M}_{\text{deg}}(\phi^2)\right),\\
\label{eq:estiphi1kerrS2Nega}
\hspace{5ex}&\hspace{-5ex}{E}_{\tau}(\widetilde{\phi^1})
+{E}_{\mathcal{H}^+(0,\tau)}(\widetilde{\phi^1})
+\int_{\mathcal{D}(0,\tau)} \mathbb{M}_{\text{deg}}(\widetilde{\phi^1})\notag\\
\lesssim_{\veps_0} {} &
\int_{\mathcal{D}(0,\tau)} \left[\mu_1^{-1}\left(\eps_1 {\mathbb{M}}(\widetilde{\phi^1})+
\eps_1^{-1}
\mathbb{M}_{\text{deg}}(\phi^2)\right)
+\mu_1\left(\eps_0{\mathbb{M}}(\widetilde{\phi^0})
+\eps_0^{-1}  \mathbb{M}_{\text{deg}}(\phi^2)\right)\right]\notag\\
&+\mu_1\inttau \left(
{\mathbb{M}}_{\text{deg}}(\widetilde{\phi^0})
+{\mathbb{M}}_{\text{deg}}(\widetilde{\phi^1})
+\mathbb{M}_{\text{deg}}(\phi^2)\right)
,\\
\label{eq:estiphi2kerrS2Nega}
\hspace{5ex}&\hspace{-5ex}{E}_{\tau}(\phi^2)+{E}_{\mathcal{H}^+(0,\tau)}(\phi^2)
+\int_{\mathcal{D}(0,\tau)} \mathbb{M}_{\text{deg}}(\phi^2)\notag\\
\lesssim_{\veps_0} {}&
\int_{\mathcal{D}(0,\tau)} \mu_2\left(\eps_0{\mathbb{M}}(\widetilde{\phi^0})
+\eps_0^{-1}  \mathbb{M}_{\text{deg}}(\phi^2)
+\eps_1 {\mathbb{M}}(\widetilde{\phi^1})+
\eps_1^{-1}
\mathbb{M}_{\text{deg}}(\phi^2)\right)\notag\\
&+\mu_2\inttau \left(
{\mathbb{M}}_{\text{deg}}(\widetilde{\phi^0})
+{\mathbb{M}}_{\text{deg}}(\widetilde{\phi^1})
+\mathbb{M}_{\text{deg}}(\phi^2)\right).
\end{align}
\end{subequations}
Similarly as the spin $+2$ case, the parameters in the above estimates are small constants to be fixed.
We add an $A_0$ multiple of \eqref{eq:estiphi0kerrS2Nega} and an $A_2$ multiple of \eqref{eq:estiphi2kerrS2Nega} to the estimate \eqref{eq:estiphi1kerrS2Nega} and find the LHS of the gained estimate bounds over
\begin{align}
\label{eq:LHSsummaryspin-2Kerr}
&\sum_{i=0,1}\left({E}_{\tau}(\widetilde{\phi^i})
+{E}_{\mathcal{H}^+(0,\tau)}(\widetilde{\phi^i})\right)
+\left({E}_{\tau}(\phi^2)
+{E}_{\mathcal{H}^+(0,\tau)}(\phi^2)\right)\notag\\
&+c\int_{\mathcal{D}(0,\tau)} \left(
A_0{\mathbb{M}}_{\text{deg}}(\widetilde{\phi^0})
+{\mathbb{M}}_{\text{deg}}(\widetilde{\phi^1})
+A_2\mathbb{M}_{\text{deg}}(\phi^2)
+{\mathbb{M}}(\widetilde{\phi^0})
+{\mathbb{M}}(\widetilde{\phi^1})
\right).
\end{align}
The reason that the terms ${\mathbb{M}}(\widetilde{\phi^0})
+{\mathbb{M}}(\widetilde{\phi^1})$ where trapping degeneracy is removed are present here is the same as showing the relations \eqref{eq:equivalentoftwobulktermsinMorawestiSubeqs} for the spin $+2$ component.
Fix the parameters in the following order
\begin{align}
&\mu_1\ll 1, \quad \eps_1\ll\mu_1,\quad A_0\gg \mu_1, \quad \mu_0\ll A_0^{-1},\notag\\
&\eps_0\ll A_0^{-1}\mu_0, \quad
A_2\gg A_0(\mu_0\eps_1^{-1}+\mu_0^{-1}\eps_0^{-1})
+\mu_1\eps_0^{-1}+\mu_1^{-1}\eps_1^{-1},\notag\\
&\mu_2\ll \min\{A_2^{-1}\eps_0^{-1},A_2^{-1}\eps_1^{-1}, \eps_0,\eps_1\}
\end{align}
such that the RHS of the obtained estimate can be absorbed by \eqref{eq:LHSsummaryspin-2Kerr}.
Hence it holds true for sufficiently small $\veps_0$ that
\begin{align}\label{eq:estiMoraphi012Kerrnegacomp}
\hspace{6ex}&\hspace{-6ex}\sum_{i=0,1}\left({E}_{\tau}(\widetilde{\phi^i})
+{E}_{\mathcal{H}^+(0,\tau)}(\widetilde{\phi^i})\right)
+\left({E}_{\tau}(\phi^2)
+{E}_{\mathcal{H}^+(0,\tau)}(\phi^2)\right)\notag\\
\hspace{6ex}&\hspace{-6ex}+\int_{\mathcal{D}(0,\tau)} \left({\mathbb{M}}(\widetilde{\phi^0})
+{\mathbb{M}}(\widetilde{\phi^1})+\mathbb{M}_{\text{deg}}(\phi^2)
\right)\notag\\
\lesssim_{\veps_0} {}&
{E}_{0}(\widetilde{\phi^0})+{E}_{0}(\widetilde{\phi^1})
+E_0\left(\phi^2\right)
\notag\\
\lesssim {}&
{E}_{0}(\widetilde{\phi^0})+{E}_{0}(\widetilde{\phi^1})
+E_0\left(\phi^2\right),
 \end{align}
where in the last step, we use Proposition \ref{prop:InitialEnergyControlsMore} to estimate the integral term $\int_{\Sigma_0}r(|\nablaslash \widetilde{\phi^0}|^2+|\nablaslash \widetilde{\phi^1}|^2)$ implicit in $\lesssim_{\veps_0}$ and take $\veps_0$ sufficiently small such that the $\sqrt{\veps_0}$-dependent terms are absorbed by the LHS of \eqref{eq:estiMoraphi012Kerrnegacomp}.
 From the estimate \eqref{eq:estiMoraphi012Kerrnegacomp},
the desired estimate \eqref{eq:MoraNegaSpinMain} is proved for the other regular N-P component $\widetilde{\Phi_4}$ in the case of $n=0$.

\subsection*{Overview of the Paper}
In Section \ref{sect:PrelimandNotations}, we give some preliminaries and introduce some further notations. Red-shift estimates near horizon and Morawetz estimates in large radius region for different quantities are proved in Section \ref{sect:redshiftMoralargerall}. Afterwards, we derive in Section \ref{sect:anEstiforphi1STinte} some \textit{a priori} estimates on any fixed subextremal Kerr background by considering the definitions \eqref{eq:DefOfphi012BothSpinS2} in the context of transport equations. The basic estimates \eqref{eq:estiphi02hatphi1kerrS2} and \eqref{eq:estiphi02hatphi1kerrS2Nega} are proved in Section \ref{sect:pfMainthmSchwS2} on Schwarzschild and in Section \ref{sect:finishpfS2Kerr} on a slowly rotating Kerr background.
These complete the proof of the estimates \eqref{eq:MoraEstiFinal(2)KerrRegularpsiBothSpinComp} based on the discussions in Section \ref{sect:outlineproof} for $n=0$ and Section \ref{sect:highorderS2} for $n\geq 1$.

\section{Preliminaries and Further Notation}\label{sect:PrelimandNotations}

\subsection{Well-posedness Theorem}\label{sect:LWPandGlobalExistenceLinearWaveSystem}
We refer to \cite[Section 2.1]{Ma17spin1Kerrlatest} for the well-posedness (WP) theorem for a general system of linear wave equations, which is cited from \cite[Chapter 3.2]{bar2007wave}. Similarly as the reduction in \cite[Section 2.1]{Ma17spin1Kerrlatest}, we can assume that the regular extreme N--P components are smooth and of compact support on the initial hypersurface $\Sigma_0$.

\subsection{Generic constants and general rules}
\label{sect:genericconsts}
Constants $C$ and $c$, depending only on $\veps_0$, $M$, $\delta$ and $\Sigma_0$, are always understood as large constants and small constants respectively, and may change value from term to term throughout this paper based on the algebraic rules: $C+C=C$, $CC=C$, $Cc\leq C$, etc. When there is no confusion, the dependence on $M$, $\veps_0$, $\delta$, and $\Sigma_0$ may always be suppressed. Once the constants $\veps_0$ and $0<\delta<1/2$ in Theorem \ref{thm:EneAndMorEstiExtremeCompsNoLossDecayVersion2} are chosen, these constants can be made to be only dependent on $M$.

For any two functions $F_1$ and $F_2$, $F_1\lesssim F_2$ means that there exists a constant $C$ such that $F_1\leq CF_2$ holds everywhere. $F_1\sim F_2$  indicates that $F_1\lesssim F_2$ and $F_2\lesssim F_1$,
and we say that $F_1$ \emph{is equivalent to} $F_2$.

The standard Laplacian on unit $2$-sphere is denoted as $\triangle_{\mathbb{S}^2}$, and the volume form $\di \sigma_{\mathbb{S}^2}$ on unit $2$-sphere is $\sin\theta \di \theta \di \phi^*$ or $\sin\theta \di \theta \di \phi$ depending on which coordinate system is used.

Some cutoff functions will be used in this paper. Denote $\chi_R(r)$ to be a smooth cutoff function utilized in Section \ref{sect:MorawetzLarger} which is $1$ for $r\geq R$ and vanishes identically for $r\leq R-M$, and $\chi_0(r)$ a smooth cutoff function which equals $1$ for $r\leq r_0$ and is identically zero for $r\geq r_1$; see Section \ref{sect:redshiftMoralargerall} for the choices of $r_0$ and $r_1$. The function $\chi$ is a smooth cutoff both to the future time and to the past time, which will be applied to the solution in the proof of Theorem \ref{prop:MoraEstiAlmostScalarWave}.

An overline or a bar will always denote the complex conjugate, $\Re(\cdot)$ denotes the real part, and \textquotedblleft{LHS\textquotedblright}  and \textquotedblleft{RHS\textquotedblright}  are short for \textquotedblleft{left-hand side(s)\textquotedblright} and \textquotedblleft{right-hand side(s)\textquotedblright}, respectively.

Throughout this paper, whenever we talk about \textquotedblleft{choosing some multiplier for some equation\textquotedblright}, it should always be understood as multiplying the equation by this chosen multiplier, taking the real part and finally integrating in the spacetime region $\mathcal{D}(0,\tau)$ (or $\mathcal{D}(\tau_1,\tau_2)$) in global Kerr coordinate system with respect to the measure $\Sigma \sin \theta \di t^*\di r\di \theta \di \phi^*$. If necessary, an integration by parts may be performed as well.

\section{Estimates near Horizon and in Large Radius Region}\label{sect:redshiftMoralargerall}

Morawetz estimates in large radius region and red-shift estimates near horizon for different quantities are proved in this section. We emphasize that all the $R_0$ in the estimates in this whole section can be \textit{a priori} different, so do all the $r_0$ and the $r_1$, but we will take the minimal $r_0$, the maximal $r_1$ and the maximal $R_0$ among them such that the estimates hold true uniformly with the constants $C$ independent on these parameters, and still denote them as $r_0$, $r_1$ and $R_0$.

For a complex scalar $\psi$ with spin weight $s$, we shall define  a (rescaled) spin-weighted wave operator
\begin{align}
\label{eq:SigmaTildeBoxpsi}
\Sigma\widetilde{\Box}_g\psi
\triangleq {}&\left\{\partial_r(\Delta\partial_r)
-\tfrac{\left((r^2+a^2)\partial_t+a\partial_{\phi}\right)^2}{\Delta}\right.\notag\\
&\left.\ \ +
\tfrac{1}{\sin \theta}\tfrac{\partial}{\partial \theta}\left(\sin \theta \tfrac{\partial }{\partial \theta}\right)+\Big(\tfrac{\partial_{\phi}+is\cos \theta}{\sin \theta} +a\sin \theta\partial_{t}\Big)^2\right\}\psi.
\end{align}
The operator $\Sigma\widetilde{\Box}_g$ is the same as the rescaled scalar wave operator $\Sigma \Box_g$ except that $(\tfrac{\partial_{\phi}+is\cos \theta}{\sin \theta} +a\sin \theta\partial_{t})^2$ is now in place of the operator $(\tfrac{\partial_{\phi}}{\sin \theta}+a\sin \theta\partial_{t})^2$ in the expansion of $\Sigma \Box_g$.

\subsection{Morawetz Estimate in Large Radius Region}\label{sect:MorawetzLarger}
Recall (see e.g. \cite{dafermos2011bdedness,luk2012vector}) that for any $0<\delta<1/2$, there exist constants $R_0\gg 4M$ and $C=C(\delta)$ such that for all $R\geq R_0$, one can choose a multiplier
\begin{equation}
X_w\bar{\psi}=-\tfrac{1}{\Sigma}\left(f(r)\partial_{r^*}
+\tfrac{1}{4}w(r)\right)\bar{\psi}
\end{equation}
with
\begin{subequations}
\begin{align}
f&=\chi_R(r)\cdot(1-r^{-\delta}),\\
w&=2\partial_{r^*}f+4\tfrac{1-2M/r}{r}f-2\delta\tfrac{1-2M/r}{r^{1+\delta}}f
\end{align}
\end{subequations}
for the inhomogeneous rescaled scalar wave equation
\begin{equation}\label{eq:InhomoRescaledScalarWaveEq}
\Sigma \Box_g\psi=G,
\end{equation}
and achieve the following Morawetz estimate in large $r$ region for any $\tau_2>\tau_1$:
\begin{align}\label{eq:MoraLargerScalarWaveKerr}
\hspace{4ex}&\hspace{-4ex}c\int_{\mathcal{D}(\tau_1,\tau_2)\cap\{r\geq R\}}
\left\{\frac{|\partial_{r^*}\psi|^2}{r^{1+\delta}}+\frac{|\partial_t \psi|^2}{r^{1+\delta}}+\frac{|\check{\nablaslash}\psi|^2}{r}+
\frac{|\psi|^2}{r^{3+\delta}}\right\}\notag\\
\leq{} &C\bigg(\check{E}_{\tau_1}^{R-M}(\psi)+\check{E}_{\tau_2}^{R-M}(\psi)
+\int_{\mathcal{D}(\tau_1,\tau_2)\cap\{R-M\leq r\leq R\}}\left(|\check{\partial} \psi|^2+|\psi|^2\right)\bigg)\notag\\
&+\int_{\mathcal{D}(\tau_1,\tau_2)}\Re\left(G\cdot X_w\bar{\psi}\right).
\end{align}
Here, $\chi_R(r)$ is defined as in Section \ref{sect:genericconsts},
$\check{\nablaslash}$ are the standard rotational angular derivatives on sphere $\mathbb{S}^2(t,r)$ as in \eqref{SpinWeightedAngularDerivaBasisOnSphere}, and
\begin{equation}
\check{E}_{\tau}^{R-M}(\psi)\sim \int_{\Sigma_{\tau}\cap\{r\geq R-M\}}|\check{\partial} \psi|^2= \int_{\Sigma_{\tau}\cap\{r\geq R-M\}}( |\partial_{t^*}\psi|^2
+|\partial_r\psi|^2
+|\check{\nablaslash}\psi|^2).
\end{equation}
Using this and defining
\begin{align}
{E}_{\tau}^{R-M}(\psi)\sim \int_{\Sigma_{\tau}\cap\{r\geq R-M\}}|{\partial} \psi|^2,
\end{align}
we show the following Morawetz estimate in large $r$ region.
\begin{prop}\label{prop:ImprovedMoraEstiLargerSWRWE1}
In a subextremal Kerr spacetime, for any fixed $0<\delta<\half$, and for any solution $\psi$ to \eqref{eq:RewrittenFormofISWWEphi0OpeForm} or  \eqref{eq:RewrittenFormofSWRWEOpeForm}, there exists a constant $R_0(M)$ and a universal constant $C$ such that for all $R\geq R_0$, the following estimate holds for any $\tau_2>\tau_1$:
\begin{align}\label{eq:ImprovedMoraEstiLargerSWRWE}
\int_{\mathcal{D}(\tau_1,\tau_2)\cap\{r\geq R\}}
\mathbb{M}(\psi)
\leq {}&
C{E}^{R-M}_{\tau_1}(\psi)
+C{E}_{\tau_2}^{R-M}(\psi)\notag\\
&
+C\int_{\mathcal{D}(\tau_1,\tau_2)\cap\{R-M\leq r\leq R\}}|\partial \psi|^2\notag\\
&+C\bigg|\int_{\mathcal{D}(\tau_1,\tau_2)\cap\{r\geq R-M\}}\Re\left(FX_w\bar{\psi}\right)\bigg|.
\end{align}
\end{prop}
\begin{remark}
Recall here the definition of the Morawetz density $\mathbb{M}(\psi)$ in \eqref{def:mathbbMpsiWholeSet}. This estimate will be applied to $\psi=\phi_s^i$ defined in \eqref{eq:DefOfphi012PosiSpinS2} and \eqref{eq:DefOfphi012NegaSpinS2} with the corresponding inhomogeneous term $F=F_s^i$ in \eqref{eq:ReggeWheeler Phi^012KerrS2} and \eqref{eq:ReggeWheeler Phi^012KerrNegaS2}.
\end{remark}

\begin{proof}
Using definition \eqref{eq:SigmaTildeBoxpsi}, the form \eqref{eq:RewrittenFormofSWRWEOpeForm}
reads
\begin{align}\label{eq:RewrittenFormofSWRWE}
\Sigma\widetilde{\Box}_g\psi
={}&
\left(4ias\cos\theta\partial_t+s^2\tfrac{\Delta+a^2}{r^2}\right)\psi+F,
\end{align}
and equation
\eqref{eq:RewrittenFormofISWWEphi0OpeForm} can be rewritten as
\begin{align}\label{eq:RewrittenFormofISWWEphi0}
\Sigma\widetilde{\Box}_g\psi
={}&
\left(4ias\cos\theta\partial_t+s^2\tfrac{r^2+2Mr-2a^2}{2r^2}\right)\psi+F.
\end{align}
Since the difference between the operator $(\tfrac{\partial_{\phi}+is\cos\theta}{\sin \theta}+a\sin \theta\partial_{t})^2$ in $\Sigma\widetilde{\Box}_g$ and $(\tfrac{\partial_{\phi}}{\sin \theta}+a\sin \theta\partial_{t})^2$ in the expansion of $\Sigma \Box_g$ has terms with coefficients independent of $r$, we achieve the same type of Morawetz estimate in large $r$ region by utilizing the same multiplier $X_w\bar{\psi}$, with $|\nablaslash\psi|^2-s^2|\psi|^2/r^2$ and $|\partial \psi|^2-s^2|\psi|^2/r^2$ in place of $|\check{\nablaslash}\psi|^2$ and $|\check{\partial} \psi|^2$, $E^{R-M}(\psi)$ replacing $\check{E}^{R-M}(\psi)$, and $G$ replaced by the the RHS of \eqref{eq:RewrittenFormofSWRWE} or \eqref{eq:RewrittenFormofISWWEphi0} in \eqref{eq:MoraLargerScalarWaveKerr}.

Consider equation \eqref{eq:RewrittenFormofSWRWE}. The bulk term coming from the source term
\begin{equation}\label{eq:sourcetermMoraLarger}
G=\left(4ias\cos\theta\partial_t+s^2\tfrac{\Delta+a^2}{r^2}\right)\psi+F
\end{equation}
is
\begin{align}
\label{eq:Errortermlarger:case1}
\int_{\mathcal{D}(\tau_1,\tau_2)}-\frac{1}{\Sigma}\Re\left(\left(\left(
4ias\cos\theta\partial_t+s^2\frac{\Delta+a^2}{r^2}\right)\psi+F\right)
\left(f\partial_{r^*}+\frac{1}{4}w\right)\bar{\psi}\right).
\end{align}
The term $\int_{\mathcal{D}(\tau_1,\tau_2)}
-\Re\left(\tfrac{\Delta+a^2}{r^2}\psi w\bar{\psi}\right)$ is bounded by $\int_{\mathcal{D}(\tau_1,\tau_2)\cap[R,\infty)}-\tfrac{c|\psi|^2}{r^3}$, and an integration by parts applied to the term $\int_{\mathcal{D}(\tau_1,\tau_2)}
\tfrac{-4}{\Sigma}\Re\left(\tfrac{\Delta+a^2}{r^2}\psi f\partial_{r^*}\bar{\psi}\right)$ eventually implies that the term \eqref{eq:Errortermlarger:case1}
is bounded for $R$ large enough by
\begin{align}
&-\int_{\mathcal{D}(\tau_1,\tau_2)\cap[R,\infty)}\frac{c|\psi|^2}{r^3}
+C\int_{\mathcal{D}(\tau_1,\tau_2)\cap[R-M,R]} |\partial\psi|^2\notag\\
&
+\int_{\mathcal{D}(\tau_1,\tau_2)\cap[R-M,\infty)}
\left\{\frac{Ca^2}{r^2}
|\partial\psi|^2
+\Re\left(FX_w\bar{\psi}\right)\right\}.
\end{align}
We can move this first term to the LHS and find from \eqref{def:nablaslashModuleSquare} that the resulting LHS bounds over
\begin{align}
\hspace{4ex}&\hspace{-4ex}\int_{\mathcal{D}(\tau_1,\tau_2)\cap\{r\geq R\}}
c\left\{\frac{|\partial_{r^*}\psi|^2}{r^{1+\delta}}+\frac{|\partial_t \psi|^2}{r^{1+\delta}}+\frac{|{\nablaslash}\psi|^2-s^2|\psi|^2/r^2}{r}+
\frac{|\psi|^2}{r^{3+\delta}}+\frac{|\psi|^2}{r^3}\right\}\notag\\
\geq{}& c\int_{\mathcal{D}(\tau_1,\tau_2)\cap\{r\geq R\}}
\left\{\frac{|\partial_{r^*}\psi|^2}{r^{1+\delta}}+\frac{|\partial_t \psi|^2}{r^{1+\delta}}+\frac{|{\nablaslash}\psi|^2}{r}+
\frac{|\psi|^2}{r^{3}}\right\}.
\end{align}
The spacetime integral of $\tfrac{Ca^2}{r^2}|\partial\psi|^2$ over $\mathcal{D}(\tau_1,\tau_2)\cap\{r\geq R\}$ is absorbed by this LHS for sufficiently large $R$. Therefore, the estimate \eqref{eq:ImprovedMoraEstiLargerSWRWE} holds true for \eqref{eq:RewrittenFormofSWRWE}, and hence \eqref{eq:RewrittenFormofSWRWEOpeForm}.

The proof of the Morawetz estimate \eqref{eq:ImprovedMoraEstiLargerSWRWE} in large radius region for \eqref{eq:RewrittenFormofISWWEphi0OpeForm} follows in the same way, and we omit it.
\qed
\end{proof}

In fact, an improved Morawetz estimate in the large radius region for spin $+2$ component can be proved. This can be seen by expanding the $\phi^1_{+2}$ term on the RHS of \eqref{eq:ReggeWheeler Phi^0KerrS2} into first order derivatives of $\phi^0_{+2}$ and observing that the sign of the coefficient of $\partial_t \phi^0_{+2}$ provides a damping effect for sufficiently large $r$.
\begin{prop}\label{prop:ImprovedMoraEstiLargerSWRWE}
Let $0<\delta <1/2$ be given. In a subextremal Kerr spacetime,  there exists a constant $R_0(M)$ and a universal constant $C$ such that for all $R\geq R_0$ and any $\tau_2>\tau_1$, the following estimates hold true for $\phi_{+2}^0$ and $\phi^1_{+2}$ respectively:
\begin{subequations}\label{eq:MoraEstirphi0larger}
\begin{align}\label{eq:MoraInftyr4minusdeltaphi0}
\hspace{4ex}&\hspace{-4ex}\int_{\Sigma_{\tau_2}\cap [R,+\infty)}\left|\partial\left(r^{4-\delta}\phi^{0}_{+2}\right)\right|^2
+\int_{\mathcal{D}(\tau_1,\tau_2)\cap [R,\infty)}r^{-1}\left|\partial \left(r^{4-\delta}\phi^{0}_{+2}\right)\right|^2\notag\\
\lesssim &
\int_{\Sigma_{\tau_2}\cap [R-M,R)}\left|\partial \left(r^{4-\delta}\phi^{0}_{+2}\right)\right|^2+\int_{\Sigma_{\tau_1}\cap [R-M,+\infty)}\left|\partial \left(r^{4-\delta}\phi^{0}_{+2}\right)\right|^2\notag\\
&+\int_{\mathcal{D}(\tau_1,\tau_2)\cap [R-M,R]}\left|\partial \left(r^{4-\delta}\phi^{0}_{+2}\right)\right|^2,\\
\label{eq:MoraInftyEstir2minusdeltaphi1}
\hspace{4ex}&\hspace{-4ex}\int_{\Sigma_{\tau_2}\cap [R,+\infty)}\left|\partial\left(r^{2-\delta}\phi^{1}_{+2}\right)\right|^2
+\int_{\mathcal{D}(\tau_1,\tau_2)\cap [R,\infty)}r^{-1}\left|\partial \left(r^{2-\delta}\phi^{1}_{+2}\right)\right|^2\notag\\
\lesssim &
\int_{\Sigma_{\tau_2}\cap [R-M,R)}\left|\partial \left(r^{2-\delta}\phi^{1}_{+2}\right)\right|^2
+\int_{\Sigma_{\tau_1}\cap [R-M,+\infty)}\left|\partial \left(r^{2-\delta}\phi^{1}_{+2}\right)\right|^2\notag\\
&
+\int_{\mathcal{D}(\tau_1,\tau_2)\cap [R-M,\infty)}\tfrac{\left|\partial \left(r^{4-\delta}\phi^{0}_{+2}\right)\right|^2}{r^2}
+\int_{\mathcal{D}(\tau_1,\tau_2)\cap [R-M,R]}\left|\partial \left(r^{2-\delta}\phi^{1}_{+2}\right)\right|^2.
\end{align}
\end{subequations}
\end{prop}
\begin{proof}
We define the variables
\begin{subequations}\label{eq:DefOfGoodPhi01PosiSpin}
\begin{align}
\phi^{0,4-\delta}_{+2}&=\Big(\tfrac{\R}{\sqrt{\Delta}}\Big)^{4-\delta}
\cdot\left(\psi_{[+2]}/(\R)^2\right),\\
\phi^{1,2-\delta}_{+2}&=\Big(\tfrac{\R}{\sqrt{\Delta}}\Big)^{2-\delta} \cdot\left(\sqrt{\R}Y\left(\psi_{[+2]}/(\R)^{3/2}\right)\right)
\end{align}
\end{subequations}
and derive the governing equations of them as follows
\begin{subequations}
\begin{align}
\label{eq:ImprEqForr4minusdeltaphi0}
\hspace{4ex}&\hspace{-4ex}\left(\Sigma \Box_g+\tfrac{4i\cos\theta}{\sin^2 \theta}\partial_{\phi}-4\cot^2 \theta+(2+\delta^2-5\delta)\right)\phi^{0,4-\delta}_{+2}\notag\\
={}&\tfrac{(r^3-3Mr^2+a^2r+a^2M)}{\R}
\Big(\tfrac{(4-2\delta)V(\sqrt{\R}\phi^{0,4-\delta}_{+2})}{\sqrt{\R}}
+2\delta\left(\tfrac{\R}{\Delta}\partial_t +\tfrac{a}{\Delta}\partial_{\phi}\right)\phi^{0,4-\delta}_{+2}\Big)\notag\\
&
+\left(4ia\cos \theta \partial_t
-\tfrac{8ar}{\R}\partial_{\phi}\right)\phi^{0,4-\delta}_{+2}
+\tfrac{P_5(r)}{\Delta(\R)^2}\phi^{0,4-\delta}_{+2},\\
\label{eq:ImprEqForr2minusdeltaphi1}
\hspace{4ex}&\hspace{-4ex}\left(\Sigma \Box_g+\tfrac{4i\cos\theta}{\sin^2 \theta}\partial_{\phi}-4\cot^2 \theta+(2+\delta^2-5\delta)\right)\phi^{1,2-\delta}_{+2}\notag\\
={}&\tfrac{(r^3-3Mr^2+a^2r+a^2M)}{\R}
\Big(\tfrac{(2-2\delta)V(\sqrt{\R}\phi^{1,2-\delta}_{+2})}{\sqrt{\R}}
+2\delta\left(\tfrac{\R}{\Delta}\partial_t +\tfrac{a}{\Delta}\partial_{\phi}\right)\phi^{1,2-\delta}_{+2}\Big)
\notag\\
&
+\tfrac{\underline{P}_5(r)}{\Delta(\R)^2}\phi^{1,2-\delta}_{+2}+\left(4ia\cos \theta \partial_t-\tfrac{4ar}{\R}\partial_{\phi}\right)\phi^{1,2-\delta}_{+2}
\notag\\
&+\tfrac{6a\Delta(a^2-r^2)}{(\R)^3}\partial_{\phi}\phi^{0,4-\delta}_{+2}
+\tfrac{6r\Delta(Mr^3-a^2r^2-3Ma^2r-a^4)}{(\R)^4}\phi^{0,4-\delta}_{+2}.
\end{align}
\end{subequations}
Here, $P_5(r)$ and $\underline{P}_5(r)$ are both polynomials in $r$ with powers no larger than $5$, and the coefficients of these two polynomials depend only on $a, M$, and $\delta$ and can be calculated explicitly. We shall make use of the following expansion for any smooth complex scalar $\psi$ of spin weight $s$
\begin{align}\label{eq:BoxInTermsOfYV}
\hspace{4ex}&\hspace{-4ex}\left(\Sigma \Box_g+\tfrac{2is\cos\theta}{\sin^2 \theta}\partial_{\phi}-s^2\cot^2 \theta+|s|+\delta^2 - 5\delta\right)\psi\notag\\
={}&\left(\tfrac{1}{\sin{\theta}} \partial_{\theta}(\sin \theta \partial_{\theta})+\tfrac{1}{\sin^2\theta}\partial_{\phi\phi}^2+\tfrac{2is\cos\theta}{\sin^2 \theta}\partial_{\phi}-s^2\cot^2 \theta+|s|+\delta^2 - 5\delta\right)\psi\notag\\
&-\sqrt{\R}Y\left(\tfrac{\Delta}{\R}V\left(\sqrt{\R}\psi\right)\right)
+\tfrac{2ar}{\R}\partial_{\phi}\psi\notag\\
&+\left(2a\partial_{t\phi}^2+a^2 \sin^2 \theta\partial_{tt}^2\right)\psi-\tfrac{2Mr^3+a^2r^2-4a^2Mr+a^4}{(\R)^2}\psi.
\end{align}
From Remark \ref{rem:EigenvalueSpinWeightedAngular}, the eigenvalues of the operator in the first line on the RHS of \eqref{eq:BoxInTermsOfYV} are not greater than $\delta^2 - 5\delta$ which is negative, hence if we choose the multiplier
\begin{align}
\hspace{4ex}&\hspace{-4ex}-\Sigma^{-1}\chi_R X_0\overline{\phi^{0,4-\delta}_{+2}}\notag\\
\triangleq{}&-\Sigma^{-1}\chi_R \tfrac{\Delta}{\R}\Big(\tfrac{(4-2\delta)V(\sqrt{\R}\overline{\phi^{0,4-\delta}_{+2}})}{\sqrt{\R}}
+2\delta\left(\tfrac{\R}{\Delta}\partial_t +\tfrac{a}{\Delta}\partial_{\phi}\right)\overline{\phi^{0,4-\delta}_{+2}}\Big)\notag\\
={}&-\Sigma^{-1}\chi_R \tfrac{\Delta}{\R}
\Big(\tfrac{(4-\delta)
V(\sqrt{\R}\overline{\phi^{0,4-\delta}_{+2}})
+\delta
Y(\sqrt{\R}\overline{\phi^{0,4-\delta}_{+2}})}{\sqrt{\R}}
\Big)
\end{align}
for \eqref{eq:ImprEqForr4minusdeltaphi0}, it then follows from integration by parts and Cauchy--Schwarz inequality that
\begin{align}\label{eq:MultplyEqForr4minusdeltaphi0ByTimelikeVF}
&\int_{\Sigma_{\tau_2}\cap [R,+\infty)}|\partial\phi^{0,4-\delta}_{+2}|^2
+\int_{\mathcal{D}(\tau_1,\tau_2)\cap [R,\infty)}r^{-1}\left(| X_0 \phi^{0,4-\delta}_{+2}|^2+|\nablaslash \phi^{0,4-\delta}_{+2}|^2\right)\notag\\
\lesssim &
\bigg(\int_{\Sigma_{\tau_2}\cap [R-M,R)}+\int_{\Sigma_{\tau_1}\cap \{r\geq R-M\}}\bigg)|\partial \phi^{0,4-\delta}_{+2}|^2
+\int_{\mathcal{D}(\tau_1,\tau_2)\cap \{r\geq R-M\}}\tfrac{|\partial\phi^{0,4-\delta}_{+2}|^2}{r^{2}}.
\end{align}
Moreover,  by choosing the multiplier $-\chi_R r^{-3}(1-2M/r)\overline{\phi^{0,4-\delta}_{+2}}$ for \eqref{eq:ImprEqForr4minusdeltaphi0}, we arrive at
\begin{align}\label{eq:MultplyEqForr4minusdeltaphi0Byself}
\hspace{4ex}&\hspace{-4ex}\int_{\mathcal{D}(\tau_1,\tau_2)\cap [R,\infty)}r^{-1}\left(|\partial_{r}\phi^{0,4-\delta}_{+2}|^2+|\nablaslash \phi^{0,4-\delta}_{+2}|^2\right)\notag\\
\lesssim {}&
\int_{\Sigma_{\tau_2}\cap [R-M,+\infty)}|\partial\phi^{0,4-\delta}_{+2}|^2
+\int_{\Sigma_{\tau_1}\cap [R-M,+\infty)}|\partial\phi^{0,4-\delta}_{+2}|^2
\notag\\
&
+\int_{\mathcal{D}(\tau_1,\tau_2)\cap [R-M,\infty)}\left(r^{-1}|\partial_{t^*} \phi^{0,4-\delta}_{+2}|^2+r^{-2}|\partial \phi^{0,4-\delta}_{+2}|^2\right).
\end{align}
Adding a sufficiently large multiple of \eqref{eq:MultplyEqForr4minusdeltaphi0ByTimelikeVF} to \eqref{eq:MultplyEqForr4minusdeltaphi0Byself} and taking $R$ sufficiently large, we conclude the inequality \eqref{eq:MoraInftyr4minusdeltaphi0}.

To show the estimate \eqref{eq:MoraInftyEstir2minusdeltaphi1}, we use the multipliers
$$-\tfrac{1}{\Sigma}\chi_R \tfrac{\Delta}{\R}\Big(\tfrac{(2-2\delta)V(\sqrt{\R}\overline{\phi^{1,2-\delta}_{+2}})}{\sqrt{\R}}
+2\delta\left(\tfrac{\R}{\Delta}\partial_t +\tfrac{a}{\Delta}\partial_{\phi}\right)\overline{\phi^{1,2-\delta}_{+2}}\Big)$$
and
$$-\chi_R r^{-3}(1-2M/r)\overline{\phi^{1,2-\delta}_{+2}},$$
follow the same way as proving the estimate \eqref{eq:MoraInftyr4minusdeltaphi0} for the equation \eqref{eq:ImprEqForr2minusdeltaphi1}, and use the Cauchy--Schwarz inequality to estimate the bulk integrals arising from the last line of \eqref{eq:ImprEqForr2minusdeltaphi1}.
\qed
\end{proof}

\subsection{Red-shift Estimate near $\mathcal{H}^+$}\label{sect:Redshift}
The following red-shift estimate near $\mathcal{H}^+$ for the inhomogeneous rescaled scalar wave equation \eqref{eq:InhomoRescaledScalarWaveEq} is taken from \cite[Section 5.2]{dafermos2011bdedness}.

\begin{lemma}
\label{lem:RedshiftInhomoScalarWaveKerr}
In a slowly rotating Kerr spacetime,  there exist constants $\tilde{\veps}_0>0$, $r_+\leq 2M<r_0(M,\tilde{\veps}_0)<r_1(M,\tilde{\veps}_0)<(1+\sqrt{2})M$, $c>0$ and $C>0$, two smooth real functions $y_1(r)$ and $y_2(r)$ on $[r_+,\infty)$ with $y_1(r)\to 1$ and $y_2(r)\to 0$ as $r\to r_+$, and a $\varphi_{\tau}$-invariant timelike vector field
\begin{equation}\label{def:NVectorField}
N=T+\chi_0(r)\left(y_1(r)Y+y_2(r)T\right)
\end{equation}
such that for all $|a|/M \leq  \tilde{\veps}_0$, by choosing a multiplier
\begin{equation}\label{def:NchiVF}
N_{\chi_0}\bar{\psi}=-\chi_0(r)\Sigma^{-1} N\bar{\psi},
\end{equation}
the following estimate holds for any solution $\psi$ to the inhomogeneous rescaled scalar wave equation \eqref{eq:InhomoRescaledScalarWaveEq} for  any $\tau_2>\tau_1$:
\begin{align}\label{eq:RedshiftInhomoScalarWaveKerr}
\hspace{4ex}&\hspace{-4ex}c\bigg(\int_{\Sigma_{\tau_2}\cap\{r\leq r_0\}}\left|\check{\partial}
\psi\right|^2+\check{E}_{\mathcal{H}^{+}(\tau_1,\tau_2)}(\psi)
+\int_{\mathcal{D}(\tau_1,\tau_2)\cap\{r\leq r_0\}}\left|\check{\partial}\psi\right|^2\bigg)
\notag\\
\leq{} &C\bigg(\int_{\Sigma_{\tau_1}\cap\{r\leq r_1\}}\left|\check{\partial}\psi\right|^2+
\int_{\mathcal{D}(\tau_1,\tau_2)\cap\{r_0\leq r\leq r_1\}}\left|\check{\partial}\psi\right|^2\bigg)\notag\\
&
+\int_{\mathcal{D}(\tau_1,\tau_2)\cap\{r\leq r_1\}}\Re\left(G \cdot N_{\chi_0}\bar{\psi}\right).
\end{align}
Here, the cutoff function $\chi_0(r)$ is defined as in Section \ref{sect:genericconsts}, and in ingoing E--F coordinates,
\begin{align}
\label{eq:eventhorizonfluxscalarwaveredshift}
\check{E}_{\mathcal{H}^{+}(\tau_1,\tau_2)}(\psi) = \int_{\mathcal{H}^+(\tau_1,\tau_2)}(|\partial_v\psi|^2
+|\check{\nablaslash}\psi|^2)r^2\di v\sin\theta \di \theta \di \tilde{\phi} .
\end{align}
\end{lemma}

We generalize it to a spin-weighted wave equation.

\begin{lemma}
\label{lem:Redshiftspinweightedwavegeneral}
Let $\psi$ be a complex scalar with non-zero spin weight $s$ satisfying a spin-weighted wave equation
\begin{align}
\Sigma \widetilde{\Box}_g \psi =G,
\label{eq:Generalspinweigtedwaveform}
\end{align}
with $g=g_{M,a}$ being a slowly rotating Kerr metric. Then there exist constants $\tilde{\veps}_0>0$, $r_+\leq 2M<r_0(M,\tilde{\veps}_0)<r_1(M,\tilde{\veps}_0)<(1+\sqrt{2})M$, $c({|s|})>0$ and $C(|s|)>0$
such that for all $|a|/M\leq \tilde{\veps}_0$ and any $\tau_2>\tau_1$,
\begin{align}\label{eq:RedShiftEstiSpinweightedwave}
\hspace{4ex}&\hspace{-4ex}c(|s|)\bigg(
E_{\mathcal{H}^{+}(\tau_1,\tau_2)}(\psi)
+\int_{\Sigma_{\tau_2}\cap\{r\leq r_0\}}\left|\partial\psi\right|^2 +\int_{\mathcal{D}(\tau_1,\tau_2)\cap\{r\leq r_0\}}\left|\partial\psi\right|^2\bigg)
\notag\\
\leq {}&
C(|s|)\bigg(\int_{\Sigma_{\tau_1}\cap\{r\leq r_1\}}\left|\partial\psi\right|^2
+\int_{\mathcal{D}(\tau_1,\tau_2)\cap\{r_0\leq r\leq r_1\}}\left|\partial\psi\right|^2\bigg)\notag\\
&
+\int_{\mathcal{D}(\tau_1,\tau_2)\cap\{r\leq r_1\}}\Re\left(G \cdot N_{\chi_0}\bar{\psi}\right),
\end{align}
where $N_{\chi_0}$ is chosen as in \eqref{def:NchiVF} in Lemma \ref{lem:RedshiftInhomoScalarWaveKerr}.
\end{lemma}

\begin{proof}
As is mentioned in the beginning of this section,
the operator $\Sigma\widetilde{\Box}_g$ is the same as the rescaled scalar wave operator $\Sigma \Box_g$ except for $(\tfrac{\partial_{\phi}+is\cos \theta}{\sin \theta} +a\sin \theta\partial_{t})^2$ in place of the operator $(\tfrac{\partial_{\phi}}{\sin \theta}+a\sin \theta\partial_{t})^2$ in the expansion of $\Sigma \Box_g$.
The difference between the operator $(\tfrac{\partial_{\phi}+is\cos\theta}{\sin \theta} +a\sin \theta\partial_{t})^2$ in $\Sigma \widetilde{\Box}_g$ and $(\tfrac{\partial_{\phi}}{\sin \theta}+a\sin \theta\partial_{t})^2$ in the expansion of $\Sigma \Box_g$ involves only terms with coefficients independent of $t$, $\phi$ and $r$, therefore using the multiplier $N_{\chi_0}\psi$ chosen in Lemma \ref{lem:RedshiftInhomoScalarWaveKerr}
yields the same red-shift estimate as in \eqref{eq:RedshiftInhomoScalarWaveKerr}
but with $|\nablaslash\psi|^2-s^2|\psi|^2/r^2$ in place of $|\check{\partial} \psi|^2$ in the horizon flux \eqref{eq:eventhorizonfluxscalarwaveredshift} and $|\partial \psi|^2-s^2|\psi|^2/r^2$ in place of  $|\check{\nablaslash}\psi|^2$. Since from the lower bound \eqref{eq:LowerBdfornablaslash} of $\Vert r\nablaslash\psi\Vert_{L^2(S^2(r))}$, the sphere integrals of $|\nablaslash\psi|^2-s^2|\psi|^2/r^2$ and $|\partial \psi|^2-s^2|\psi|^2/r^2$ are bounded below by sphere integrals of $c({|s|})(|\nablaslash\psi|^2+|\psi|^2/r^2)$ and $c({|s|})(|\partial \psi|^2 + |\psi|^2/r^2)$, respectively, the estimate \eqref{eq:RedShiftEstiSpinweightedwave} follows.
\qed
\end{proof}

The above red-shift estimate for a general spin-weighted wave equation \eqref{eq:Generalspinweigtedwaveform} is used here to obtain red-shift estimates for different quantities.
\begin{prop}\label{prop:RedShiftEstiInhomoSWRWE}
In a slowly rotating Kerr spacetime , there exist constants $\tilde{\veps}_0>0$, $r_+\leq 2M<r_0(M,\tilde{\veps}_0)<r_1(M,\tilde{\veps}_0)<(1+\sqrt{2})M$ and $C$
such that for all $|a|/M\leq  \tilde{\veps}_0$ and any $\tau_2>\tau_1\geq 0$,
\begin{itemize}
  \item for  $\psi\in \{\phi^1_{+2},\phi^2_{+2},\phi^2_{-2}\}$ whose governing equations \eqref{eq:ReggeWheeler Phi^1KerrS2}, \eqref{eq:ReggeWheeler Phi^2KerrS2}
and \eqref{eq:ReggeWheeler Phi^2KerrNegaS2} can be put into the form of \eqref{eq:RewrittenFormofSWRWEOpeForm} with the relevant inhomogeneous term $F$, the following estimate holds:
\begin{align}\label{eq:RedShiftEstiInhomoSWRWE}
\hspace{4ex}&\hspace{-4ex}E_{\mathcal{H}^{+}(\tau_1,\tau_2)}(\psi)
+\int_{\Sigma_{\tau_2}\cap\{r\leq r_0\}}\left|\partial\psi\right|^2 +\int_{\mathcal{D}(\tau_1,\tau_2)\cap\{r\leq r_0\}}\left|\partial\psi\right|^2
\notag\\
\leq{} &
C\int_{\Sigma_{\tau_1}\cap\{r\leq r_1\}}\left|\partial\psi\right|^2
+C\int_{\mathcal{D}(\tau_1,\tau_2)\cap\{r_0\leq r\leq r_1\}}\left|\partial\psi\right|^2\notag\\
&
+C\int_{\mathcal{D}(\tau_1,\tau_2)\cap[r_+,r_1]}|F|^2;
\end{align}
  \item for the equation \eqref{eq:ReggeWheeler Phi^0KerrS2} of $\phi^0_{+2}$, the following estimate near horizon holds:
\begin{align}\label{eq:redshiftforphi0posiKerrS2}
\hspace{4ex}&\hspace{-4ex}E_{\mathcal{H}^{+}(\tau_1,\tau_2)}(\phi^0_{+2})
+\int_{\Sigma_{\tau_2}\cap\{r\leq r_0\}}|\partial\phi^0_{+2}|^2 +\int_{\mathcal{D}(\tau_1,\tau_2)\cap\{r\leq r_0\}}|\partial\phi^0_{+2}|^2
\notag\\
\leq {}&
C\int_{\Sigma_{\tau_1}\cap\{r\leq r_1\}}|\partial\phi^0_{+2}|^2
+C\int_{\mathcal{D}(\tau_1,\tau_2)\cap\{r_0\leq r\leq r_1\}}|\partial\phi^0_{+2}|^2\notag\\
&
+C\int_{\mathcal{D}(\tau_1,\tau_2)\cap[r_+,r_1]}|\phi^1_{+2}|^2.
\end{align}
\end{itemize}
\end{prop}
\begin{proof}

As in the previous subsection, we refer to the rewritten form \eqref{eq:RewrittenFormofSWRWE} of the equation \eqref{eq:RewrittenFormofSWRWEOpeForm}, which can in turn be put into the form of the equation \eqref{eq:Generalspinweigtedwaveform}.
We apply the conclusion in Lemma \ref{lem:Redshiftspinweightedwavegeneral} and obtain an estimate as \eqref{eq:RedShiftEstiSpinweightedwave} with the last line now replaced by
\begin{align}
\int_{\mathcal{D}(\tau_1,\tau_2)\cap\{r\leq r_1\}}
\Re\left(\left(\left(4ias\cos\theta\partial_t
+\frac{s^2(\Delta+a^2)}{r^2}\right)\psi+F\right)
N_{\chi_0}\bar{\psi}\right).
\end{align}
After applying integration by parts and Cauchy--Schwarz inequality, this bulk term is bounded by
\begin{align}
&\int_{\mathcal{D}(\tau_1,\tau_2)\cap\{r\leq r_1\}}-N\left(\chi_0\frac{s^2(\Delta+a^2)}{r^2}|\psi|^2\right)
-\int_{\mathcal{D}(\tau_1,\tau_2)\cap\{r\leq r_0\}}\frac{s^2(r_+-r_-)}{2r^2}|\psi|^2\notag\\
&+C\int_{\mathcal{D}(\tau_1,\tau_2)\cap\{r\leq r_1\}}\left(|a||\partial \psi|^2+|F|^2\right)+C\int_{\mathcal{D}(\tau_1,\tau_2)\cap\{r_0\leq r\leq r_1\}}|\partial \psi|^2.
\end{align}
By moving the first two integrals above to the LHS, the first integral produces extra positive energy fluxes and the second one adds additionally positive spacetime integrals of $|\psi|^2$. The spacetime integral term $|a||\partial\psi|^2$ in the third term is absorbed by the LHS by choosing $\tilde{\veps}_0$ sufficiently small. Altogether, we conclude the estimate \eqref{eq:RedShiftEstiInhomoSWRWE}.

For $\phi^0_{+2}$, we write equation \eqref{eq:ReggeWheeler Phi^0KerrS2} into the form of \eqref{eq:Generalspinweigtedwaveform} as
\begin{align}
\label{eq:phi+20tospinweightedwavegeneral}
\Sigma\widetilde{\Box}_g(\phi^0_{+2})
={}G(\phi^0_{+2})={}&
\tfrac{2(r^2+2Mr-2a^2)}{r^2}\phi^0_{+2}
+\tfrac{4(r^2-3Mr+2a^2)}{r^3}\phi^1_{+2}\notag\\
&
+8ia\cos\theta\partial_t \phi^0_{+2}
-\tfrac{8(a^2\partial_t+a\partial_{\phi})\phi^0_{+2}}{r}.
\end{align}
Using the definition of $\phi^1_{+2}$, one can eliminate the $\phi^0_{+2}$ term and introduce extra $Y\phi^0_{+2}$ term with a positive coefficient:
\begin{align}\label{eq:phi+20toRS}
\Sigma\widetilde{\Box}_g(\phi^0_{+2})
={}G(\phi^0_{+2})={}&
\tfrac{2(r^2+2Mr-2a^2)}{r}Y\phi^0_{+2}
+\tfrac{2r^2-16Mr+12a^2}{r^3}\phi^1_{+2}\notag\\
&-\tfrac{8(a^2\partial_t +a \partial_{\phi})\phi^0_{+2}}{r}
+8ia\cos\theta\partial_t(\phi^0_{+2}).
\end{align}
The estimate \eqref{eq:RedShiftEstiSpinweightedwave} can be applied again, and we need to estimate the bulk term
\begin{align}
\int_{\mathcal{D}(\tau_1,\tau_2)\cap\{r\leq r_1\}}
\Re\left(G(\phi^0_{+2})\cdot
N_{\chi_0}\overline{\phi^0_{+2}}\right).
\end{align}
This integral over $r_0\leq r\leq r_1$ is bounded by the RHS of \eqref{eq:redshiftforphi0posiKerrS2} from Cauchy--Schwarz inquality, and the left integral over $r\leq r_0$ is
\begin{align}
\label{eq:phi+20redshifterror4}
&-\int_{\mathcal{D}(\tau_1,\tau_2)\cap\{r\leq r_0\}}
\frac{1}{\Sigma}\Re\left(G(\phi^0_{+2})
(1+y_2(r))T\overline{\phi^0_{+2}}\right)\notag\\
&-\int_{\mathcal{D}(\tau_1,\tau_2)\cap\{r\leq r_0\}}\frac{1}{\Sigma}\Re\left(G(\phi^0_{+2})
y_1(r)Y\overline{\phi^0_{+2}}\right).
\end{align}
Substituting the expression \eqref{eq:phi+20tospinweightedwavegeneral} of $G(\phi^0_{+2})$ into the first term of \eqref{eq:phi+20redshifterror4}, the integral term from the term $\tfrac{2(r^2+2Mr-2a^2)}{r^2}\phi^0_{+2}$ has a good sign after integrating over the spacetime region, and the rest terms can be either estimated by Cauchy--Schwarz inequality or absorbed by setting $\tilde{\veps}_0$ sufficiently small. For the second term of \eqref{eq:phi+20redshifterror4}, we instead use the expansion of $G(\phi^0_{+2})$ on the RHS of \eqref{eq:phi+20toRS}. The first term on the RHS of \eqref{eq:phi+20toRS} contributes positive bulk terms after moving to the LHS, and we can similarly estimate the remaining integral terms by either
requiring $\tilde{\veps}_0$ sufficiently small or an application of the Cauchy--Schwarz inequality. These together prove the estimate \eqref{eq:redshiftforphi0posiKerrS2}.
\qed
\end{proof}

Notice that the quantities $\phi^0_{+2}$ and $\phi^1_{+2}$ are, however, degenerate at horizon, hence we shall prove as well the red-shift estimates for $\widetilde{\phi^0}$ and $\widetilde{\phi^1}$ (which are non-degenerate at $\mathcal{H}^+$) defined as in \eqref{def:widetildephi01KerrNegaS2}.
\begin{prop}\label{prop:RedShiftEstiInhomoSWRWEtildephi01}
In a slowly rotating Kerr spacetime $(\mathcal{M},g_{M,a})$, there exist constants $\tilde{\veps}_0>0$, $r_+<2M<r_0(M,\tilde{\veps}_0)<r_1(M,\tilde{\veps}_0)<(1+\sqrt{2})M$ and $C$
such that for all $|a|/M\leq \tilde{\veps}_0$, the following red-shift estimates hold for $\widetilde{\phi^0_{-2}}$ and $\widetilde{\phi^1_{-2}}$ for any $\tau_2>\tau_1$:
\begin{align}\label{eq:RedShiftEstipsi0InhomoSWRWENega}
\hspace{4ex}&\hspace{-4ex}E_{\mathcal{H}^{+}(\tau_1,\tau_2)}(\widetilde{\phi^0})
+\int_{\Sigma_{\tau_2}\cap\{r\leq r_0\}}|\partial\widetilde{\phi^0}|^2 +\int_{\mathcal{D}(\tau_1,\tau_2)\cap\{r\leq r_0\}}|\partial\widetilde{\phi^0}|^2
\notag\\
\leq {}&
C\int_{\Sigma_{\tau_1}\cap\{r\leq r_1\}}|\partial\widetilde{\phi^0}|^2
+C\int_{\mathcal{D}(\tau_1,\tau_2)\cap\{r_0\leq r\leq r_1\}}|\partial\widetilde{\phi^0}|^2\notag\\
&
+C\int_{\mathcal{D}(\tau_1,\tau_2)\cap[r_+,r_1]}|\widetilde{\phi^1}|^2,\\
\label{eq:RedShiftEstipsi1InhomoSWRWENega}
\hspace{4ex}&\hspace{-4ex}E_{\mathcal{H}^{+}(\tau_1,\tau_2)}(\widetilde{\phi^1})
+\int_{\Sigma_{\tau_2}\cap\{r\leq r_0\}}|\partial\widetilde{\phi^1}|^2 +\int_{\mathcal{D}(\tau_1,\tau_2)\cap\{r\leq r_0\}}|\partial\widetilde{\phi^1}|^2
\notag\\
\leq{} &
C\int_{\Sigma_{\tau_1}\cap\{r\leq r_1\}}|\partial\widetilde{\phi^1}|^2
+C\int_{\mathcal{D}(\tau_1,\tau_2)\cap\{r_0\leq r\leq r_1\}}|\partial\widetilde{\phi^1}|^2\notag\\
&
+C\int_{\mathcal{D}(\tau_1,\tau_2)\cap[r_+,r_1]}
\left(|\phi^2_{-2}|^2+\tfrac{ |a|}{M}|\partial\widetilde{\phi^0}|^2
+|\widetilde{\phi^0}|^2\right).
\end{align}
\end{prop}

\begin{proof}

The equation for $\widetilde{\phi^0}$ reads
\begin{align}\label{eq:eqPsi[-2]}
\Sigma\widetilde{\Box}_g\widetilde{\phi^0}
=&\left(\tfrac{4(r-M)r-5\Delta}{r}Y+2r\partial_t\right)
\widetilde{\phi^0}
+8\widetilde{\phi^0}
-\tfrac{10a^2}{r^2}\widetilde{\phi^0}
-\tfrac{5\Delta}{r^2}\widetilde{\phi^0}\notag\\
&+\tfrac{10}{r}\left(a^2\partial_t
+a\partial_{\phi}\right)\widetilde{\phi^0}
+\tfrac{5}{r}\widetilde{\phi^1}
-8ia\cos\theta\partial_t\widetilde{\phi^0},
\end{align}
and the governing equation for $r^2\widetilde{\phi^1}$
is
\begin{align}\label{eq:r2tildephi1toRS}
\Sigma\widetilde{\Box}_g(r^2\widetilde{\phi^1})
=&\left(\tfrac{4(r-M)r-9\Delta}{2r}Y+r\partial_t\right)
(r^2\widetilde{\phi^1})
+\tfrac{7}{2}(r^2\widetilde{\phi^1})
-\tfrac{3a^2}{2r^2}(r^2\widetilde{\phi^1})
+\tfrac{r}{2}\phi^2\notag\\
&+\tfrac{6\Delta}{r}\left((Mr-2a^2)\widetilde{\phi^0}
+r\left(a^2\partial_t+a\partial_{\phi}\right)\widetilde{\phi^0}\right)\notag\\
&
+\tfrac{5}{r}\left(a^2\partial_t+a\partial_{\phi}\right)
(r^2\widetilde{\phi^1})
-8ia\cos\theta\partial_t(r^2\widetilde{\phi^1}).
\end{align}
Lemma \ref{lem:Redshiftspinweightedwavegeneral} can be applied to both equations. Near $r_+$, we find for the first term on the RHS of both equations that
\begin{align}
\left(\tfrac{4(r-M)r-5\Delta}{r}Y+2r\partial_t\right)
\widetilde{\phi^0}={}&
2r(Y+\partial_t)\widetilde{\phi^0}
+O(r-r_+)Y\widetilde{\phi^0},\\
\left(\tfrac{4(r-M)r-9\Delta}{2r}Y+r\partial_t\right)
(r^2\widetilde{\phi^1})={}&
r(Y+\partial_t)(r^2\widetilde{\phi^1}) + O(r-r_+)Y(r^2\widetilde{\phi^1}),
\end{align}
and the used operator
\begin{align}
\Sigma N_{\chi_0}={}-(T+Y)-((y_1(r)-1)Y+y_2(r)T),
\end{align}
where the coefficients $y_1(r)-1$ and $y_2(r)$ both converges to $0$ as $r\to r_+$. Therefore, the bulk integral over $\mathcal{D}(\tau_1,\tau_2)\cap\{r\leq r_0\}$ from the the first term on the RHS of both equations  can be absorbed by the LHS by choosing $r_0$ close to $r_+$ and taking $\tilde{\veps}_0$ small. The other bulk terms are estimated by either integration by parts or Cauchy--Schwarz inquality and requiring $\tilde{\veps}_0$ sufficiently small, and this completes the proof.
\qed
\end{proof}

\begin{remark}\label{rem:r0r1choice}
We fix $\tilde{\veps}_0$ in this subsection by taking the minimal value among all the $\tilde{\veps}_0$ in the above discussions. Given that $\tilde{\veps}_0$ is fixed, as have been discussed at the beginning of Section \ref{sect:redshiftMoralargerall}, the choices of $r_0=r_0(M)$ and $r_1=r_1(M)$ are made by taking the minimal $r_0$ and the maximal $r_1$ in all the above estimates.
\end{remark}

\section{Estimates for Spacetime Integrals of $\phi^0_s$ and $\phi^1_s$}\label{sect:anEstiforphi1STinte}
We derive in this section some a priori estimates for $\phi^0_s$ and $\phi^1_{s}$ which are used in Section \ref{sect:outlineproof}.

\subsection{Spin $+2$ Component}
\begin{prop}\label{prop:estiphi1posi}
Let $\hat{\eps}_1>0$ be arbitrary. In a subextremal Kerr spacetime, the following estimate holds for $\phi^1_{+2}$ defined as in \eqref{eq:DefOfphi012PosiSpinS2}:
\begin{align}\label{eq:estiphi1byphi02final1}
&\inttau \frac{|\phi^1|^2}{r^2}
\lesssim
\hat{\eps}_1 \inttau \frac{|r\phi^1|^2}{r^3}+\hat{\eps}_1^{-1} \inttau \frac{|\phi^2|^2}{r^3}+\int_{\Sigma_{0}}
\frac{|r\phi^1|^2}{r^2}.
\end{align}
\end{prop}
\begin{proof}
We start with a simple identity for any real smooth function $f_{+2}(r)$ and any real value $\alpha$:
\begin{align}\label{eq:estiphi1posiiden1}
Y\left(f_{+2}r^{\alpha}|r\phi^1|^2\right)
+f_{+2}\alpha r^{\alpha-1}|r\phi^1|^2-Y(f_{+2})r^{\alpha}|r\phi^1|^2
={}f_{+2}r^{\alpha}\Re(\phi^1\overline{\phi^2}).
\end{align}
Integrate \eqref{eq:estiphi1posiiden1} over $\mathcal{D}(0,\tau)$ with the measure
\begin{align}\label{eq:measureinestiforphi1}
\di \check{V}=r^{-2}\di V=\di t^*\di r\sin\theta \di \theta \di \phi^*
\end{align}
for $\alpha=0$ and $f_{+2}=\tfrac{\Delta}{\R}$. Then, since
\begin{equation}
-Y(f_{+2})=\partial_r f_{+2}=\frac{2M(r^2-a^2)}{(\R)^2}\geq\frac{c}{r^2},
\end{equation}
an application of the Cauchy--Schwarz inequality to $\inttau f_{+2}\Re(\phi^1\overline{\phi^2}) \di \check{V}$ proves the estimate \eqref{eq:estiphi1byphi02final1}.
\qed
\end{proof}

\subsection{Spin $-2$ Component}
\begin{prop}\label{prop:estiphi1nega}
In a subextremal Kerr spacetime, it holds true for $\phi^0_{-2}$ and $\phi^1_{-2}$ defined as in \eqref{eq:DefOfphi012NegaSpinS2} that
\begin{subequations}
\begin{align}\label{eq:estiphi0byphi02negafinal1}
&\inttau r^{-2}{|\widetilde{\phi^0}|^2}
\lesssim
\inttau r^{-3}{|\phi^2|^2}
+\int_{\Sigma_0}r^{-1}\Big({|\widetilde{\phi^0}|^2}
+|\widetilde{\phi^1}|^2\Big),\\
\label{eq:estiphi1byphi02negafinal1}
&\inttau r^{-2}{|\widetilde{\phi^1}|^2}
\lesssim
\inttau r^{-3}|\phi^2|^2
+\int_{\Sigma_0}r^{-1}|\widetilde{\phi^1}|^2.
\end{align}
\end{subequations}
Moreover, for the angular derivatives, we have
\begin{subequations}\label{eq:estiphiibyphii+1negafinal1}
\begin{align}
\label{eq:estiphi0byphi1angunegafinal1}
\hspace{4ex}&\hspace{-4ex}\intMinfty |\nablaslash \widetilde{\phi^0}|^2 +\int_{\Sigma_{\tau}\cap [6M,\infty)}r|\nablaslash \widetilde{\phi^0}|^2\notag\\
\lesssim {}&
\int_{\mathcal{D}(0,\tau)\cap [5M,\infty)} \frac{|\nablaslash \widetilde{\phi^{1}}|^2}{r}
+\int_{\Sigma_0\cap [5M, \infty)}r|\nablaslash \widetilde{\phi^0}|^2
+\intMcut \frac{|\nablaslash \widetilde{\phi^0}|^2}{r},\\
\label{eq:estiphi1byphi2angunegafinal1}
\hspace{4ex}&\hspace{-4ex}\intMinfty |\nablaslash \widetilde{\phi^1}|^2 +\int_{\Sigma_{\tau}\cap [6M,\infty)}r|\nablaslash \widetilde{\phi^1}|^2\notag\\
\lesssim {}&
\int_{\mathcal{D}(0,\tau)\cap [5M,\infty)} \frac{|\nablaslash {\phi^{2}}|^2}{r}
+\int_{\Sigma_0\cap [5M, \infty)}r|\nablaslash \widetilde{\phi^1}|^2
+\intMcut \frac{|\nablaslash \widetilde{\phi^1}|^2}{r} .
\end{align}
\end{subequations}
\end{prop}
\begin{proof}
We derive for any real function $f_{-2}(r)$ and
any real value $\beta$ that
\begin{align}\label{eq:estiphi1negaiden1}
&V(f_{-2}r^{\beta}|r\phi^1|^2)
-f_{-2}\beta r^{\beta-1}|r\phi^1|^2-\partial_r f_{-2} r^{\beta}|r\phi^1|^2
={}-r^{\beta}f_{-2}\Re(\phi^1\overline{\phi^2}).
\end{align}
By choosing $\beta=-1$ and $f_{-2}=\tfrac{\R}{\Delta}$, since $\partial_r f_{-2}=\tfrac{-2M(r^2-a^2)}{\Delta^2}$, the estimate \eqref{eq:estiphi1byphi02negafinal1} follows from integrating \eqref{eq:estiphi1negaiden1} over $\mathcal{D}(0,\tau)$ with the measure $\di \check{V}$ in
\eqref{eq:measureinestiforphi1}
and applying Cauchy--Schwarz to the integral of the RHS of \eqref{eq:estiphi1negaiden1}.

Similarly, for $\phi^0$, we have
\begin{align}\label{eq:estiphi0byphi1negav1}
&\inttau r^{-2}{|\widetilde{\phi^0}|^2}
\lesssim
\inttau r^{-3}|\widetilde{\phi^1}|^2
+\int_{\Sigma_0}r^{-1}{|\widetilde{\phi^0}|^2}.
\end{align}
Combining \eqref{eq:estiphi1byphi02negafinal1} with \eqref{eq:estiphi0byphi1negav1} proves the estimate \eqref{eq:estiphi0byphi02negafinal1}.

We prove the inequality \eqref{eq:estiphi0byphi1angunegafinal1} below, the proof for \eqref{eq:estiphi1byphi2angunegafinal1} being analogous. For a smooth cutoff function $\chi_2(r)$ which equals $1$ in $[6M,\infty)$ and vanishes in $[r_+,5M]$,
any real value $\beta$ and $\nablaslash_j$ $(j=1,2,3)$ as defined in \eqref{SpinWeightedAngularDerivaBasisOnSphere}, it holds true that
\begin{align}\label{eq:estiphi0negaiden1v4}
&V(f_{-2}\chi_2 r^{\beta}|r^2\nablaslash_j\phi^0|^2)
-\chi_2\partial_rf_{-2}r^{\beta}|r^2\nablaslash_j\phi^0|^2\notag\\
&
-\left(\beta \chi_2 f_{-2}+\partial_r\chi_2 f_{-2}r\right)r^{\beta-1}|r^2\nablaslash_j\phi^0|^2
={}\chi_2f_{-2}
r^{2+\beta}\Re(\nablaslash_j\phi^0
\overline{\nablaslash_j\phi^1}).
\end{align}
Choosing $\beta=-1$ and $f_{-2}=\tfrac{(\R)^3}{\Delta^3}$, integrating \eqref{eq:estiphi0negaiden1v4} over $\mathcal{D}(0,\tau)$ with the measure $\di \check{V}$ in \eqref{eq:measureinestiforphi1}, and applying the Cauchy--Schwarz inequality to the last term, the estimate \eqref{eq:estiphi0byphi1angunegafinal1}  for $i=0$ follows by summing over $j=1,2,3$.
\qed
\end{proof}

\section{Proof of Theorem \ref{thm:EneAndMorEstiExtremeCompsNoLossDecayVersion2} on Schwarzschild}\label{sect:pfMainthmSchwS2}
We derive the estimates
\eqref{eq:estiphi02hatphi1kerrS2} and \eqref{eq:estiphi02hatphi1kerrS2Nega} on Schwarzschild backgrounds, thus finishing the proof of Theorem \ref{thm:EneAndMorEstiExtremeCompsNoLossDecayVersion2} on Schwarzschild for $n=0$ from the discussions in Section \ref{sect:outlineproof}. The $n\geq 1$ cases are discussed in Section \ref{sect:highorderS2}.

\subsection{Coupled System on Schwarzschild}\label{sect:coupledsysScwhS2}
In a Schwarzschild spacetime, systems \eqref{eq:ReggeWheeler Phi^012KerrS2} and  \eqref{eq:ReggeWheeler Phi^012KerrNegaS2} for $\phi_{s}^i$ $(i=0,1,2)$ can be written in a unified form:

\begin{subequations}\label{eq:RWschwphi012ope}
\begin{align}
\label{eq:RWphi02schwope}
\mathbf{L}_{s}^0\phi^0_{s} =&F_{s}^0=\tfrac{4(r-3M)}{r^2}\phi^1_{s},\\
\label{eq:RWphi1schwope}
\mathbf{L}_{s}^1\phi^1_{s} =&F_{s}^1=\tfrac{2(r-3M)}{r^2}\phi^2_{s}+6M\phi^0_{s},\\
\label{eq:RWphi22schwope}
\mathbf{L}_{s}^1\phi^2_{s} =&F_{s}^2=0,
\end{align}
\end{subequations}
with the operators simplified to
\begin{subequations}\label{eq:SWFIOpesSchw}
\begin{align}
\label{eq:SWFIOpe0Schw}
\mathbf{L}_{s}^0=&\Sigma\Box_g+\tfrac{2is\cos\theta}{\sin^2 \theta}\partial_{\phi}-s^2\left(\cot^2\theta +\tfrac{r+2M}{2r}\right),\\
\label{eq:SWFIOpe2Schw}
\mathbf{L}_{s}^1=&\Sigma\Box_g+\tfrac{2is\cos\theta}{\sin^2 \theta}\partial_{\phi}-s^2\left(\cot^2\theta +\tfrac{r-2M}{r}\right).
\end{align}
\end{subequations}

\subsection{Decomposition}\label{sect:decompSchwS2}

The equations \eqref{eq:RWphi1schwope} and \eqref{eq:RWphi22schwope} are both in the form of an ISWWE
\begin{equation}\label{eq:SWRWphi12generalformSchw}
\mathbf{L}_s^1\varphi_s^1=
\Sigma\Box_g\varphi_s^1+\tfrac{2is\cos\theta}{\sin^2 \theta}\partial_{\phi}\varphi_s^1-s^2\left(\cot^2\theta +\tfrac{r-2M}{r}\right)\varphi_s^1=G_s^1.
\end{equation}
We will from now on suppress the subscript $s$ in the functions $\varphi_s^1$ and $G_s^1$, as well as in $\varphi_s^0$ and $G_s^0$ in \eqref{eq:SWRWphi0generalformSchw}, but retain it for the operators.

Decompose the solution $\varphi^1$ and the inhomogeneous term $G^1$ into
\begin{align}\label{eq:decompSWSH Schw}
\varphi^1&=\sum_{m,\ell}\varphi^1_{m\ell}(t,r)Y_{m\ell}^{s}(\cos\theta)e^{im\phi},
\quad m\in \mathbb{Z},\\
\label{eq:decompSWSH Schw SourceTermF}
G^1&=\sum_{m,\ell}G^1_{m\ell}(t,r)Y_{m\ell}^{s}(\cos\theta)e^{im\phi},
\quad m\in \mathbb{Z}.
\end{align}
Here, for each $m$, $\left\{Y_{m\ell}^{s}(\cos\theta)\right\}_{\ell}$ with $\min{\{\ell\}}=\max{(|m|,|s|)}$ are the eigenfunctions of the self-adjoint operator
\begin{align}\label{eq:SWSHOpe Schw}
\textbf{S}_m=\tfrac{1}{\sin\theta}\partial_{\theta}\sin\theta\partial_{\theta}
-\tfrac{m^2+2ms\cos\theta+s^2}{\sin^2\theta}
\end{align}
on $L^2(\sin\theta \di \theta)$. These eigenfunctions, called as \emph{spin-weighted spherical harmonics}, form a complete orthonormal basis on $L^2(\sin\theta \di \theta)$ and have eigenvalues $-\Lambda_{m\ell}=-\ell(\ell+1)$ \footnote{Note that in Schwarzschild case, $\Lambda_{m\ell}=A+s+s^2$, with $A$ being the separation constant in \cite{Teu1972PRLseparability}.} defined by
\begin{equation}
\textbf{S}_mY_{m\ell}^{s}(\cos\theta)=-\Lambda_{m\ell}
Y_{m\ell}^{s}(\cos\theta).
\end{equation}
An integration by parts, together with a usage of Plancherel lemma and the orthonormality property of the basis $\left\{Y_{m\ell}^{s}(\cos\theta)e^{im\phi}\right\}_{m\ell}$, gives
\begin{align}\label{eq:IdenOfEigenvaluesAndAnguDeriSchw}
\sum_{m,\ell}\ell(\ell+1)\left|\varphi^1_{m\ell}(t,r)\right|^2&
=\int_{0}^{\pi}\int_{0}^{2\pi}\left|\nablaslash \varphi^1(t,r)\right|^2r^2\sin\theta \di \phi \di \theta.
\end{align}
\begin{remark}
\label{rem:EigenvalueSpinWeightedAngular}
This actually shows that the eigenvalues of the operator
\begin{align}
\tfrac{1}{\sin{\theta}} \partial_{\theta}(\sin \theta \partial_{\theta})
+\tfrac{1}{\sin^2\theta}\partial_{\phi\phi}^2
+\tfrac{2is\cos\theta}{\sin^2 \theta}\partial_{\phi}-s^2(\sin \theta)^{-2}
\end{align}
are not greater than $-s^2-|s|$, and for a scalar $\varphi$ with spin weight $s$,
\begin{align}
\label{eq:LowerBdfornablaslash}
\int_{0}^{\pi}\int_{0}^{2\pi}\left|\nablaslash\varphi\right|^2r^2\sin\theta \di \phi \di \theta \geq {}(s^2 +|s|) \int_{0}^{\pi}\int_{0}^{2\pi}\left|\varphi\right|^2\sin\theta \di \phi \di \theta.
\end{align}
Moreover, let
$\varphi_m(t,r,\theta)=\frac{1}{\sqrt{2\pi}}\int_0^{2\pi}e^{-im\phi}
\varphi \di \phi$, then we have
\begin{align}\label{eq:AnEstiForEigenvalueKerr}
\int_{\mathbb{S}^2}\left|r\nablaslash\varphi\right|^2
 \di \sigma_{\mathbb{S}^2}
\geq {}\int_0^{\pi}\sum_{m\in \mathbb{Z}}\max\{s^2+|s|,m^2+|m|\}|\varphi_{m}|^2\sin\theta \di \theta,
\end{align}
\end{remark}

The equation for $\varphi^1_{m\ell}$ is now
\begin{align}\label{eq:SWRWReducedSchw}
r^4\Delta^{-1}
\partial_{tt}^2\varphi^1_{m\ell}-\partial_r(\Delta\partial_r)\varphi^1_{m\ell}
+\ell(\ell+1)\varphi^1_{m\ell}-8M/r\varphi^1_{m\ell}+G^1_{m\ell}=0.
\end{align}
In the case that the inhomogeneous term $G^1=0$, this is exactly the equation one obtains after decomposing into spherical harmonics the solution to the classical Regge--Wheeler equation \cite{ReggeWheeler1957} on Schwarzschild:
 \begin{equation}\label{eq:RWSchw}
\Sigma \Box_g u + \tfrac{8M}{r} u=0.
\end{equation}

The equation \eqref{eq:RWphi02schwope}, while, is in a form of an ISWWE with another potential:
\begin{equation}\label{eq:SWRWphi0generalformSchw}
\mathbf{L}_s^0\varphi^0=
\Sigma\Box_g\varphi^0+\tfrac{2is\cos\theta}{\sin^2 \theta}\partial_{\phi}\varphi^0-s^2\left(\cot^2\theta +\tfrac{r+2M}{2r}\right)\varphi^0=G^0.
\end{equation}
After the decomposition into spin-weighted spherical harmonics as above, the equation for $\varphi^0_{m\ell}$ reads
\begin{align}
\label{eq:SWRW2ReducedSchw}
r^4\Delta^{-1}
\partial_{tt}^2\varphi^0_{m\ell}-\partial_r(\Delta\partial_r)\varphi^0_{m\ell}
+\ell(\ell+1)\varphi^0_{m\ell}-(2-4M/r)\varphi^0_{m\ell}+G^0_{m\ell}=0.
\end{align}
Identity \eqref{eq:IdenOfEigenvaluesAndAnguDeriSchw} holds for $\varphi^0$ as well.

We now consider the general form of equations \eqref{eq:SWRWReducedSchw} and \eqref{eq:SWRW2ReducedSchw}:
\begin{equation}\label{eq:SWRWmodeSchwGeneral}
r^4\Delta^{-1}
\partial_{tt}^2\varphi-\partial_r(\Delta\partial_r)\varphi
+\ell(\ell+1)\varphi+V(r)\varphi+G=0,
\end{equation}
with the potential
\begin{align}\label{eq:potentialindifferenteqsSchwS2}
V(r)=
\left\{
  \begin{array}{ll}
    -8M/r, & \quad \text{for\ \eqref{eq:SWRWReducedSchw}},  \\
-2+4M/r, & \quad \text{for\ \eqref{eq:SWRW2ReducedSchw}}.
  \end{array}
  \right.
\end{align}

\subsection{Energy Estimate}

Multiplying \eqref{eq:SWRWmodeSchwGeneral} by $T\overline{\varphi}=\partial_{t}\overline{\varphi}$ and taking the real part, we arrive at an identity:
\begin{align}
\hspace{4ex}&\hspace{-4ex}\half\partial_t\left( \tfrac{r^4}{\Delta}|\partial_t \varphi|^2+\Delta|\partial_r\varphi|^2
+ \ell(\ell+1)|\varphi|^2+ V|\varphi|^2\right)
-\partial_r\left(\Re(\Delta\partial_r\varphi
\partial_t\overline{\varphi})\right)\notag\\
={}&-\Re(G\partial_{t}\overline{\varphi}).
\end{align}
Since $\ell\geq |s|=2$ and $\ell(\ell+1)\geq 6$, the inequality
 \begin{equation}
\ell(\ell+1)+V(r)\geq \tfrac{1}{3} \ell(\ell+1)
\end{equation}
is valid for both potentials in \eqref{eq:potentialindifferenteqsSchwS2}.
Summing over $m$ and $\ell$, applying the identity \eqref{eq:IdenOfEigenvaluesAndAnguDeriSchw} for $\varphi^1$ and $\varphi^0$, and finally integrating with respect to the measure $\di t^*\di r$ over $\{(t^*,r)|0\leq t^*\leq \tau, 2M\leq r<\infty\}$, we have the following energy estimate for $\varphi^i$  $(i=0,1)$:
\begin{align}\label{eq:enerestipartialtSchw}
{E}_{\tau}^T(\varphi^i)\lesssim E_{0}^T(\varphi^i)
+\int_{\mathcal{D}(0,\tau)}r^{-2}\left|
\Re(G^i\partial_t\overline{\varphi^i})\right|.
\end{align}
In global Kerr coordinates, for any $\tau\geq0$,
\begin{equation}
E_{\tau}^T(\varphi^i)\sim
\int_{\Sigma_{\tau}}\bigg(|\partial_{t^*}\varphi^i|^2+|\nablaslash \varphi^i|^2+\frac{\Delta}{r^2}|\partial_r\varphi^i|^2\bigg).
\end{equation}

\subsection{Morawetz Estimate}\label{proof of integrated decay}

In this subsection, following the choices of the multipliers in \cite{larsblue15hidden,Jinhua17LinGraSchw}, we prove the Morawetz estimate for the separated equations \eqref{eq:SWRWReducedSchw} and \eqref{eq:SWRW2ReducedSchw} which are both in the form of \eqref{eq:SWRWmodeSchwGeneral}, and then derive the Morawetz estimate for the ISWWE \eqref{eq:SWRWphi12generalformSchw} and \eqref{eq:SWRWphi0generalformSchw}.

We multiply \eqref{eq:SWRWmodeSchwGeneral} by
\begin{align}
X\bar{\varphi}=\hat{f}\partial_r \bar{\varphi}+\hat{q}\bar{\varphi},
\end{align}
take the real part, and arrive at
\begin{align}\label{eq:MoraEstidegSchw}
\hspace{4ex}&\hspace{-4ex}\partial_t\left(\Re\left(\tfrac{r^4}{\Delta}X(\varphi)
\partial_t\bar{\varphi}\right)\right)
+\tfrac{1}{2}\partial_r\left(\hat{f}\left[\ell(\ell+1)|\varphi|^2
-\tfrac{r^4}{\Delta}|\partial_t\varphi|^2-\Delta |\partial_r\varphi|^2+V|\varphi|^2\right]\right)\notag\\
\hspace{4ex}&\hspace{-4ex}
+\tfrac{1}{2}\partial_r\left(\Re(\partial_r(\Delta\hat{q})|\varphi|^2-2\Delta \hat{q}\bar{\varphi}\partial_r\varphi-2\hat{q}(r-M)|\varphi|^2-r^{-1}B^r(r)|\varphi|^2)\right)
+B(\varphi)\notag\\
={}&-\Re(X\overline{\varphi}{G}).
\end{align}
Here, the bulk term
\begin{align}
B(\varphi)=B^t(r)|\partial_t\varphi|^2+r^{-2}B^r(r)|\partial_r(r\varphi)|^2
+B^{\ell}(r)\left(\ell(\ell+1)|\varphi|^2\right)+B^0(r)|\varphi|^2,
\end{align}
with
\begin{align}
B^t(r)={}&\tfrac{1}{2}\partial_r
(r^4\Delta^{-1}\hat{f})
-\hat{q}r^4\Delta^{-1}\notag\\
B^r(r)={}&\tfrac{1}{2}\partial_r(\Delta\hat{f})-2\hat{f}(r-M)+\Delta \hat{q} \notag\\
B^{\ell}(r)={}&-\tfrac{1}{2}\partial_r(\hat{f})+\hat{q}\notag\\
B^0(r)={}&\partial_r(\hat{q}(r-M))-\tfrac{1}{2}\partial_{rr}^2(\Delta \hat{q})+V\hat{q}-\frac{1}{2} \partial_r(V\hat{f})\notag\\
&+r^2(\partial_r (r^{-3}B^r(r))+ r^{-4}B^r(r)).
\end{align}
By choosing $\hat{f}=2r^{-2}(r-2M)(r-3M)$ and $\hat{q}= r^{-4}(2r-3M)\Delta$,
we find
\begin{align}
B^t(r)=0,\quad
B^r(r)=\tfrac{6M\KDelta^2}{r^4}, \quad
B^{\ell}(r)=\tfrac{2(r-3M)^2}{r^3},
\end{align}
and
\begin{align}
B^0(r)=
\left\{
  \begin{array}{ll}
    -27Mr^{-2}+162M^2r^{-3}-234M^3r^{-4}, & \quad \text{for\ \eqref{eq:SWRWReducedSchw}},  \\
-4r^{-1}+{33}Mr^{-2}-78M^2r^{-3}+54M^3r^{-4}, & \quad \text{for\ \eqref{eq:SWRW2ReducedSchw}}.
  \end{array}
  \right.
\end{align}

We first treat \eqref{eq:SWRWReducedSchw} by calculating $P^1_{s}(r)= B^0(r) + 6 B^{\ell}(r)$ that
\begin{align}
\label{eq:PotentialS2schwHardy}
P^1(r)= {}&3r^{-4}(4r^3-33Mr^2+90M^2r-78M^3).
\end{align}
Since $\ell(\ell+1)\geq 6$, we have then
\begin{align}
B(\varphi)\geq r^{-2}B^r(r)|\partial_r(r\varphi)|^2
+P^1(r)|\varphi|^2.
\end{align}
By calculating the roots of the the third order polynomial $\tfrac{1}{3}r^4P^1(r)$, we find there exists only one real root for this polynomial for $|s|=2$, and this real root is less than $2M$. Hence,  there exists a constant $c>0$ such that for any $r\geq 2M$,
\begin{align}\label{eq:Moraestibulktermcontrol1}
B(\varphi)\geq
c \left(\tfrac{\Delta^2}{r^4} |\p_r\varphi|^2+\tfrac{1}{r} |\varphi|^2+\tfrac{(r-3M)^2}{r^3}\ell(\ell+1)|\varphi|^2\right).
\end{align}

Instead, if we multiply \eqref{eq:SWRWmodeSchwGeneral} by $h\bar{\varphi}$ with $h=-\tfrac{\Delta(r-3M)^2}{r^7}$
and take the real part, the identity \eqref{eq:MoraEstidegSchw} becomes
\begin{align}\label{eq:MoraEstidegSchw2}
\hspace{4ex}&\hspace{-4ex}\tfrac{1}{2}\partial_r\left(\Re(\partial_r(\Delta h)|\varphi|^2-2\Delta h\bar{\varphi}\partial_r\varphi-2h(r-M)|\varphi|^2)\right)
+h\left(\ell(\ell+1)|\varphi|^2\right)
\notag\\
\hspace{4ex}&\hspace{-4ex}
+\partial_t\left(\Re\left(\tfrac{r^4}{\Delta}h\varphi
\partial_t\bar{\varphi}\right)\right)
-h\tfrac{r^4}{\Delta}|\partial_t\varphi|^2+\Delta h|\partial_r\varphi|^2\notag\\
\hspace{4ex}&\hspace{-4ex}
+\left(\partial_r\left(h(r-M)\right)-\half\partial_{rr}^2(\Delta h)+hV\right)|\varphi|^2\notag\\
={}&-\Re(h\varphi\overline{G}).
\end{align}
After integration, this allows us to control the bulk integral of $|\partial_t \varphi|^2$ part by the bulk integral of the RHS of \eqref{eq:Moraestibulktermcontrol1}.
We sum over $m$ and $\ell$ for \eqref{eq:MoraEstidegSchw} and \eqref{eq:MoraEstidegSchw2} with $\varphi=\varphi_{m\ell}^1$ and $G=G_{m\ell}^1$, apply the identity \eqref{eq:IdenOfEigenvaluesAndAnguDeriSchw}, integrate with respect to the measure $\di t^*\di r$ over $\{(t^*,r)|0\leq t^*\leq \tau, 2M\leq r<\infty\}$,
and take \eqref{eq:Moraestibulktermcontrol1} into account, then
we obtain a Morawetz estimate for \eqref{eq:SWRWphi12generalformSchw} in global Kerr coordinates:
\begin{align}\label{Morawetz-0}
\hspace{4ex}&\hspace{-4ex}\int_{\mathcal{D}(0,\tau)}\left(\tfrac{\Delta^2}{r^6} |\p_r\varphi^1|^2+\tfrac{1}{r^4} |\varphi^1|^2+\tfrac{(r-3M)^2}{r^2}\left(\tfrac{1}{r^3}|\p_{t^*}\varphi^1|^2 + \tfrac{1}{r}|\nablaslash\varphi^1|^2\right)\right)\notag\\
 \lesssim {}&
E_{\tau}^T(\varphi^1)+E_{0}^T(\varphi^1)
+\int_{\mathcal{D}(0,\tau)}r^{-2}\left(\left|
\Re\left(X(\varphi^1)\overline{G^1}\right)\right|
+\left|\Re\left(h\varphi^1\overline{G^1}\right)\right|\right).
\end{align}

Turning now to \eqref{eq:SWRW2ReducedSchw}, similarly as above, we calculate
\begin{align}
P^0(r)={}&B^0(r) + 6B^{\ell}(r) \notag\\
={}&r^{-4}(8r^3-39Mr^2+30M^2r+54M^3).
\end{align}
One finds that there is only one real root for the third order polynomial $r^4 P^0(r)$, and this real root is negative, hence $r^4 P^0(r)$ is positive for $r\geq 2M$, which yields
\begin{align}
B(\varphi)\geq {}&r^{-2}B^r(r)|\partial_r(r\varphi)|^2
+P^0(r)|\varphi|^2+B^{\ell}(r)\left(\ell(\ell+1)-6\right)|\varphi|^2\notag\\
\geq {}& c \left(\tfrac{\Delta^2}{r^4} |\p_r\varphi|^2+\tfrac{1}{r} |\varphi|^2+\tfrac{(r-3M)^2}{r^3}\ell(\ell+1)|\varphi|^2\right).
\end{align}
Following the argument above in proving from \eqref{eq:MoraEstidegSchw2} to \eqref{eq:SWRWReducedSchw}, we obtain the following Morawetz estimate for equation \eqref{eq:SWRWphi0generalformSchw} in global Kerr coordinates:
\begin{align}\label{Morawetz-0-psi0}
\hspace{4ex}&\hspace{-4ex}\int_{\mathcal{D}(0,\tau)}\left(\tfrac{\Delta^2}{r^6} |\p_r\varphi^0|^2+\tfrac{1}{r^4} |\varphi^0|^2+\tfrac{(r-3M)^2}{r^2}\left(\tfrac{1}{r^3}|\p_{t^*}\varphi^0|^2 + \tfrac{1}{r}|\nablaslash\varphi^0|^2\right)\right)\notag\\
 \lesssim {}&
E_{\tau}^T(\varphi^1)+E_{0}^T(\varphi^0)
+\int_{\mathcal{D}(0,\tau)}r^{-2}\left(\left|
\Re\left(X(\varphi^0)\overline{G^0}\right)\right|
+\left|\Re\left(h\varphi^0\overline{G^0}\right)\right|\right).
\end{align}

\subsection{Close the Proof of the Estimates \eqref{eq:estiphi02hatphi1kerrS2} and \eqref{eq:MoraEstiFinal(2)KerrRegularpsiBothSpinComp} on Schwarzschild}\label{sect:finishpfSchwS2}

The general idea is to combine the Morawetz estimate \eqref{Morawetz-0} applied to \eqref{eq:RWphi1schwope} and \eqref{eq:RWphi22schwope} (or \eqref{Morawetz-0-psi0} applied to \eqref{eq:RWphi02schwope}), the energy estimate \eqref{eq:enerestipartialtSchw} applied to all subequations of \eqref{eq:RWschwphi012ope}, the Morawetz estimate in large radius region in Propositions \ref{prop:ImprovedMoraEstiLargerSWRWE1}
and \ref{prop:ImprovedMoraEstiLargerSWRWE}, and the red-shift estimate near horizon in Propositions \ref{prop:RedShiftEstiInhomoSWRWE} and
\ref{prop:RedShiftEstiInhomoSWRWEtildephi01}.
We consider spin $\pm 2$ components separately, and estimate the error terms arising from the above mentioned combined estimates.

\subsubsection{Spin $+2$ Component}
Applying the Morawetz estimates \eqref{Morawetz-0} to \eqref{eq:RWphi1schwope} and \eqref{eq:RWphi22schwope}, and \eqref{Morawetz-0-psi0} to \eqref{eq:RWphi02schwope}, and together with the energy estimate \eqref{eq:enerestipartialtSchw}, the Morawetz estimates in large radius region for $r^{4-\delta}\phi^0$ and $r^{2-\delta}\phi^1$ in Proposition \ref{prop:ImprovedMoraEstiLargerSWRWE}, and the red-shift estimates near horizon in Proposition \ref{prop:RedShiftEstiInhomoSWRWE},
we obtain
\begin{align}
\label{eq:enerMoraphi0Schw}
\hspace{4ex}&\hspace{-4ex}{E}_{\tau}(r^{4-\delta}\phi^0)+{E}_{\mathcal{H}^+(0,\tau)}(r^{4-\delta}\phi^0)
+\int_{\mathcal{D}(0,\tau)} \widetilde{\mathbb{M}}_{\text{deg}}(r^{4-\delta}\phi^0)\notag\\
\lesssim {}& {E}_{0}(r^{4-\delta}\phi^0)
+\mathcal{E}_{\text{schw}}(\phi^0),\\
\label{eq:enerMoraphi1Schw}
\hspace{4ex}&\hspace{-4ex}{E}_{\tau}(r^{2-\delta}\phi^1)
+{E}_{\mathcal{H}^+(0,\tau)}(r^{2-\delta}\phi^1)
+\int_{\mathcal{D}(0,\tau)}\widetilde{\mathbb{M}}_{\text{deg}}
(r^{2-\delta}\phi^1)\notag\\
\lesssim {}&
{E}_{0}(r^{2-\delta}\phi^1)
+\mathcal{E}_{\text{schw}}(\phi^1)
+\int_{\mathcal{D}(0,\tau)\cap [R-1,\infty)}r^{-2}{\left|\partial \left(r^{4-\delta}\phi^{0}\right)\right|^2},\\
\label{eq:enerMoraphi2Schw}
\hspace{4ex}&\hspace{-4ex}{E}_{\tau}(\phi^2)+{E}_{\mathcal{H}^+(0,\tau)}(\phi^2)
+\int_{\mathcal{D}(0,\tau)} \mathbb{M}_{\text{deg}}(\phi^2)
\lesssim
{E}_{0}(\phi^2).
\end{align}
The error term $\mathcal{E}_{\text{schw}}(\phi^0)$ is
\begin{align}\label{eq:EstiErrortermphi0Schw}
\hspace{4ex}&\hspace{-4ex}\int_{\mathcal{D}(0,\tau)}r^{-2}\left(\left|
\Re\left(X(\phi^0)\overline{F_{+2}^0}\right)\right|
+\left|\Re\left(h\phi^0\overline{F_{+2}^0}\right)\right|\right)
+\inttau r^{-2}|F_{+2}^0||\partial_t\phi^0|\notag\\
\hspace{4ex}&\hspace{-4ex}
+\int_{\mathcal{D}(0,\tau)\cap\{r\geq R-M\}}\left|\Re\left(F_{+2}^0 X_w\overline{\phi^0}\right)\right|
+\int_{\mathcal{D}(0,\tau)\cap\{r\leq r_1\}}{|\phi^1|^2},
\end{align}
and is bounded by $\eps_0 \inttau \widetilde{\mathbb{M}}(r\phi^0) +\frac{1}{\eps_0}\inttau \frac{|\phi^1|^2}{r^3}$ by the Cauchy--Schwarz inequality, and hence by $\int_{\mathcal{D}(0,\tau)} \left(\eps_0 \widetilde{\mathbb{M}}(r^{4-\delta}\phi^0)
+
\frac{\hat{\eps}_1}{\eps_0} \widetilde{\mathbb{M}}(r\phi^1)+\frac{1}{\eps_0\hat{\eps}_1} \mathbb{M}_{\text{deg}}(\phi^2)\right)+{E}_{0}(r^{4-\delta}\phi^0)$ from the estimate \eqref{eq:estiphi1byphi02final1}.
For the other error term $\mathcal{E}_{\text{schw}}(\phi^1)$, it is bounded by
\begin{align}\label{eq:EstiErrortermphi1Schw}
&
\int_{\mathcal{D}(0,\tau)}r^{-2}\left(\left|
\Re\left(X(\phi^1)\overline{F_{+2}^1}\right)\right|
+\left|\Re\left(h\phi^1\overline{F_{+2}^1}\right)\right|\right)
+\inttau r^{-2}|F_{+2}^1||\partial_t\phi^0|\notag\\
&
+\int_{\mathcal{D}(0,\tau)\cap\{r\geq R-M\}}\Big(\left|\Re(F_{+2}^1 X_w\overline{\phi^1})\right|
+\tfrac{\left|\partial \left(r^{4-\delta}\phi^{0}_{+2}\right)\right|^2}{r^2}\Big)
+\int_{\mathcal{D}(0,\tau)\cap\{r\leq r_1\}}{|F_{+2}^1 |^2},
\end{align}
By applying the Cauchy--Schwarz inequality, we can bound it by
\begin{align}\label{eq:Estierrortermphi1Schw}
C\eps_1\inttau \widetilde{\mathbb{M}}(r^{2-\delta}\phi^1)
+C\eps_1^{-1}\inttau \left(\mathbb{M}_{\text{deg}}(\phi^2)
+r^{-1}|r^{4-\delta}\phi^0|^2\right).
\end{align}
Hence, this completes the proof of \eqref{eq:estiphi02hatphi1kerrS2}.

\subsubsection{Spin $-2$ Component}
The Morawetz estimate \eqref{Morawetz-0} applied to \eqref{eq:RWphi1schwope} and \eqref{eq:RWphi22schwope}, and \eqref{Morawetz-0-psi0} applied to \eqref{eq:RWphi02schwope}, the energy estimate \eqref{eq:enerestipartialtSchw} applied to \eqref{eq:RWschwphi012ope}, the Morawetz estimates in large $r$ region for $\{\phi^i_{-2}\}|_{i=0,1,2}$ in Proposition \ref{prop:ImprovedMoraEstiLargerSWRWE1} and the red-shift estimates near horizon in in Propositions \ref{prop:RedShiftEstiInhomoSWRWE} and
\ref{prop:RedShiftEstiInhomoSWRWEtildephi01} for $\widetilde{\phi^0}$, $\widetilde{\phi^1}$ and ${\phi^2}$ together imply
\begin{subequations}
\label{eq:enerMoraphi012SchwNega}
\begin{align}
\label{eq:enerMoraphi0SchwNega}
&{E}_{\tau}(\widetilde{\phi^0})+{E}_{\mathcal{H}^+(0,\tau)}(\widetilde{\phi^0})
+\int_{\mathcal{D}(0,\tau)} \mathbb{M}_{\text{deg}}(\widetilde{\phi^0})\lesssim {E}_{0}(\widetilde{\phi^0})
+\mathcal{E}_{\text{schw}}(\widetilde{\phi^0}),\\
\label{eq:enerMoraphi1SchwNega}
&{E}_{\tau}(\widetilde{\phi^1})
+{E}_{\mathcal{H}^+(0,\tau)}(\widetilde{\phi^1})
+\int_{\mathcal{D}(0,\tau)}\mathbb{M}_{\text{deg}}(\widetilde{\phi^1})
\lesssim
{E}_{0}(\widetilde{\phi^1})
+\mathcal{E}_{\text{schw}}(\widetilde{\phi^1}),\\
\label{eq:enerMoraphi2SchwNega}
&{E}_{\tau}(\phi^2)+{E}_{\mathcal{H}^+(0,\tau)}(\phi^2)
+\int_{\mathcal{D}(0,\tau)} \mathbb{M}_{\text{deg}}(\phi^2)
\lesssim
{E}_{0}(\phi^2),
\end{align}
\end{subequations}
where
\begin{align}\label{eq:EstiErrortermphi0SchwNega}
&\mathcal{E}_{\text{schw}}(\widetilde{\phi^0})
\lesssim {}\int_{\mathcal{D}(0,\tau)}r^{-2}\left(\left|
\Re\left(X(\phi^0)\overline{F_{-2}^0}\right)\right|
+\left|\Re\left(h\phi^0\overline{F_{-2}^0}\right)\right|
+|F_{-2}^0||\partial_t\phi^0|\right)\notag\\
&\quad \quad
+\int_{\mathcal{D}(0,\tau)\cap\{r\geq R-M\}}|F_{-2}^0| |X_w\overline{\phi^0}|
+\int_{\mathcal{D}(0,\tau)\cap[r_+,r_1]}|\widetilde{\phi^1}|^2,\\
\label{eq:EstiErrortermphi1SchwNega}
&\mathcal{E}_{\text{schw}}(\widetilde{\phi^1})
\lesssim {}\int_{\mathcal{D}(0,\tau)}r^{-2}\left(\left|
\Re\left(X(\phi^1)\overline{F_{-2}^1}\right)\right|
+\left|\Re\left(h\phi^0\overline{F_{-2}^1}\right)\right|
+|F_{-2}^1||\partial_t\phi^1|\right)\notag\\
&\quad \quad +
\int_{\mathcal{D}(0,\tau)\cap\{r\geq R-M\}}|F_{-2}^1| |X_w\overline{\phi^1}|
+\int_{\mathcal{D}(0,\tau)\cap[r_+,r_1]}
\Big(|\widetilde{\phi^0}|^2+|\phi^2|^2\Big).
\end{align}
Substituting the expressions of $F_{-2}^0$ and $F_{-2}^1$ and applying the Cauchy--Schwarz inequality to the RHS of the estimates \eqref{eq:EstiErrortermphi0SchwNega} and \eqref{eq:EstiErrortermphi1SchwNega}, one finds
\begin{align}
\label{eq:errorSchwphi0-2}
&\mathcal{E}_{\text{schw}}(\widetilde{\phi^0})\lesssim \int_{\mathcal{D}(0,\tau)} \Big(\eps_0{\mathbb{M}}(\widetilde{\phi^0})
+{\eps_0^{-1}}
\tfrac{|\widetilde{\phi^1}|^2}{r^3}\Big),\\
\label{eq:errorSchwphi1-2}
&\mathcal{E}_{\text{schw}}(\widetilde{\phi^1})
\lesssim {}\int_{\mathcal{D}(0,\tau)} \Big(\eps_1{\mathbb{M}}(\widetilde{\phi^1})
+{\eps_1^{-1}}\Big(
\mathbb{M}_{\text{deg}}(\phi^2)
+\tfrac{|\widetilde{\phi^0}|^2}{r^2}\Big)\Big).
\end{align}
The estimates \eqref{eq:estiphi02hatphi1kerrS2Nega} follow from the estimates \eqref{eq:enerMoraphi012SchwNega}, applying \eqref{eq:estiphi1byphi02negafinal1} to \eqref{eq:errorSchwphi0-2}, and applying \eqref{eq:estiphi0byphi02negafinal1} to \eqref{eq:errorSchwphi1-2}.

\section{Proof of Theorem \ref{thm:EneAndMorEstiExtremeCompsNoLossDecayVersion2} on Slowly Rotating Kerr}\label{sect:pfMainThmKerrS2}

\subsection{Energy Estimate}\label{sect:EnerEstiPartialtKerr}
We start  by choosing a multiplier $-2\Sigma^{-1}\partial_t \bar{\psi}$ for \eqref{eq:RewrittenFormofSWRWEOpeForm} and obtain a conservation law that for $\tau_2>\tau_1$ that
\begin{equation}\label{eq:SWRWEnergyIdentity}
\int_{\Sigma_{\tau_2}}e^1(\psi)
+\int_{\mathcal{H}^+(\tau_1,\tau_2)}e^1_{\mathcal{H}}(\psi)=
\int_{\Sigma_{\tau_1}}e^1(\psi)
-\int_{\mathcal{D}(\tau_1,\tau_2)}\Re\left(
2\Sigma^{-1}F\partial_t \bar{\psi}\right).
\end{equation}
Here, the energy densities are
\begin{subequations}
\begin{align}\label{eq:SWRWEnergyDensity}
e^1(\psi)={}\di t^*(\partial_t) e^1_t(\psi) +\di t^* (\partial_{r^*})e^1_{r^*}(\psi), \quad e^1_{\mathcal{H}}(\psi)={}\di r(\partial_t) e^1_t(\psi) +\di r(\partial_{r^*})e^1_{r^*}(\psi),
\end{align}
and
\begin{align}
\label{eq:energyt1}
e^1_t(\psi)={}&
\Sigma^{-1}\Big(|\partial_{\theta}\psi|^2
+\Big|\tfrac{\partial_{\phi}\psi+
is\cos\theta\psi}{\sin\theta}\Big|^2
-\tfrac{a^2}{\Delta}|\partial_{\phi}\psi|^2
+\tfrac{s^2(\Delta+a^2)}{r^2} |\psi|^2\Big)\notag\\
&+
\tfrac{(r^2+a^2)^2-a^2\sin^2\theta\Delta}{\Delta\Sigma}
|\partial_t\psi|^2
+\tfrac{(r^2+a^2)^2}{\Delta\Sigma}|\partial_{r^*}\psi|^2,\\
e^1_{r^*}(\psi)={}&
-\tfrac{2(r^2+a^2)^2}{\Sigma\Delta}\Re\left(\partial_t \psi \partial_{r^*}\bar{\psi}\right).
\end{align}
\end{subequations}
Let $\psi_m(t,r,\theta)=
\frac{1}{\sqrt{2\pi}}\int_0^{2\pi}e^{-im\phi}
\psi \di \phi$.
It follows from \eqref{eq:AnEstiForEigenvalueKerr} and \eqref{def:nablaslashModuleSquare} that
\begin{align}\label{eq:SWRGEnergyNonnegative Kerr}
\hspace{4ex}&\hspace{-4ex}\int_{\mathbb{S}^2}\left(|\partial_{\theta}\psi|^2+\left|
\tfrac{\partial_{\phi}\psi+is\cos\theta\psi}{\sin\theta}
\right|^2-\tfrac{a^2}{\Delta}|\partial_{\phi}\psi|^2
+\tfrac{s^2(\Delta+a^2)}{r^2}|\psi|^2\right)\di \sigma_{\mathbb{S}^2}\notag\\
\geq &\int_0^{\pi}\sum_{m\in \mathbb{Z}} \left(\max\{s^2+|s|,m^2+|m|\}-\tfrac{a^2m^2}{\Delta}
+s^2\tfrac{\Delta+a^2-r^2}{r^2}\right)
|\psi_{m}|^2\sin\theta \di \theta.
\end{align}
Denote
\begin{equation}
A_{|m|,|s|}^1=\max\{s^2+|s|,m^2+|m|\}-\tfrac{a^2m^2}{\Delta}
+s^2\tfrac{\Delta+a^2-r^2}{r^2}.
\end{equation}
Recall that $|s|=2$. If $|m|=0$ or $1$, then
\begin{equation}
A_{|m|,2}^1\geq 2-\tfrac{a^2m^2}{\Delta}+\tfrac{4(\Delta+a^2)}{r^2},
\end{equation}
which is positve when $r> 2M$. When $|m|\geq 2$, $A_{|m|,2}^1$ is a monotonically increasing function of $|m|$ in the region $r\geq 2M$. Consider an even restricted region $r>M+ \sqrt{M^2+a^2}$. For $|m|\geq 2$, by monotonicity,
\begin{equation}
A_{|m|,2}^1\geq {}A_{2,2}^1 \geq \tfrac{2\Delta-4a^2}{\Delta}+\tfrac{4(\Delta+a^2)}{r^2},
\end{equation}
and the RHS is positive for $r\geq M+\sqrt{M^2+a^2}$. Therefore, the LHS of \eqref{eq:SWRGEnergyNonnegative Kerr}, and hence the energy density $\int_{\mathbb{S}^2}e^1(\psi)\di \sigma_{\mathbb{S}^2}$, is positive for $r>M+\sqrt{M^2+a^2}$.

One can similarly choose the multiplier $-2\Sigma^{-1}\partial_t \bar{\psi}$ for \eqref{eq:RewrittenFormofISWWEphi0OpeForm} satisfied by $\phi_s^0$, and arrive at an energy identity for any $\tau_2>\tau_1$:
\begin{equation}\label{eq:SWRWEnergyIdentityphi0}
\int_{\Sigma_{\tau_2}}e^0(\psi)
+\int_{\mathcal{H}^+(\tau_1,\tau_2)}e^0_{\mathcal{H}}(\psi)=
\int_{\Sigma_{\tau_1}}e^0(\psi)
-\int_{\mathcal{D}(\tau_1,\tau_2)}\Re\left(
2\Sigma^{-1}F\cdot\partial_t \bar{\psi}\right).
\end{equation}
Here, the energy densities
\begin{subequations}
\begin{align}\label{eq:SWRWEnergyDensityphi0}
e^0(\psi)={}&\di t^*(\partial_t) e^0_t(\psi) +\di t^* (\partial_{r^*})e^0_{r^*}(\psi), & e^0_{\mathcal{H}}(\psi)={}&\di r(\partial_t) e^0_t(\psi) +\di r(\partial_{r^*})e^0_{r^*}(\psi),
\end{align}
$e^0_{r^*}(\psi)$ is the same as $e^1_{r^*}(\psi)$, and
\begin{align}
e^0_t(\psi)={}&
\Sigma^{-1}\Big(|\partial_{\theta}\psi|^2
+\Big|\tfrac{\partial_{\phi}\psi+
is\cos\theta\psi}{\sin\theta}\Big|^2
-\tfrac{a^2}{\Delta}|\partial_{\phi}\psi|^2
+\tfrac{s^2(r^2+2Mr-2a^2)}{2r^2} |\psi|^2\Big)\notag\\
&+
\tfrac{(r^2+a^2)^2-a^2\sin^2\theta\Delta}{\Delta\Sigma}
|\partial_t\psi|^2+\tfrac{(r^2+a^2)^2}{\Delta\Sigma}|\partial_{r^*}\psi|^2.
\end{align}
\end{subequations}
The estimate \eqref{eq:AnEstiForEigenvalueKerr} gives
\begin{align}\label{eq:SWRGEnergyNonnegative Kerrphi0}
\hspace{4ex}&\hspace{-4ex}\int_{\mathbb{S}^2}\left(|\partial_{\theta}\psi|^2+\left|
\tfrac{\partial_{\phi}\psi+is\cos\theta\psi}{\sin\theta}
\right|^2-\tfrac{a^2}{\Delta}|\partial_{\phi}\psi|^2
+\tfrac{s^2(r^2+2Mr-2a^2)}{2r^2}|\psi|^2\right)\di \sigma_{S^2}\notag\\
\geq {}&\int_0^{\pi}\sum_{m\in \mathbb{Z}} \left(\max\{s^2+|s|,m^2+|m|\}-\tfrac{a^2m^2}{\Delta}
-s^2\tfrac{\Delta+a^2}{2r^2}\right)
|\psi_{m}|^2\sin\theta \di \theta.
\end{align}
Denote
\begin{equation}
A_{|m|,|s|}^0=\max\{|s|(|s|+1),|m|(|m|+1)\}-\tfrac{a^2m^2}{\Delta}
-s^2\tfrac{\Delta+a^2}{2r^2}.
\end{equation}
Recall that $|s|=2$. If $|m|=0$ or $1$,
$A_{|m|,2}^0\geq 2-\tfrac{a^2m^2}{\Delta}+\tfrac{2(r^2+2Mr-2a^2)}{r^2}$
and the RHS is positive when $r> 2M$. Note that $A_{|m|,2}^0$ is also a monotonically increasing function of $|m|$ in the region $r\geq 2M$ when $|m|\geq 2$. Hence, for $|m|\geq 2$, we have in the region $r>M+\sqrt{M^2+a^2}$ that
\begin{equation*}
A_{|m|,2}^0\geq A_{2,2}^0\geq \tfrac{2\Delta-4a^2}{\Delta}+\tfrac{2(r^2+2Mr-2a^2)}{r^2}>0.
\end{equation*}
This implies the energy density $\int_{\mathbb{S}^2}e^0(\psi)\di \sigma_{\mathbb{S}^2}$ is positive for $r>M+\sqrt{M^2+a^2}$.
For sufficiently small $|a|/M$, $M+\sqrt{M^2+a^2}<r_0$, where $r_0>2M$ has been fixed in Section \ref{sect:Redshift} (see also Remark \ref{rem:r0r1choice}). In conclusion, for $r\geq r_0$, the energy densities $\int_{\mathbb{S}^2}e^k(\psi)\di \sigma_{\mathbb{S}^2}$ $(k=0,1)$ above for both \eqref{eq:RewrittenFormofSWRWEOpeForm} and \eqref{eq:RewrittenFormofISWWEphi0OpeForm} are strictly positive and satisfy $\int_{\mathbb{S}^2}e^k(\psi)\geq c\int_{\mathbb{S}^2}|\partial \psi|^2$.

Since the energy densities $\int_{\mathbb{S}^2}e^k(\psi)$ are both nonnegative in the Schwarzschild case $(a=0)$, it holds true for sufficiently small $|a|/M \leq \veps_0\ll 1$,
\begin{subequations}
\label{eq:conservedenergy:nearandhorizon:esti}
\begin{align}
\label{eq:conservedenergy:nearhorizon:esti}
-\int_{\mathbb{S}^2}e^k(\psi)\leq {}&\frac{Ca^2}{M^2} \int_{\mathbb{S}^2}|\partial \psi|^2 \quad \text{for} \  r\in [r_+, r_0],\\
\label{eq:conservedenergy:horizon:esti}
\int_{\mathcal{H}^+(\tau_1,\tau_2)}e^k_{\mathcal{H}}(\psi)\gtrsim {}& -\frac{a^2}{M^2}{E}_{\mathcal{H}^+(\tau_1,\tau_2)}(\psi).
\end{align}
\end{subequations}
These can also be seen in another way by taking $k=1$ in \eqref{eq:conservedenergy:nearhorizon:esti} as an example. We consider only $r\leq M+r_0/2$ since the estimate for $M+r_0/2\leq r\leq r_0$ is manifestly valid. One can write $e_t^1(\psi)$ as
\begin{align}
\label{eq:energyt1:1}
&\Sigma^{-1}\Big(|\partial_{\theta}\psi|^2
+\Big|\tfrac{\partial_{\phi}\psi+
is\cos\theta\psi}{\sin\theta}\Big|^2
+\tfrac{s^2(\Delta+a^2)}{r^2} |\psi|^2+{\Delta}|Y\psi|^2\Big)\notag\\
&-
\tfrac{a^2\sin^2\theta}{\Sigma}
|\partial_t\psi|^2
-\tfrac{2a}{\Sigma}\Re(Y\psi\partial_{\phi}\bar{\psi})
+\tfrac{2(\R)^2}{\Delta\Sigma}\Re(\partial_t{\psi}\partial_{r^*}\bar{\psi}).
\end{align}
From definition \eqref{def:globalkerrcoord} of the global Kerr coordinates, $\di t^*(\partial_t)-\di t^*(\partial_{r^*})= \Delta/(\R)$ for $r\leq M+r_0/2$.
We add an $\Delta/(\R)$ multiple of \eqref{eq:energyt1} to a $1-\Delta/(\R)$ multiple of \eqref{eq:energyt1:1} to obtain a new expression of $e_{t}^1(\psi)$ and then substitute this into \eqref{eq:SWRWEnergyDensity} to obtain
\begin{align}
e^1(\psi)={}&
\tfrac{1}{\Sigma}\Big(|\partial_{\theta}\psi|^2
+\Big|\tfrac{\partial_{\phi}\psi+
is\cos\theta\psi}{\sin\theta}\Big|^2
+\tfrac{s^2(\Delta+a^2)}{r^2} |\psi|^2\Big)
+\tfrac{2Mr{\Delta}|Y\psi|^2}{\Sigma(\R)}
\notag\\
&+\tfrac{1}{\Sigma(\R)}
(((r^2+a^2)^2-a^2\sin^2\theta\Delta)
|\partial_t\psi|^2
+(r^2+a^2)^2|\partial_{r^*}\psi|^2)\notag\\
&-\tfrac{1}{\Sigma(\R)}
(2a^2Mr\sin^2\theta|\partial_t\psi|^2
+a^2|\partial_{\phi}\psi|^2
+4aMr\Re(Y\psi\partial_{\phi}\bar{\psi})).
\end{align}
The estimate \eqref{eq:conservedenergy:nearhorizon:esti} then follows easily.

In summary, from the conservation laws \eqref{eq:SWRWEnergyIdentity} and \eqref{eq:SWRWEnergyIdentityphi0} and the estimates \eqref{eq:conservedenergy:nearandhorizon:esti},  there exist constants $c>0$ and $C>0$ such that the following energy estimate for \eqref{eq:RewrittenFormofISWWEphi0OpeForm} and \eqref{eq:RewrittenFormofSWRWEOpeForm}:
\begin{align}\label{eq:SWRWEEnerEsti}
\hspace{4ex}&\hspace{-4ex}
c\int_{\Sigma_{\tau_2}\cap [r_0,\infty)}|\partial \psi|^2
+c\int_{\mathcal{H}^+(\tau_1,\tau_2)}|\partial_t \psi+\partial_{r^*}\psi|^2\notag\\
\leq{}&
\int_{\Sigma_{\tau_1}}e^k(\psi)
+\bigg|\int_{\mathcal{D}({\tau_1},
\tau_2)}\Re\left(\tfrac{F}{\Sigma} T\bar{\psi}\right)\bigg|\notag\\
&
+\frac{Ca^2}{M^2}\bigg(\int_{\Sigma_{\tau_2}
\cap [r_+,r_0]}|\partial \psi|^2
+{E}_{\mathcal{H}^+(\tau_1,\tau_2)}(\psi)\bigg).
\end{align}
Since it is clear which energy density we are referring to from the equation $\psi$ satisfies, we will from now on suppress the superscript $k$ in the energy density and simply write it as $e(\psi)$.

There exists an $\veps_0\geq 0$ and a nonnegative differentiable function $e_0(\veps_0)\sim \veps_0^2$ with $e_0(0)=0$ such that for all $|a|/M \leq \veps_0$ and any $\tilde{e}\geq e_0$, by adding to this energy estimate $\tilde{e}$ times the red-shift estimate in Proposition \ref{prop:RedShiftEstiInhomoSWRWE} for $\psi\in \{\phi^0_{+2},\phi^1_{+2},\phi^2_{+2},\phi^2_{-2}\}$ and in Proposition \ref{prop:RedShiftEstiInhomoSWRWEtildephi01} for $\widetilde{\phi^0}$ and $\widetilde{\phi^1}$, we obtain the following result analogous to \cite[Proposition 5.3.1]{dafermos2010decay}.
\begin{prop}\label{prop:energyestimateforT+eN}
Let $\psi=\phi_s^i$ and $F=F_s^i$ be as in the linear wave systems \eqref{eq:ReggeWheeler Phi^012KerrS2} and \eqref{eq:ReggeWheeler Phi^012KerrNegaS2}. Define
\begin{align}\label{def:tildepsi}
\tilde{\psi}=
\left\{
  \begin{array}{ll}
   \widetilde{\phi^j_{-2}}, & \quad \text{if $\psi=\phi^j_{-2}$ $(j=0,1)$;} \\
\psi, & \quad \text{if $\psi=\phi^0_{+2}, \phi^1_{+2},\phi^2_{+2}$ or $\phi^2_{-2}$.}\\
  \end{array}
  \right.
\end{align}
There exists an $\veps_0\geq0$ and a nonnegative differentiable function $e_0(\veps_0)\sim \veps_0^2$ with $e_0(0)=0$ such that for all $|a|/M \leq \veps_0$ and any $\tilde{e}\geq e_0$,
it holds true that
\begin{align}\label{eq:SWRWEEnerEstiII}
\hspace{4ex}&\hspace{-4ex}\int_{\Sigma_{\tau_2}}|e_{\tau_2}(\tilde{\psi})|
+\int_{\mathcal{H}^+(\tau_1,\tau_2)}|\partial_t \psi+\partial_{r^*}\psi|^2
+\tilde{e}E_{\tau_2}(\tilde{\psi})
+\tilde{e}{E}_{\mathcal{H}^+(\tau_1,\tau_2)}(\tilde{\psi})\notag\\
\lesssim {}&
\int_{\Sigma_{\tau_1}}|e_{\tau_1}(\tilde{\psi})|
+\tilde{e}E_{\tau_1}(\tilde{\psi})
+\tilde{e}\int_{\mathcal{D}(\tau_1,\tau_2)\cap[r_0,r_1]}|\partial\tilde{\psi}|^2
+\mathcal{E}_{F,\tilde{e},\tau_1,\tau_2}.
\end{align}
Here, for any $\tau_2>\tau_1\geq0$,
\begin{align}\label{def:bulktermtotalS2Kerr}
\mathcal{E}_{F,\tilde{e},\tau_1,\tau_2}={}&
\tilde{e}
\int_{\mathcal{D}({\tau_1},\tau_2)\cap[r_+,r_1]}
\mathbb{B}(\tilde{\psi},F)+\bigg|\int_{\mathcal{D}({\tau_1},
\tau_2)}\Re\left(\Sigma^{-1} F T\bar{\psi}\right)\bigg|,
\end{align}
and
\begin{align}
\mathbb{B}(\tilde{\psi},F)=
\left\{
  \begin{array}{ll}
    |F|^2, & \quad \text{for $\tilde{\psi}=\phi^1_{+2},\phi^2_{+2}$ or $\phi^2_{-2}$;} \\
|\phi_{+2}^1|^2, & \quad \text{for $\tilde{\psi}=\phi^0_{+2}$;}\\
|\widetilde{\phi^1}|^2, & \quad \text{for $\tilde{\psi}=\widetilde{\phi^0_{-2}}$;}\\
|\phi^2_{-2}|^2+ |\widetilde{\phi^0}|^2 +\tfrac{|a|}{M}|\partial \widetilde{\phi^0}|^2, & \quad \text{for $\tilde{\psi}=\widetilde{\phi^1_{-2}}$.}
  \end{array}
  \right.
\end{align}
\end{prop}

We here state a finite in time energy estimate for the ISWWE \eqref{eq:RewrittenFormofISWWEphi0OpeForm} and \eqref{eq:RewrittenFormofSWRWEOpeForm} based on the above discussions, which is an analogue of \cite[Proposition 5.3.2]{dafermos2010decay}.
\begin{prop}\label{prop:FiniteInTimeEnergyEstimateInhomoSWRWE}
\textbf{(Finite in time energy estimate)}.
Given an arbitrary $\epsilon>0$, there exists an $\veps_0>0$ depending on $\epsilon$ and a universal constant $C$ such that for $|a|/M\leq \veps_0$, $1\geq \tilde{e}\geq e_0(\veps_0)\sim \veps_0^2 $ and for any $\tau_1\geq 0$ and $\tau_1< \tau_2\leq \tau_1+\epsilon^{-1}$, the following results hold true:
For $\psi=\phi_s^i$ $(i=0,1,2)$, $\tilde{\psi}$ defined as in \eqref{def:tildepsi} and the corresponding inhomogeneous function $F=F_s^i$ in the linear wave systems \eqref{eq:ReggeWheeler Phi^012KerrS2} and \eqref{eq:ReggeWheeler Phi^012KerrNegaS2}, we have
\begin{align}\label{eq:FiniteInTimeEnergyEstimateInhomoSWRWETo T+eN}
\hspace{4ex}&\hspace{-4ex}\int_{\Sigma_{\tau_2}}|e(\tilde{\psi})|
+\tilde{e}E_{\tau_2}(\tilde{\psi})
+\tilde{e}{E}_{\mathcal{H}^+(\tau_1,\tau_2)}(\tilde{\psi})\notag\\
\leq  {} & (1+C\tilde{e})\bigg(\int_{\Sigma_{\tau_1}}|e(\tilde{\psi})|
+\tilde{e}E_{\tau_2}^{\text{total}}(s)\bigg)+C\mathcal{E}_{F,\tilde{e},\tau_1,\tau_2},
\end{align}
and, depending on the spin weight $s=\pm2$ of $\tilde{\psi}$,
\begin{align}\label{eq:FiniteInTimeEnergyEstimateInhomoSWRWEIINegaPosiS2}
\int_{\mathcal{D}({\tau_1},{\tau_2})\cap[r_0,r_1]}|\partial \tilde{\psi}|^2\leq CE_{\tau_1}^{\text{total}}(s).
\end{align}
Here, $\mathcal{E}_{F,\tilde{e},\tau_1,\tau_2}$ is already defined in \eqref{def:bulktermtotalS2Kerr} and, for any $\tau\geq 0$,
\begin{align}\label{def:EnergynormtotalS2Kerr}
E_{\tau}^{\text{total}}(s)=
\left\{
  \begin{array}{ll}
   {E}_{\tau}(r^{4-\delta}\phi^0_{+2})
   +{E}_{\tau}(r^{2-\delta}\phi^1_{+2})+{E}_{\tau}(\phi^2_{+2}), & \quad \text{for $s=+2$;} \\
{E}_{\tau}(\widetilde{\phi^0_{-2}})
+{E}_{\tau}(\widetilde{\phi^1_{-2}})
+{E}_{\tau}(\phi^2_{-2}), & \quad \text{for $s=-2$.}
  \end{array}
  \right.
\end{align}
\end{prop}
\begin{proof}
The first estimate follows from a combination of the previous prop with the second estimate \eqref{eq:FiniteInTimeEnergyEstimateInhomoSWRWEIINegaPosiS2}. The second estimate follows from the fact that it holds true for Schwarzschild case for all $\epsilon$ from the discussions in Sections \ref{sect:pfMainthmSchwS2} and \ref{sect:outlineproof} and the well-posedness property in Section \ref{sect:LWPandGlobalExistenceLinearWaveSystem} applied to the linear wave system of $\{r^{4-\delta}\phi_{+2}^0, r^{2-\delta}\phi_{+2}^1, \phi_{+2}^2\}$
or $\{\widetilde{\phi_{-2}^0},\widetilde{\phi_{-2}^1},\phi_{-2}^2\}$.
\qed
\end{proof}

\subsection{Separated Angular and Radial Equations}\label{sect:SeparateAngAndRadialEqs}
We consider in this subsection only the operators $\mathbf{L}_s^i$ acting on hypothetical \emph{integrable}\footnote{A solution to \eqref{eq:RewrittenFormofISWWEphi0OpeForm} or \eqref{eq:RewrittenFormofSWRWEOpeForm} is \emph{integrable} if for every integer $n\geq 0$, every multi-index $0\leq |i|\leq n$ and any $r'>r_+$, we have
\begin{align}\label{eq:integrabledefinition}
\sum_{0\leq |i|\leq n}\int_{\mathcal{D}(-\infty,\infty)\cap \{r=r'\}}(|\partial^i \psi|^2+|\partial^i F|^2)<\infty.
\end{align}} functions $\psi$ solving either \eqref{eq:RewrittenFormofISWWEphi0OpeForm} or
\eqref{eq:RewrittenFormofSWRWEOpeForm}, and in the next subsection, a cutoff in time will be applied to the solutions to each subequation in systems \eqref{eq:ReggeWheeler Phi^012KerrS2} and \eqref{eq:ReggeWheeler Phi^012KerrNegaS2} so as to ensure the integrability condition and allow for the separation introduced below.

If the solution $\psi$ to the equation \eqref{eq:RewrittenFormofSWRWEOpeForm} is {integrable}, we can write
\begin{align}\label{eq:fourierintime}
\psi=\frac{1}{\sqrt{2\pi}}\int_{-\infty}^{\infty}e^{-i\omega t}\psi_{\omega}(r,\theta,\phi)\di \omega,
\end{align}
where $\psi_{\omega}$ is defined as the Fourier transform of $\psi$:
\begin{equation}
\psi_{\omega}=\frac{1}{\sqrt{2\pi}}\int_{-\infty}^{\infty}e^{i\omega t}\psi(t,r,\theta,\phi)\di t.
\end{equation}
The equality \eqref{eq:fourierintime} can be interpreted in $L^2(\sin \theta \di \theta \di \phi \di t)$.
We further decompose $\psi_{\omega}$ in $L^2(\sin\theta \di \theta \di \phi)$ into
\begin{equation}
\psi_{\omega}=\sum_{m,\ell}\psi_{m\ell}^{(a\omega)} (r)Y_{ m\ell}^{s}(a\omega, \cos\theta)e^{im\phi}, \ \ m\in \mathbb{Z}.
\end{equation}
Here, for each $m$, $\{Y_{ m\ell}^{s}(a\omega, \cos\theta)\}_{\ell}$, with $\min{\{\ell\}}=\max\{|m|,|s|\}$, are the eigenfunctions of the self-adjoint operator
\begin{align}\label{eq:SWSHOpe Kerr}
\textbf{S}_m=\tfrac{1}{\sin\theta}\partial_{\theta}\sin\theta\partial_{\theta}
-\tfrac{m^2+2ms\cos \theta+s^2}{\sin^2 \theta}+a^2\omega^2 \cos^2 \theta-2a\omega s\cos \theta
\end{align}
on $L^2(\sin\theta \di \theta)$. These eigenfunctions, called \textquotedblleft{\emph{spin-weighted spheroidal harmonics},\textquotedblright}  form a complete orthonormal basis on $L^2(\sin\theta \di \theta)$ and have eigenvalues $\Lambda_{m\ell}^{(a\omega)}$ defined by
\begin{equation}\label{eq:SWSHOpeEq Kerr}
\textbf{S}_mY_{ m\ell}^{s}(a\omega, \cos\theta)=-\Lambda_{m\ell}^{(a\omega)}
Y_{ m\ell}^{s}(a\omega, \cos\theta).
\end{equation}
One could similarly define $F_{\omega}$ and $F_{m\ell}^{(a\omega)}$.

An integration by parts, together with a usage of the Plancherel lemma and the orthonormality property of the basis $\{Y_{ m\ell}^{s}(a\omega, \cos\theta)e^{im\phi}\}_{m\ell}$, gives
\begin{align}\label{eq:IdenOfEigenvaluesAndAnguDeriKerr1}
\hspace{4ex}&\hspace{-4ex}\int_{-\infty}^{+\infty}\sum_{m,\ell}\Lambda_{m\ell}^{(a\omega)}
|\psi_{m\ell}^{(a\omega)}|^2\di \omega\notag\\
=& \int_{-\infty}^{\infty}\int_{\mathbb{S}^2}
\Big\{|\partial_{\theta}\psi|^2+\left|\tfrac{\partial_{\phi}\psi+
is\cos\theta\psi}{\sin\theta}\right|^2
-|a\cos\theta\partial_t\psi+is\psi|^2
+2s^2|\psi|^2\Big\}\di \sigma_{\mathbb{S}^2}\di t.
\end{align}
The radial equation for $\psi_{m\ell}^{(a\omega)}$ is then
\begin{align}\label{eq:SWRWRadialEqKerr}
\Big\{\partial_r(\Delta\partial_r)
+(\underline{V})_{m\ell,1}^{(a\omega)}(r)\Big\}\psi_{m\ell}^{(a\omega)}
=F_{m\ell}^{(a\omega)},
\end{align}
with the potential
\begin{equation}
(\underline{V})_{m\ell,1}^{(a\omega)}(r)
=\tfrac{(r^2+a^2)^2\omega^2+a^2m^2
-4aMrm\omega}{\Delta}-\left(\lambda_{m\ell,1}^{(a\omega)}(r)
+a^2\omega^2\right).
\end{equation}
We utilized here a substitution of
\begin{equation}\label{eq:defOflambda}
\lambda_{m\ell,1}^{(a\omega)}(r)=\Lambda_{m\ell}^{(a\omega)}-\tfrac{s^2(2Mr-2a^2)}{r^2},
\end{equation}
by which the above radial equation \eqref{eq:SWRWRadialEqKerr} is the same as the radial equation \cite[Equation (33)]{dafermos2010decay}\footnote{The authors in \cite{dafermos2010decay} missed one term $-4aMrm\omega$ in the Equation $(33)$, but what is used thereafter is the Schr\"{o}dinger equation $(34)$ in Section $9$ which is correct. Therefore, the validity of the proof will not be influenced by the missing term.} for the scalar field.

One can obtain for \eqref{eq:RewrittenFormofISWWEphi0OpeForm} the same angular equation and the following radial equation after decomposition:
The radial equation for $\psi_{m\ell}^{(a\omega)}$ is
\begin{align}\label{eq:SWRWRadialEqKerrphi0}
\Big\{\partial_r(\Delta\partial_r)
+(\underline{V})_{m\ell,0}^{(a\omega)}(r)\Big\}\psi_{m\ell}^{(a\omega)}
=F_{m\ell}^{(a\omega)},
\end{align}
with the potential
\begin{equation}
(\underline{V})_{m\ell,0}^{(a\omega)}(r)
=\tfrac{(r^2+a^2)^2\omega^2+a^2m^2
-4aMrm\omega}{\Delta}-\left(\lambda_{m\ell,0}^{(a\omega)}(r)
+a^2\omega^2\right)
\end{equation}
and a substitution of
\begin{equation}\label{eq:defOflambdaphi0}
\lambda_{m\ell,0}^{(a\omega)}(r)=\Lambda_{m\ell}^{(a\omega)}
-\tfrac{s^2(\Delta+a^2)}{2r^2}.
\end{equation}
The above two radial equations can now be put into the following form
\begin{align}\label{eq:SWRWRadialEqKerrgeneral}
\Big\{\partial_r(\Delta\partial_r)
+(\underline{V})_{m\ell}^{(a\omega)}(r)\Big\}\psi_{m\ell}^{(a\omega)}
=F_{m\ell}^{(a\omega)},
\end{align}
with
\begin{equation}
(\underline{V})_{m\ell}^{(a\omega)}(r)
=\tfrac{(r^2+a^2)^2\omega^2+a^2m^2
-4aMrm\omega}{\Delta}-\left(\lambda_{m\ell}^{(a\omega)}(r)
+a^2\omega^2\right).
\end{equation}
Here, $\lambda_{m\ell}^{(a\omega)}(r)=\lambda_{m\ell,k}^{(a\omega)}(r)$ and $(\underline{V})_{m\ell}^{(a\omega)}(r)=(\underline{V})_{m\ell,k}^{(a\omega)}(r)$, and the value of $k=0,1$ depends on which of the two equations \eqref{eq:RewrittenFormofISWWEphi0OpeForm} and \eqref{eq:RewrittenFormofSWRWEOpeForm} $\psi$ satisfies.

We state here some basic identities for any $r>r_+$ from properties of Fourier transform and Plancherel lemma:
\begin{align*}
\int_{-\infty}^{\infty}\int_{0}^{2\pi}\int_{0}^{\pi}|\psi(t,r,
\theta,\phi)|^2\sin\theta \di \theta \di \phi \di t ={}&\int_{-\infty}^{\infty}\sum_{m,\ell}\left|\psi_{m\ell}^{(a\omega)}(r)
\right|^2\di \omega,\notag\\
\int_{-\infty}^{\infty}\int_{0}^{2\pi}\int_{0}^{\pi}\left|
\partial_{r^*} \psi(t,r,
\theta,\phi)\right|^2\sin\theta \di \theta \di \phi \di t ={}&\int_{-\infty}^{\infty}\sum_{m,\ell}\left|\partial_{r^*}\psi_{m\ell}^{(a\omega)}(r)
\right|^2\di \omega,\notag\\
\int_{-\infty}^{\infty}\int_{0}^{2\pi}\int_{0}^{\pi}\left|
\partial_t\psi(t,r,\theta,\phi)\right|^2\sin\theta \di \theta \di \phi \di t ={}&\int_{-\infty}^{\infty}\sum_{m,\ell}\omega^2\left|\psi_{m\ell}^{(a\omega)}(r)
\right|^2\di \omega.
\end{align*}

\subsection{Frequency Localised Multiplier Estimates}\label{sect:MoraEstiCurrentsKerr}

\subsubsection{Cutoff in Time}\label{sect:cutoffintime}
To justify the separation procedures in Section \ref{sect:SeparateAngAndRadialEqs}, the assumption that the solution $\psi(t,r,\theta,\phi)$ is integrable is sufficient, but this is \textit{a priori} unknown. Therefore, we apply cutoffs to the solution both to the future and to the past, and then do separation for the wave equation which the gained function satisfies.

Let $\chi_2(x)$ be a smooth cutoff function which equals $0$ for $x\leq 0$ and is identically $1$ when $x\geq 1$. Let $\tau> 2\varepsilon^{-1}$ be arbitrary with $\varepsilon>0$ to be chosen. Define
\begin{equation}
\chi=\chi_{\tau,\varepsilon}(t^*)=\chi_2(\varepsilon t^*)\chi_2(\varepsilon(\tau-t^*))
\end{equation}
and
\begin{equation}
\psi_{\chi}=\chi\psi
\end{equation}
in coordinate system $(t^*,r,\theta,\phi^*)$. The cutoff function $\psi_{\chi}$ is now a smooth function supported in $0\leq t^*\leq \tau$, and $\psi_{\chi}=\psi$ in $\varepsilon^{-1}\leq t^*\leq \tau -\varepsilon^{-1}$. Moreover, it satisfies the following wave equation
\begin{align}\label{eq:SWRWE CutoffVersion}
\mathbf{L}^k_s\psi_{\chi}&=\chi F+\Sigma\left(2\nabla^{\mu}\chi\nabla_{\mu}\psi
+\left(\Box_g\chi\right)\psi\right)-2isa\cos\theta\partial_t\chi\psi \notag\\
&\triangleq F_{\chi} .
\end{align}
The fact that the afore-defined $\chi$ is $\phi$-independent is utilized here in deriving this equation.

Note the fact that the functions $\psi_{\chi}$ and $F_{\chi}$ are compactly supported in $0\leq t^*\leq \tau$ at each fixed $r> r_+$, and the assumption that $\psi$ is a compactly supported smooth section solving one subequation of a linear wave system, hence $\psi_{\chi}$ is an integrable solution to \eqref{eq:SWRWE CutoffVersion} from Section \ref{sect:LWPandGlobalExistenceLinearWaveSystem}. In the following discussions, we apply the mode decompositions in Section \ref{sect:SeparateAngAndRadialEqs} to $\psi_{\chi}$ and $F_{\chi}$, and separate the wave equation \eqref{eq:SWRWE CutoffVersion} into the angular equation \eqref{eq:SWSHOpeEq Kerr} and radial equation \eqref{eq:SWRWRadialEqKerrgeneral}, but with $R_{m\ell}^{(a\omega)}=(\psi_{\chi})_{m\ell}^{(a\omega)}$ and $\left(F_{\chi}\right)_{m\ell}^{(a\omega)}$ in place of $\psi_{m\ell}^{(a\omega)}$ and $F_{m\ell}^{(a\omega)}$ respectively.

Before introducing the microlocal currents, we give some estimates for the inhomogeneous term $F_{\chi}$ here. Due to the fact that $\nabla \chi$ and $\Box_g \chi$ are supported in
\begin{equation}
\left\{0\leq t^* \leq \varepsilon^{-1}\right\}\cup
\left\{\tau -\varepsilon^{-1}\leq t^*\leq \tau\right\},
\end{equation}
one obtains in the coordinate system $\left(t^*, r,\theta,\phi^*\right)$ that
\begin{subequations}\label{eq:PropertyofCutoffchi}
\begin{align}
|\partial_{t^*}\chi|\leq C\varepsilon, \quad &\left|\Box_g\chi\right|\leq C\varepsilon^2,\\
\left|\nabla^{\mu}\chi \nabla_{\mu}\psi\right|^2
+\left|\tfrac{ias\cos\theta\partial_t\chi\psi}{\Sigma}\right|^2&\leq C\varepsilon^2\left(|\partial\psi|^2
+a^2 M^{-2}\left|r^{-1}{\psi}\right|^2\right).
\end{align}
\end{subequations}

\subsubsection{Currents in Phase Space}
In what follows, we will suppress the dependence on $a$, $\omega$, $m$ and $\ell$ of the functions $R_{m\ell}^{(a\omega)}(r)$, $F_{m\ell}^{(a\omega)}(r)$, $\Lambda_{m\ell}^{(a\omega)}$, $\lambda_{m\ell}^{(a\omega)}(r)$, $(\underline{V})_{m\ell}^{(a\omega)}(r)$ and other functions defined by them. When there is no confusion, the dependence on $r$ may always be implicit (except for the radial part $R(r)$ to avoid misunderstanding with the radius parameter $R$). Thus we will write them as $R(r)$, $F$, $\Lambda$, $\lambda$ and $\underline{V}$.

We transform the radial equation \eqref{eq:SWRWRadialEqKerrgeneral} into a Schr\"{o}dinger form, which will be of great use to define the microlocal currents below, by setting
\begin{equation}
\label{def:u(r)andH(r)}
u(r)=\sqrt{r^2+a^2}R(r), \ \ \ \
H(r)=\tfrac{\Delta F_{\chi}(r)}{\left(r^2+a^2\right)^{3/2}}.
\end{equation}
The Schr\"{o}dinger equation for $u(r)$ reads after some calculations
\begin{equation}\label{eq:eqofu}
u(r)''+\left(\omega^2-V(r)\right)u(r)=H(r),
\end{equation}
where
\begin{align}\label{eq:SepaRadiSchroFormPoential}
V= &\omega^2-\tfrac{\KDelta}{(r^2+a^2)^2} \underline{V} + \tfrac{1}{(r^2+a^2)} \tfrac{d^2}{dr^{*2}} (r^2+a^2)^{1/2}\notag\\
=&\tfrac{4Mram\omega-a^2m^2+\Delta\left(\lambda+a^2\omega^2
\right)}{(r^2+a^2)^2}+\tfrac{\Delta}{(r^2+a^2)^4} \left( a^2\Delta + 2Mr(r^2-a^2) \right ),
\end{align}
and a prime $'$ denotes a partial derivative with respect to $r^*$ in tortoise coordinates.

Given any real, smooth functions $y,h$ and $f$, define the microlocal currents
\begin{subequations}\label{eq:currents}
\begin{align}
\label{eq:currenty(3)}
  Q^y&=y\left(\left|u'\right|^2+\left(\omega^2-V\right)|u|^2\right),\\
\label{eq:currenth(3)modified}
  Q^h&=h\Re\left(u'\overline{u}\right)-\tfrac{1}{2}h'|u|^2,\\
Q^f&=Q^{h=f'}+Q^{y=f}=f'\Re\left(u'\overline{u}\right)
-\left(\tfrac{1}{2}f''-f \left(\omega^2-V\right)\right)|u|^2+f\left|u'\right|^2.
\label{eq:currentf(3)}
\end{align}
\end{subequations}
The currents $Q^y$ and $Q^h$ are constructed via multiplying the equation \eqref{eq:eqofu} by $2y\bar{u}'$ and $h\bar{u}$ respectively.
The derivatives of the above currents are
\begin{subequations}\label{eq:currentsderi}
\begin{align}
\left(Q^y\right)'&=y'\left(\left|u'\right|^2+
\left(\omega^2-V\right)|u|^2\right)-yV'|u|^2
+2y\Re\left(u'\overline{H}\right),
\label{eq:Qyderi}\\
\left(Q^h\right)'&=h\left(\left|u'\right|^2+
\left(V-\omega^2\right)|u|^2\right)-\tfrac{1}{2}h''|u|^2
+h\Re\left(u\overline{H}\right),
\label{eq:Qhderi}\\
\left(Q^f\right)'&=2f'\left|u'\right|^2-fV'|u|^2
-\tfrac{1}{2}f'''|u|^2+
\Re\left(2fu'\overline{H}+f'u\overline{H}\right).
\label{eq:Qfderi}
\end{align}
\end{subequations}

\subsubsection{Frequency Regimes}

We start to define the separated frequency regimes.  Let $\omega_l$ and $\lambda_l$ be (potentially large) parameters and $\lambda_s$ be a (potentially small) parameter, all to be determined in the proof below. The frequency space is divided into
\begin{itemize}
\item $\mathcal{F}_{T}=\left\{(\omega, m, \ell): |\omega|\geq \omega_l, \lambda<\lambda_s\omega^2\right\}$;
\item $\mathcal{F}_{Tr}=\left\{(\omega, m, \ell): |\omega|\geq \omega_l, \lambda\geq\lambda_s\omega^2\right\}$;
\item $\mathcal{F}_{A}=\left\{(\omega, m, \ell): |\omega|\leq \omega_l, \Lambda>\lambda_l\right\}$;
\item $\mathcal{F}_{B}=\left\{(\omega, m, \ell): |\omega|\leq \omega_l, \Lambda\leq \lambda_l\right\}$.
\end{itemize}
We fix $R_0$ as in Section \ref{sect:MorawetzLarger} and an arbitrary $r_c \in (2M, r_0)$, with $r_0$ fixed in Section \ref{sect:Redshift}.

\begin{lemma}
\label{lem:PotentialDereasingForRadiusLarge}
For all $|a|<M$ and all frequency triplets $(\omega, m, \ell)$,
\begin{align}
 \label{eq:eigenesti:largeradius}
 \Lambda+3a^2\omega^2 \geq {}  \frac{2}{3}\max\{m^2+|m|, s^2+|s|\},
 \end{align}
and there exists a sufficiently large $R_5\geq R_0+M$ and two constants $c=c(M)>0$ and $C=C(M)>0$ such that for all $r\geq R_5$,
\begin{align}
\label{eq:estimatederiVnearinf}
V'(r)<-cr^{-3}(\Lambda+4a^2\omega^2+m^2+1)+Cr^{-3}a^2\omega^2.
\end{align}
\end{lemma}

\begin{proof}
Let
\begin{subequations}\label{eq:decompofpotentialV}
\begin{align}
V={}&V_m+V_e,\\
V_m={}&\tfrac{\Delta(\Lambda+a^2\omega^2)}{(r^2+a^2)^2}.
\end{align}
\end{subequations}
For $V=V_1$ which corresponds to the case $(\underline{V})_{m\ell}^{(a\omega)}(r)=(\underline{V})_{m\ell,1}^{(a\omega)}(r)$, we calculate
\begin{align}
\label{eq:decompofpotentialV1e}
V_e={}&\tfrac{4Mram\omega-a^2m^2}{(r^2+a^2)^2}
-\tfrac{6(Mr-a^2)\Delta}{r^2(r^2+a^2)^2}\notag\\
&
+\tfrac{a^2\Delta}{r^2(r^2+a^2)^4} (3r^4-8Mr^3+5a^2r^2-2a^2Mr+2a^4),
\end{align}
and for $V=V_0$ corresponding to $(\underline{V})_{m\ell}^{(a\omega)}(r)=(\underline{V})_{m\ell,0}^{(a\omega)}(r)$,
\begin{align}\label{eq:decompofpotentialVother}
V_e={}&\tfrac{4Mram\omega-a^2m^2}{(r^2+a^2)^2}
-\tfrac{2(r^2-3Mr+3a^2)\Delta}{r^2(r^2+a^2)^2}\notag\\
&
+\tfrac{a^2\Delta}{r^2(r^2+a^2)^4} (3r^4-8Mr^3+5a^2r^2-2a^2Mr+2a^4).
\end{align}
One finds for all $r\in [r_+,\infty)$,
\begin{align}
\label{eq:derivativeofVm}
V_m'={}&-2\Delta(r^2+a^2)^{-4}(\Lambda+a^2\omega^2)
(r^3-3Mr^2+a^2r+a^2M),\\
\label{eq:generalformofderiofVe}
V_e'\leq{}&4\Delta (r^2+a^2)^{-1}r^{-3}+C\Delta{r^{-2}}\left(r^{-4} |am\omega|+r^{-5}a^2m^2+r^{-4}\right).
\end{align}
One adds $-3a^2\omega^2Y_{ m\ell}^{s}(a\omega, \cos\theta)$ to both sides of \eqref{eq:SWSHOpeEq Kerr} and obtains
\begin{align}
\hspace{4ex}&\hspace{-4ex}(\tfrac{1}{\sin\theta}\partial_{\theta}\sin\theta\partial_{\theta}
-\tfrac{m^2+2ms\cos \theta+s^2}{\sin^2 \theta}+a^2\omega^2 \cos^2 \theta-2a\omega s\cos \theta -3a^2\omega^2)Y_{ m\ell}^{s}\notag\\
={}&-(\Lambda+3a^2\omega^2)
Y_{ m\ell}^{s}.
\end{align}
Rewrite the LHS as
\begin{align}
&\left(\tfrac{1}{\sin\theta}\partial_{\theta}\sin\theta\partial_{\theta}
-\tfrac{m^2+2ms\cos \theta+s^2}{\sin^2 \theta}\right)Y_{ m\ell}^{s}\notag\\
&-\left(a^2\omega^2 \sin^2 \theta+\tfrac{1}{2}(2a\omega  +s\cos \theta)^2 -\tfrac{1}{2}s^2\cos^2\theta\right)Y_{ m\ell}^{s}.
\end{align}
By Remark \ref{rem:EigenvalueSpinWeightedAngular}, the operator at the first line has eigenvalues not greater than $-\max\{m^2+|m|, s^2+|s|\}$, and the eigenvalues of the operator at the second line are less than or equal to $\frac{1}{2}s^2$, hence
\begin{align}
\label{eq:estieigenvaluepositiveplus}
\Lambda + 3a^2\omega^2 \geq \max\{m^2+|m|, s^2+|s|\}-\frac{1}{2}s^2 \geq \frac{2}{3}\max\{m^2+|m|, s^2+|s|\}.
\end{align}
From \eqref{eq:derivativeofVm}, given any $\eps>0$, $V_m'\leq -(2-\eps)r^{-3}(\Lambda+a^2\omega^2)$ for $r\gg 4M$ large.
Together with \eqref{eq:generalformofderiofVe}, this implies
\begin{align}
V'-6r^{-3}a^2\omega^2\leq {}& (-(2-\eps)(\Lambda+3a^2\omega^2)+4)r^{-3}-2r^{-3}a^2\omega^2 \notag\\
& +Cr^{-4}(|am\omega|+a^2m^2+1).
\end{align}
For sufficiently small $\eps$, there exists a $0<c_0<1$ such that the coefficient of the first term on the RHS is negative and bounded above by $-c_0r^{-3}(\Lambda +3a^2\omega^2+m^2+s^2)$. For $r\geq R_5$ large enough, the second line of the RHS is bounded by $\frac{c_0}{2}r^{-3}(a^2\omega^2+m^2+1)$ using the Cauchy--Schwarz inequality, hence this completes the proof.  \qed
\end{proof}

We shall obtain a phase-space version of Morawetz estimate by choosing different functions $y$, $h$ and $f$ in each of the four frequency regimes in Sections \ref{sect:timedominatedfreq}--\ref{sect:boundfreq}. The proofs in Sections \ref{sect:timedominatedfreq}--\ref{sect:angulardominatedfreq} follow from the discussions in \cite[Section 9.4--9.6]{dafermos2010decay}, nevertheless, we present here the entire proof for completeness.

\subsubsection{$\mathcal{F}_{T}$ Regime (Time-dominated Frequency Regime)}\label{sect:timedominatedfreq}
In this section, $\lambda_s$ will be fixed and $\omega_l$ will be left unspecified until Section \ref{sect:trapfreq}.
For $|a|/M\leq \veps_0\ll 1$, by choosing small enough $\lambda_s$ and large enough $\omega_l$, there exists a constant $c_0<1$ such that we have in $\mathcal{F}_{T}$ that
\begin{equation}
\omega^2-V\geq \tfrac{1-c_0}{2}\omega^2\ \text{in}\ [r_+,\infty).
\end{equation}
As to the potential $V$, apart from the fact in Lemma \ref{lem:PotentialDereasingForRadiusLarge}, we have for all $r^*$ that
\begin{equation}
\left|V'\right|\leq C\Delta/r^{5}\left((\Lambda+3a_0^2\omega^2)+1\right).
\end{equation}
Choose a function $y$ to satisfy the following properties:
\begin{enumerate}
  \item $y\geq 0$, $y'\geq c\Delta/r^4$ in $(r_+, R_5-M]$,
  \item $y\geq 0$, $y'\geq 0$ in $[R_5-M, R_5]$,
  \item $y=1$ in $[R_5,\infty)$.
\end{enumerate}
In view of the above properties,
$\left(Q^y\right)'\geq 2y\Re\left(u'\overline{H}\right)-Ca^2r^{-3}\omega^2$ for $r\geq R_5$,
$\left(Q^y\right)'\geq 2y\Re\left(u'\overline{H}\right)$ for $r\leq r_c$,
and in the region $r\in (r_c, R_5-M)$,  we have
\begin{align}
y'(\omega^2-V)-yV'\geq{}& \tfrac{1-c_0}{2}y'\omega^2 -Cy\Delta/r^5\left((\lambda+a_0^2\omega^2)+1\right)\notag\\
\geq{}&\left(c\tfrac{1-c_0}{2} - Cr^{-1}(\lambda_s+a_0^2+\omega_l^{-2})\right)\Delta/r^4\omega^2.
\end{align}
By taking both $\lambda_s$ and $\veps_0$ sufficiently small and for $\omega_l$ sufficiently large, the right-hand side is larger than $c\Delta/r^4\omega^2$.
Hence, integrating \eqref{eq:Qyderi} over $[r_{-\infty}^*, r_{\infty}^*]$ gives
\begin{lemma}\label{lem:TimeDominatedFreqEstimate}
Let $r_{\infty}^*>R_5^*> (R_0+M)^*$ and $r_{-\infty}^*<r_c^*$ be arbitrary. There exists small enough $\lambda_s$ and large enough $\omega_l$ such that we have in $\mathcal{F}_T$ frequency regime the following estimate for sufficiently small $\veps_0$:
\begin{align}\label{eq:TimeDominatedFreqEstimate}
\hspace{4ex}&\hspace{-4ex}c\int_{r_c^*}^{R_0^*}\frac{\Delta}{r^{4}}
\left(\left|u'\right|^2
+\left(\omega^2+(\Lambda+4a^2\omega^2)
+1\right)|u|^2\right)\notag\\
\leq {}& \int_{r_{-\infty}^*}^{r_{\infty}^*}2y\Re\left(u'\overline{H}\right)
+Q^{y}\left(r_{\infty}^*\right)-Q^y\left(r_{-\infty}^*\right)
+C\int_{R_5^*}^{r_{\infty}^*}a^2r^{-3}\omega^2|u|^2.
\end{align}
\end{lemma}

\subsubsection{$\mathcal{F}_{Tr}$ Regime (Trapped Frequency Regime)}\label{sect:trapfreq}

We have fixed $\lambda_s$ (which is sufficiently small) as in Section \ref{sect:timedominatedfreq}, and we will fix $\omega_l$ here.
This is the only frequency regime where trapping phenomenon could happen. Recall from \eqref{eq:derivativeofVm} that
\begin{equation}
V_m'=\tfrac{-2\Delta}{(r^2+a^2)^4}(\Lambda+a^2\omega^2)
(r^3-3Mr^2+a^2r+a^2M),
\end{equation}
and the polynomial $P_3(r)=r^3-3Mr^2+a^2r+a^2M$ has a unique zero point $r_{a,M}$ in $r\in [r_+,\infty)$ which satisfies moreover $|r_{a,M}-3M|\leq Ca^2$. Meanwhile, the estimate \eqref{eq:eigenesti:largeradius} and the requirements of $|\omega|\geq \omega_1$ and $\lambda\geq\lambda_2\omega^2$ in this frequency regime imply that for large enough $\omega_1$ and sufficiently small $|a|/M\leq \veps_0$,
$\Lambda+a^2\omega^2\geq c(\omega^2+\max\{m^2+|m|, s^2+|s|\})$.
On the other hand,
\begin{equation}
r^3|V_e'|\leq C\Delta/r^2(a^2 m^2+|am\omega|+1).
\end{equation}
Therefore, given any small $\varepsilon_4>0$, we have for sufficiently large $\omega_l$ (depending on the $\lambda_s$) and sufficiently small $\veps_0$ that $V'$ has no zeros outside the region $[r_{a,M}-\varepsilon_4, r_{a,M}+\varepsilon_4]$. In fact, $V'$ has a unique simple zero in this neighborhood. This can be seen as follows.  Note from \eqref{eq:derivativeofVm}, \eqref{eq:decompofpotentialV1e} and
\eqref{eq:decompofpotentialVother} that for sufficiently small $\veps_0$,
\begin{subequations}
\begin{align}
\left(\tfrac{(r^2+a^2)^4}{\Delta r^2}V_m'\right)'(r_{a,M})={}&-\tfrac{2\Delta(r_{a,M})(\Lambda+a^2\omega^2)}
{r_{a,M}^2+a^2}
(1-\tfrac{a^2}{r_{a,M}^2}-\tfrac{2a^2M}{r_{a,M}^3})\notag\\
\leq{}& -c(\Lambda+a^2\omega^2),\\
\label{eq:Vesecondderi}
\left|\left(\tfrac{(r^2+a^2)^4}{\Delta r^2}V_e'\right)'\right|\lesssim{}&\Delta r^{-2}(a^2m^2+|am\omega|+1).
\end{align}
\end{subequations}
For $\omega_l$ sufficiently large and $\veps_0$ sufficiently small, $\Lambda+a^2\omega^2$ will be much bigger than the RHS of \eqref{eq:Vesecondderi}. This implies that in the small region $[r_{a,M}-\varepsilon_4, r_{a,M}+\varepsilon_4]$, we have
\begin{align}
\left(\tfrac{(r^2+a^2)^4}{\Delta r^2}V'\right)'
\leq{}& -c(\Lambda+a^2\omega^2).
\end{align}
Therefore, this proves that for any $\veps_4>0$, we have for sufficiently large $\omega_l$  and sufficiently small $\veps_0$ that there is a unique simple zero, which we denote by $r_{m\ell}^{(a\omega)}$, in $[r_+,\infty)$ for any $(\omega, m, \ell)\in \mathcal{F}_{Tr}$, and $r_{m\ell}^{(a\omega)}\in [r_{a,M}-\varepsilon_4, r_{a,M}+\varepsilon_4]$. Moreover, for $r\leq r_{m\ell}^{(a\omega)}$,
\begin{align}
V'\geq -c(r-r_{m\ell}^{(a\omega)})\Delta/r^6(\Lambda +a^2\omega^2),
\end{align}
and for $r\geq r_{m\ell}^{(a\omega)}$,
\begin{align}
V'\leq -c(r-r_{m\ell}^{(a\omega)})\Delta/r^6(\Lambda +a^2\omega^2).
\end{align}

Choose a function $f$ associated with $Q^f$ current to satisfy the following properties:
\begin{enumerate}
  \item \label{functionf:point1}  $\lim\limits_{r^*\to -\infty}f=-1$, $f=1$ for some large $R_4\geq R_0$, and $f(r_c)=-1/2$, $f(r_{m\ell}^{(a\omega)}-M/10)=-1/4$,
      $f(r_{m\ell}^{(a\omega)})=0$,
      $f(r_{m\ell}^{(a\omega)}+M/10)=1/4$ and $f(R_0)=1/2$,
  \item \label{functionf:point2} $f'\geq 0$ for all $r^*$, and $f'\geq c\Delta/r^{4}$ for $r_c \leq r\leq R_0$,
  \item \label{functionf:point3} $f'''\leq -c$ for $r_{m\ell}^{(a\omega)}-M/10\leq r\leq r_{m\ell}^{(a\omega)}+M/10$, and $|f'''|\leq C\Delta/r^5$ elsewhere.
\end{enumerate}
Such a function $f$ can be manifestly constructed.
Upon making such a choice of function $f$, it can be seen from the above properties of $V'$ and properties of the chosen function $f$ that for $\omega_l$ large enough, we have for all $r^*$ that
$$
-fV'-\tfrac{1}{2}f'''\geq \big(c+c(\omega_l)((\Lambda+a^2\omega^2)+\omega^2)
      (r-r_{m\ell}^{(a\omega)})^2\big)\Delta/r^7.
$$
By integrating \eqref{eq:Qfderi} over $[r_{-\infty}^*, r_{\infty}^*]$, we arrive at the following conclusion.
\begin{lemma}\label{lem:TrappedDominatedFreqEstimate}
Let $r_{\infty}^*>R_4^*\geq R_0^*$ and $r_{-\infty}^*<r_c^*$ be arbitrary. There exists a large $\omega_l$ such that in $\mathcal{F}_{Tr}$ frequency regime, it holds true for sufficiently small $\veps_0$ that
\begin{align}\label{eq:TrappedFreqEstimate}
\hspace{4ex}&\hspace{-4ex}\int_{r_c^*}^{R_0^*}
\frac{\Delta}{r^4}\left(c\left|u'\right|^2
+r^{-1}\left[c+c(\omega_l)(1-r^{-1}r_{m\ell}^{(a\omega)})^2
\left(\omega^2+(\Lambda+4a^2\omega^2)
\right)\right]|u|^2\right)\notag\\
\leq & \int_{r_{-\infty}^*}^{r_{\infty}^*}\left(2f\Re\left(u'\overline{H}\right)
+f'\Re\left(u\overline{H}\right)\right)
+Q^{f}\left(r_{\infty}^*\right)-Q^f\left(r_{-\infty}^*\right).
\end{align}
\end{lemma}

\subsubsection{$\mathcal{F}_{A}$ Regime (Angular-dominated Frequency Regime)}\label{sect:angulardominatedfreq}

Here, we fix $\omega_l$ as in Section \ref{sect:trapfreq}, and will choose $\lambda_l$ to be sufficiently large. The general idea is as follows. In this regime, for sufficiently small $|a|/M\leq \veps_0$, the zero points of $V'(r)$ in $[r_+,\infty)$ are located in a small neighborhood of $r=3M$.
The $Q^f$ current is utilized to achieve the positivity of the bulk term outside this small neighborhood, while in this small neighborhood, $hV|u|^2$ in $(Q^h)'$ with $h(r)$ being a positive constant is used to dominate the potentially negative bulk term in $(Q^f)'$.

Similarly to the discussions in Section \ref{sect:trapfreq}, one calculates
\begin{equation}
V_m'=\tfrac{-2\Delta}{(r^2+a^2)^4}(\Lambda+a^2\omega^2)
(r^3-3Mr^2+a^2r+a^2M),
\end{equation}
and the polynomial $r^3-3Mr^2+a^2r+a^2M$ has a unique zero point $r_{a,M}$ in $r\in [r_+,\infty)$ satisfing $|r_{a,M}-3M|\leq Ca^2$. In this frequency regime, it is clear that for $\lambda_l$ sufficiently large,  $\Lambda+a^2\omega^2>0$.  On the other hand,
\begin{equation}
r^3|V_e'|\leq C\Delta/r^2(a^2 m^2+|am\omega|+1).
\end{equation}
Therefore, given any small $\varepsilon_4>0$, we have for sufficiently small $\veps_0$ and sufficiently large $\lambda_l$ (depending on the $\omega_l$) that $V'$ has no zeros outside the region $[r_{a,M}-\varepsilon_4, r_{a,M}+\varepsilon_4]$.

Given $\veps_0\ll 1$, choose constants
$$r_+<r_c<r_{left1}<r_{left2}<r_{a,M}-\varepsilon_4 < r_{a,M}+\varepsilon_4 <r_{right1} <r_{right2} <R_0 <\infty.$$
For sufficiently small $\varepsilon_4$ and $\veps_0$, there exist constants $c_1, c_2>0$ and functions $f$ and $h$ satisfying the following conditions:
\begin{enumerate}
\item $f'\geq 0$ for $r>r_+$, $f'\geq c_1$ for $r_c\leq r\leq R_0$,
\item $f\sim -1 +(r-r_+)$ near the horizon, $f\leq -\frac{1}{2}$ for $r\leq r_{left2}$, $f\geq \frac{1}{2}$ for $r\geq r_{right1}$ and $f=1$ for $r\geq R_0+1$ such that $fV'<0$ in $(r_+, r_{a,M}-\varepsilon_4)\cup (r_{a,M}+\varepsilon_4, \infty)$ and $fV'\leq c_2 V$ in $[r_{a,M}-\varepsilon_4, r_{a,M}+\varepsilon_4]$,
\item $h=0$ in $[r_+, r_{left1}]\cup [r_{right2}, \infty)$ and $h=2c_2$ in $[r_{left2},r_{right1}]$.
\end{enumerate}
We then take $\lambda_l$ sufficiently large such that for $r\in [r_+, r_{a,M}-\varepsilon_4]\cup [r_{a,M}+\varepsilon_4, \infty)$,
\begin{align}
-\frac{1}{2}fV'-\frac{1}{2}f'''-\frac{1}{2}h''\geq 0,
\end{align}
and for $r\in[r_{a,M}-\varepsilon_4,r_{a,M}+\varepsilon_4]$,
\begin{align}
h(V-\omega^2)-\frac{1}{2}h''-fV'-\frac{1}{2}f'''
\geq{}&c_2V-2c_2\omega^2-\frac{1}{2}f'''
\geq{}cc_2\Lambda.
\end{align}
Therefore, by integrating over $r_{-\infty}^*\leq r^*\leq r_{\infty}^*$, we conclude
\begin{lemma}\label{lem:AngularDominatedFreqEstimate}
Let $r_{\infty}^*>R_0^*$ and $r_{-\infty}^*<r_c^*$ be arbitrary. Fix $\omega_l$ as in Section \ref{sect:trapfreq}. There exists a large $\lambda_l$ such that in $\mathcal{F}_A$ frequency regime, we have for sufficiently small $\veps_0$ the following estimate
\begin{align}\label{eq:AngularDominatedFreqEstimate}
\hspace{4ex}&\hspace{-4ex}c\int_{r_c^*}^{R_0^*}\left(\left|u'\right|^2+\Delta/ r^{5}\left(\omega^2+(\Lambda+4a^2\omega^2)
+1\right)|u|^2\right)\notag\\
\leq & \int_{r_{-\infty}^*}^{r_{\infty}^*}\left(2f\Re\left(u'\overline{H}\right)
+\left(f'+h\right)\Re\left(u\overline{H}\right)\right)
+Q^{f}\left(r_{\infty}^*\right)-Q^f\left(r_{-\infty}^*\right).
\end{align}
\end{lemma}

\subsubsection{$\mathcal{F}_{B}$ regime (bounded frequency regime)}\label{sect:boundfreq}
Fix $\omega_l$ and $\lambda_l$ as above. This is a compact frequency regime. Since the Morawetz estimates in the Schwarzschild case have been proved in Section \ref{proof of integrated decay}, for sufficiently small $|a|/M\leq \veps_0$, the Morawetz estimates in the phase space can be obtained by a stability argument. We shall show this below by choosing an appropriate function $f$ in the current $Q^f$.

In this regime, a key fact is that the minimum value of eigenvalues $\Lambda$ for the separated angular equation \eqref{eq:SWSHOpeEq Kerr} is close to $\max\{s^2+|s|,m^2+|m|\}$ due to smallness of $|a\omega|$. More explicitly, from \eqref{eq:SWSHOpeEq Kerr},
\begin{align}
\Lambda +a^2\omega^2\geq{}& \max\{s^2+|s|,m^2+|m|\}-2|a\omega s|
\notag\\
\geq {}&\max\{6,m^2+|m|\}-1-4a^2\omega_l^2.
\end{align}
 This gives directly that
\begin{align}
m^2\leq C(\lambda_l+a^2\omega_l^2).
\end{align}

We take
\begin{align}\label{eq:boundfreq:fchoice}
f=\frac{r^3-3Mr^2+a^2r+a^2M}{r^3}
 \end{align}
 in the current $Q^f$ and find for sufficiently small $|a|/M$,
\begin{equation}
f'=\frac{\Delta}{r^4(r^2+a^2)}(3Mr^2-2a^2r-3a^2M)\geq\frac{cM\Delta}{r^4}.
\end{equation}
Moreover, $f$ vanishes at a unique point $r_{a,M}$, which is already defined in Section \ref{sect:angulardominatedfreq}, with $|r_{a,M}-3M|\leq Ca^2$ for small enough $|a|/M$.

To obtain a Morawetz estimate in the phase space, we need to verify $-fV'-\frac{1}{2}f'''\geq 0$. Consider the case that $V=V_1$. The derivative $V_m'$ is as in \eqref{eq:derivativeofVm}, and note from \eqref{eq:decompofpotentialV1e} that
$$V_e=
-6Mr^{-1}(r^2+a^2)^{-2}\Delta+O(r^{-3})Mam\omega
+O(r^{-4})a^2m^2+O(r^{-4})a^2,$$
hence
\begin{subequations}
\begin{align}
&\big|V_e'-\tfrac{6Mr\Delta(3r-8M)}{(r^2+a^2)^4}\big|\leq {}C\Delta(Mr^{-6}|am\omega|+r^{-7}a^2m^2+r^{-7}a^2),\\
&\left|-\frac{1}{2}f'''+\frac{3M\Delta}{(r^2+a^2)^4}(3r^2-20Mr+30M^2)\right|
\leq Ca^2\Delta r^{-7}.
\end{align}
\end{subequations}
This implies
\begin{align}
\label{eq:boundedfreq:posibulk11}
-fV'-\frac{1}{2}f'''
\geq {}&-f\bigg(V_0'+\frac{6Mr\Delta(3r-8M)}{(r^2+a^2)^4}\bigg)
-\frac{3M\Delta(3r^2-20Mr+30M^2)}{(r^2+a^2)^4}\notag\\
&
-\frac{C\Delta}{r^2+a^2}(Mr^{-4}|am\omega|+a^2r^{-5}(m^2+1)).
\end{align}
For the RHS of the first line, it equals
\begin{align}
\label{eq:boundedfreq:posibulk12}
\hspace{4ex}&\hspace{-4ex}\frac{2\Delta(r^3-3Mr^2+a^2r+a^2M)^2
(\Lambda+a^2\omega^2)}{r^3(r^2+a^2)^4}
-\frac{3M\Delta(3r^2-20Mr+30M^2)}{(r^2+a^2)^4}\notag\\
\hspace{4ex}&\hspace{-4ex}
-\frac{6Mr\Delta(3r-8M)(r^3-3Mr^2+a^2r+a^2M)}{r^3(r^2+a^2)^4}
\notag\\
\geq {}&\frac{\Delta}{r^3(r^2+a^2)^4}
\left[2(r-3M)^2r^4(6-\varepsilon)-3Mr^3(9r^2-54Mr+78M^2)\right]
\notag\\
&-\frac{C\Delta}{r^2+a^2}r^{-3}\varepsilon^{-1}a^2\omega_l^2
-Ca^2r^{-5}.
\end{align}
From the above two estimates, to show that there exists a constant $c>0$ such that $-fV'-\frac{1}{2}f'''\geq \frac{c\Delta r^3}{(r^2+a^2)^4}$, it is enough to prove there exists a positive constant $c_3>0$ such that
\begin{align}
\label{eq:boundfreq:poly}
4r(r-3M)^2 -M(9r^2-54Mr+78M^2) \geq c_3r^3,
\end{align}
since then the first line on the RHS of \eqref{eq:boundedfreq:posibulk12} will dominate over $\frac{c_3}{2}\frac{\Delta r^3}{(r^2+a^2)^4}$ by taking $\varepsilon=c_3/2$, which will in turn dominate over the sum of the second line of \eqref{eq:boundedfreq:posibulk11} and the last line of \eqref{eq:boundedfreq:posibulk12} by taking $|a|/M\leq \veps_0$ sufficiently small. The LHS of \eqref{eq:boundfreq:poly} is $4r^3-33Mr^2+90M^2r-78M^3$, and it has been shown in Section \ref{proof of integrated decay} that this polynomial has no zero points in $[r_+,\infty)$, hence the relation \eqref{eq:boundfreq:poly} follows.

In the case that $V=V_0$, one follows the above argument and obtains
\begin{align}
\hspace{4ex}&\hspace{-4ex}-fV'-\frac{1}{2}f'''\notag\\
\geq {}&-f\bigg(V_m'+\frac{r\Delta(4r^2-30Mr+48M^2)}{(r^2+a^2)^4}\bigg)
-\frac{3M\Delta}{(r^2+a^2)^4}(3r^2-20Mr+30M^2)\notag\\
&
-\frac{C\Delta}{r^2+a^2}(Mr^{-4}|am\omega|+a^2r^{-5}(m^2+1))\notag\\
\geq {}&\frac{r^3\Delta}{r^3(r^2+a^2)^4}
\left[2(r-3M)^2r(6-\varepsilon)-(4r^3-33Mr^2+78M^2r-54M^3)\right]
\notag\\
&-\frac{C\Delta}{r^2+a^2}r^{-3}\varepsilon^{-1}a^2\omega_l^2
-Ca^2r^{-5}.
\end{align}
Similarly, we only need to show the polynomial
$$
12r(r-3M)^2-(4r^3-33Mr^2+78M^2r-54M^3)=8r^3-39Mr^2+30M^2r+54M^3
$$
has no zeros in $[r_+, +\infty)$, and this has been shown in Section \ref{proof of integrated decay}.

In total, by making a choice of function $f$ as in \eqref{eq:boundfreq:fchoice}, it holds true that
\begin{align}
(Q^f)'\geq{}&{cM\Delta}{r^{-4}}|u'|^2+{c\Delta}{r^{-5}}|u|^2
+
\Re\left(2fu'\overline{H}+f'u\overline{H}\right).
\end{align}
By integrating over $[r_{-\infty}^*,r_{\infty}^*]$ and taking into account of the boundedness of $|\omega|$ and $\Lambda$, we arrive at
\begin{lemma}\label{lem:BoundFreqEstimate}
Let $r_{\infty}^*>R_0^*$ and $r_{-\infty}^*<r_c^*$ be arbitrary. Fix $\omega_l$ as in Section \ref{sect:trapfreq} and $\lambda_l$ as in Section \ref{sect:angulardominatedfreq}.  In $\mathcal{F}_B$ frequency regime, there exists a $\veps_0>0$ such that the following estimate holds true for all $|a|/M\leq \veps_0$
\begin{align}\label{eq:BoundFreqEstimate}
\hspace{4ex}&\hspace{-4ex}c\int_{r_c^*}^{R_0^*}\left(\Delta/r^4 \left|u'\right|^2+\Delta/ r^{5}\left(\omega^2+(\Lambda+4a^2\omega^2)
+1\right)|u|^2\right)\notag\\
\leq & \int_{r_{-\infty}^*}^{r_{\infty}^*}\left(2f\Re\left(u'\overline{H}\right)
+f'\Re\left(u\overline{H}\right)\right)
+Q^{f}\left(r_{\infty}^*\right)-Q^f\left(r_{-\infty}^*\right).
\end{align}
\end{lemma}

\subsection{Summing Up}\label{sect:summingandIEDcurrentesti}

Fix $\omega_l, \lambda_l,\lambda_s$ such that the estimates in Sections \ref{sect:timedominatedfreq}--\ref{sect:boundfreq} hold true and all the $r_{m\ell}^{(a\omega)}$ in $\mathcal{F}_{Tr}$ frequency regime satisfy  $r_{m\ell}^{(a\omega)}\in [r_{\text{trap}}^-,r_{\text{trap}}^+]$.
Summing over $m,\ell$ and integrating over $\omega$ in the estimates proved in the previous subsection, we obtain an identity of the form
\begin{align}
\label{eq:sumup:mainidentity}
\mathbf{M}=\mathbf{E} + \mathbf{F}_{r_{\infty}^*}
+\mathbf{F}_{r_{-\infty}^*} + \mathbf{R}.
\end{align}
Here, $\mathbf{M}$, $\mathbf{E}$, $\mathbf{F}_{r_{\infty}^*}$, $\mathbf{F}_{r_{-\infty}^*}$, and $\mathbf{R}$ are Morawetz integral terms, error terms from the source $H=\Delta (r^2 +a^2)^{-3/2}F_{\chi}$, flux term at $r_{\infty}^*$, flux term at $r_{-\infty}^*$, and a remainder term from the last integral in \eqref{eq:TimeDominatedFreqEstimate}, respectively. We shall estimate these terms one by one below.

\subsubsection{Morawetz Terms $\mathbf{M}$}
From the properties in Section \ref{sect:SeparateAngAndRadialEqs}, we have
\begin{align}
\mathbf{M}\geq{}c\int_{\mathcal{D}(0,\tau)\cap \{r_c\leq r\leq R_0\}} \left(\frac{|\partial_{r^*}\psi_{\chi}|^2+|\psi_{\chi}|^2}{r^3}
+\chi_{\text{trap}}(r)\frac{|\partial_{t} \psi_{\chi}|^2+|\nablaslash\psi_{\chi}|^2}{r^3}\right).
\end{align}
In view of the fact that $\chi=1$ in $[\varepsilon^{-1}, \tau-\varepsilon^{-1}]$, we use the estimate \eqref{eq:FiniteInTimeEnergyEstimateInhomoSWRWETo T+eN} and obtain
\begin{align}
\label{eq:sumup:Moraterm}
\hspace{4ex}&\hspace{-4ex}c\int_{\mathcal{D}(0,\tau)\cap \{r_c\leq r\leq R_0\}} \left(\frac{|\partial_{r^*}\psi|^2+|\psi|^2}{r^3}
+\chi_{\text{trap}}(r)\frac{|\partial_{t} \psi|^2+|\nablaslash\psi|^2}{r^3}\right)
\notag\\
\leq{}&\mathbf{M}
+C(r_c,R)\varepsilon^{-1}\bigg(
\int_{\Sigma_{0}}\left|e(\psi)\right|
+e_0E_{0}^{\text{total}}(s)+\mathcal{E}_{F,e_0,0,\tau}
\bigg)\notag\\
&
+C(r_c,R)\varepsilon^{-1}\bigg(
\int_{\Sigma_{\tau-\varepsilon^{-1}}}\left|e(\psi)\right|
+e_0E_{\tau-\varepsilon^{-1}}^{\text{total}}(s)
\bigg).
\end{align}

\subsubsection{Error Terms $\mathbf{E}$}\label{sect:EstiErrorTermCurrents}
The error terms $\mathbf{E}$ are of this form
\begin{align}
\int_{r_{-\infty}^*}^{r_{\infty}^*}\int_{-\infty}^{\infty}\sum_{m,\ell}
\Re\left(
\tfrac{\Delta }{\left(r^2+a^2\right)^{3/2}}
\left(F_{\chi}\right)_{m\ell}^{(a\omega)}\left(c(r)\bar{u}_{m\ell}^{(a\omega)}
+d(r)\partial_{r^*}\bar{u}_{m\ell}^{(a\omega)}\right)\right)
\di \omega \di r^*.
\end{align}
Let
\begin{subequations}
\begin{align}
F_{\chi}={}&F_{\chi,c}+F_{\chi,s},\\
F_{\chi,c}={}&\Sigma\left(2\nabla^{\mu}\chi\nabla_{\mu}\psi
+\left(\Box_g\chi\right)\psi\right)-2isa\cos\theta\partial_t\chi\psi,\\
F_{\chi,s}={}&\chi F.
\end{align}
\end{subequations}
We use the subscripts $c$ and $s$ to denote the error terms coming from the cutoff part $F_{\chi,c}$ and the source $F_{\chi,s}$ respectively  and decompose
\begin{align}
\mathbf{E}={}\mathbf{E}_{c,[r_{-\infty}^*, R_7^*]}
+\mathbf{E}_{c,[R_7^*,r_{\infty}^*]}
+\mathbf{E}_{s,[r_{-\infty}^*, R_7^*]}
+\mathbf{E}_{s,[R_7^*,r_{\infty}^*]},
\end{align}
where $R_7$ is a constant to be fixed but lies in $(\max\{R_4, R_5\}, \max\{R_4, R_5\}+M)$ such that $y=f=1$ and $h=0$ for $r\geq R_7$, the intervals $[r_{-\infty}^*, R_7^*]$ and $[R_7^*,r_{\infty}^*]$ are the $r^*$ region to integrate.

From the Cauchy--Schwarz inequality, we have for any $\varepsilon_3>0$ that
\begin{align}
\mathbf{E}_{c,[r_{-\infty}^*, R_7^*]}
\leq{}& \int_{\mathcal{D}(0,\tau)\cap[r_{-\infty},R_7]}
\frac{C\varepsilon_3 }{ r^{3}}\left(|\psi_{\chi}|^2+|\partial_{r^*} \psi_{\chi}|^2\right)
+\frac{C}{\varepsilon_3}|F_{\chi,c}|^2.
\end{align}
Note that we have incorporated powers of $R_7$ into the constant $C$ since $R_7$ has been fixed.
The first term is bounded by $ \int_{\mathcal{D}(0,\tau)\cap[r_{-\infty},R_7]}
\frac{C\varepsilon_3 }{r^{3}}\left(|\psi|^2+|\partial_{r^*} \psi|^2\right)$,
and, using the estimates \eqref{eq:PropertyofCutoffchi} and in view of the support of the derivatives of $\chi$, we bound the second term by
\begin{align}
\hspace{4ex}&\hspace{-4ex}
C\varepsilon_3^{-1}\varepsilon^2
\bigg(\int_{\mathcal{D}(0,\varepsilon^{-1})\cap[r_{-\infty},R_7]}
+\int_{\mathcal{D}(\tau-\varepsilon^{-1},\tau)\cap[r_{-\infty},R_7]}\bigg)
|\partial\psi|^2\notag\\
\leq{}&C\varepsilon_3^{-1}\varepsilon \left(E_{\tau-\varepsilon^{-1}}^{\text{total}}(s)
+E_{0}^{\text{total}}(s)
+\mathcal{E}_{F,1,0,\tau}\right),
\end{align}
where the estimate \eqref{eq:FiniteInTimeEnergyEstimateInhomoSWRWETo T+eN} with $\tilde{e}=1$ is used in this step. Therefore, we arrive at
\begin{align}
\label{eq:sumup:errorterm:cnear}
\mathbf{E}_{c,[r_{-\infty}^*, R_7^*]}
\lesssim {} &
\frac{\varepsilon}{\varepsilon_3} \left(E_{0}^{\text{total}}(s)
+\mathcal{E}_{F,1,0,\tau}\right)\notag\\
&
+\frac{\varepsilon}{\varepsilon_3} E_{\tau-\varepsilon^{-1}}^{\text{total}}(s)
+\varepsilon_3 \int_{\mathcal{D}(0,\tau)\cap [r_0,R_0]}
\frac{|\psi|^2+|\partial_{r^*} \psi|^2}{r^3}\notag\\
&+\varepsilon_3 \bigg(\int_{\mathcal{D}(0,\tau)\cap [r_+,r_0]}
+\int_{\mathcal{D}(0,\tau)\cap [R_0,\infty)}\bigg)
\frac{|\psi|^2+|\partial_{r^*} \psi|^2}{r^3} .
\end{align}
We have intentionally separated the integrals over different radius regions here, the reason of which will be clear in Section \ref{sect:summingandBEAM}.

The same argument applies to $\mathbf{E}_{s,[r_{-\infty}^*, R_7^*]}$, yielding
\begin{align}
\label{eq:sumup:errorterm:snear}
\mathbf{E}_{s,[r_{-\infty}^*, R_7^*]}
\lesssim {} &
\varepsilon_3 \int_{\mathcal{D}(0,\tau)\cap [r_0,R_0]}
\frac{|\psi|^2+|\partial_{r^*} \psi|^2}{r^3}
+\int_{\mathcal{D}(0,\tau)}{\varepsilon_3^{-1}}
\frac{|F|^2}{r^3}\notag\\
&+\varepsilon_3 \bigg(\int_{\mathcal{D}(0,\tau)\cap [r_+,r_0]}
+\int_{\mathcal{D}(0,\tau)\cap [R_0,\infty)}\bigg)
\frac{|\psi|^2+|\partial_{r^*} \psi|^2}{r^3}.
\end{align}

Consider the term $\mathbf{E}_{c,[R_7^*,r_{\infty}^*]}$. Since $f(r)=y(r)=1$, $h(r)=0$ and $\partial_{r^*}\chi=0$, and from the expression of $F_{\chi,s}$, it equals
\begin{align}\label{exp:sumup:errorterm:caway}
&\int_0^{\tau}\int_{\mathbb{S}^2}
\int^{r_{\infty}^*}_{R_7^*}\Re\left(\chi\partial_{r^*}
(\sqrt{r^2+a^2}\bar{\psi})
\frac{2\Sigma\Delta }{(r^2+a^2)^{3/2}}\nabla^{\mu}\chi\nabla_{\mu}\psi\right)
\di r^*\di \sigma_{\mathbb{S}^2}\di t^*\notag\\
&+\int_0^{\tau}\int_{\mathbb{S}^2}
\int^{r_{\infty}^*}_{R_7^*}\Re\left(\partial_{r^*}
(\sqrt{r^2+a^2}\chi\bar{\psi})
\frac{2\Sigma\Delta }{(r^2+a^2)^{3/2}}\left(\Box_g\chi\right)\psi\right)
\di r^*\di \sigma_{\mathbb{S}^2}\di t^*\notag\\
&-\int_0^{\tau}\int_{\mathbb{S}^2}
\int^{r_{\infty}^*}_{R_7^*}\Re\left(\partial_{r^*}
(\sqrt{r^2+a^2}\chi\bar{\psi})
\frac{2ias\cos\theta \Delta }{(r^2+a^2)^{3/2}}\partial_t\chi\psi\right)
\di r^*\di \sigma_{\mathbb{S}^2}\di t^*.
\end{align}
The derivatives of $\chi$ are supported in $[0,\varepsilon^{-1}]\cup[\tau-\varepsilon^{-1},\tau]$, hence, using the estimates \eqref{eq:PropertyofCutoffchi}, the first term of \eqref{exp:sumup:errorterm:caway} is bounded by
 \begin{align}
 \label{eq:sumup:errorterm:cawaybulk}
 C\varepsilon\bigg(\int_{\mathcal{D}(0,\varepsilon^{-1})
\cap[R_7,r_{\infty}]}
+\int_{\mathcal{D}(\tau-\varepsilon^{-1},\tau)
\cap[R_7,r_{\infty}]}\bigg)
|\partial\psi|^2.
\end{align}
The estimate \eqref{eq:FiniteInTimeEnergyEstimateInhomoSWRWETo T+eN} yields that this is further controlled by
\begin{align}
\label{eq:sumup:errorterm:caway}
&C\int_{\Sigma_{0}}
\left|e(\psi)\right|
+Ce_0{E}_{0}^{\text{total}}(s)
+C\int_{\Sigma_{\tau-\varepsilon^{-1}}}
\left|e(\psi)\right|
+Ce_0{E}_{\tau-\varepsilon^{-1}}^{\text{total}}(s)
\notag\\
&
+C\mathcal{E}_{F,e_0,0,\varepsilon^{-1}}
+C\mathcal{E}_{F,e_0,\tau-\varepsilon^{-1},\tau}.
\end{align}
The second term can be rewritten as
\begin{align}
\int_0^{\tau}\int_{\mathbb{S}^2}
\int^{r_{\infty}^*}_{R_7^*}\left(\partial_{r^*}
\left(\left|\sqrt{r^2+a^2}\chi\bar{\psi}\right|^2\right)
\frac{\Sigma\Delta }{(r^2+a^2)^{2}}\chi^{-1}\Box_g\chi\right)
\di r^*\di \sigma_{\mathbb{S}^2}\di t^*.
\end{align}
Applying an integration by parts in $r^*$ and noting the bounds for $r\geq R_7$
\begin{align}
|\partial_{r^*}\chi \Box_g\chi|
+|\chi \partial_{r^*}\Box_g\chi|+|\partial_r(\Sigma\Delta (r^2+a^2)^{-2})\chi\Box_g\chi|\lesssim \varepsilon r^{-2},
\end{align}
this is bounded by
\begin{align}
&C\bigg|\int_0^{\tau}\int_{\mathbb{S}^2}
\left(
\left|\sqrt{r^2+a^2}\chi\bar{\psi}\right|^2
\frac{\Sigma\Delta }{(r^2+a^2)^{2}}\chi^{-1}\Box_g\chi\right)_{r=R_7}
\di r^*\di \sigma_{\mathbb{S}^2}\di t^*\bigg|\notag\\
&
+C\varepsilon\bigg(\int_{\mathcal{D}(0,\varepsilon^{-1})
\cap[R_7,r_{\infty}]}
+\int_{\mathcal{D}(\tau-\varepsilon^{-1},\tau)
\cap[R_7,r_{\infty}]}\bigg)
|\partial\psi|^2.
\end{align}
Here the boundary term at $r=r_{\infty}$ vanishes for sufficiently large $r_{\infty}$ from the reduction in Section \ref{sect:LWPandGlobalExistenceLinearWaveSystem}. A simple application of the mean-value principle in $r$ allows us to fix $R_7$ by requiring that the boundary term at $r=R_7$ is bounded by $C$ times an integral over $[\max\{R_4, R_5\}, \max\{R_4, R_5\}+M]$, which is in turn bounded by
$$
C\varepsilon\bigg(\int_{\mathcal{D}(0,\varepsilon^{-1})
\cap[R_7-M,r_{\infty}]}
+\int_{\mathcal{D}(\tau-\varepsilon^{-1},\tau)
\cap[R_7-M,r_{\infty}]}\bigg)
|\partial\psi|^2.
$$
The same argument as in estimating \eqref{eq:sumup:errorterm:cawaybulk} above then applies, suggesting the second term of \eqref{exp:sumup:errorterm:caway} is bounded by \eqref{eq:sumup:errorterm:caway}.
The Cauchy--Schwarz inequality applied to the last term of \eqref{exp:sumup:errorterm:caway} gives an upper bound
\begin{align}
C|a|\varepsilon\bigg(\int_{\mathcal{D}(0,\varepsilon^{-1})
\cap[R_7,r_{\infty}]}
+\int_{\mathcal{D}(\tau-\varepsilon^{-1},\tau)
\cap[R_7,r_{\infty}]}\bigg)
|\partial\psi|^2,
\end{align}
which can again be bounded by \eqref{eq:sumup:errorterm:caway}. These together imply that $\mathbf{E}_{c,[R_7^*,r_{\infty}^*]}$ is bounded by  \eqref{eq:sumup:errorterm:caway}.

Consider in the end the term $\mathbf{E}_{s,[R_7^*,r_{\infty}^*]}$. This equals
\begin{align}
\int_0^{\tau}\int_{\mathbb{S}^2}
\int^{r_{\infty}^*}_{R_7^*}\Re\left(\chi\partial_{r^*}
(\sqrt{r^2+a^2}\bar{\psi})
\frac{\Delta }{(r^2+a^2)^{3/2}}\chi F\right)
\di r^*\di \sigma_{\mathbb{S}^2}\di t^*.
\end{align}
Since $R_7\geq R_0$, we utilize the Cauchy--Schwarz inequality and find
\begin{align}\label{eq:sumup:errorterm:saway}
\mathbf{E}_{s,[R_7^*,r_{\infty}^*]}
\lesssim{}
\varepsilon_5\int_{\mathcal{D}(0,\tau)
\cap[R_0,r_{\infty}]}\frac{|\partial\psi|^2}{r^{1+\delta}}
+\varepsilon_5^{-1}\int_{\mathcal{D}(0,\tau)
\cap[R_0,r_{\infty}]}\frac{|F|^2}{r^{3-\delta}}.
\end{align}

\subsubsection{Flux Terms}
First note that the flux term $\mathbf{F}_{r_{\infty}^*}$ vanishes by choosing $ r_{\infty}^*$ sufficiently large\footnote{$r_{\infty}^*$ is chosen based on $\tau$ from the property of finite speed of propagation.} in view of the reduction in Section \ref{sect:LWPandGlobalExistenceLinearWaveSystem}.
We have for the flux term $\mathbf{F}_{r_{-\infty}^*}$ that
\begin{align}
\mathbf{F}_{r_{-\infty}^*}
\lesssim{}&
\int_{-\infty}^{\infty}\int_{\mathbb{S}^2} \left(|\partial_t \psi_{\chi}|^2 +|\partial_{r^*}\psi_{\chi}|^2
+a^2|\partial_{\phi}\psi_{\chi}|^2 \right)_{r^*=r_{-\infty}^*}\di \sigma_{\mathbb{S}^2}\di t\notag\\
&+\int_{-\infty}^{\infty}\int_{\mathbb{S}^2} \left(\Delta\left(|\nablaslash \psi_{\chi}|^2 +|\psi_{\chi}|^2 \right)\right)_{r^*=r_{-\infty}^*}\di \sigma_{\mathbb{S}^2}\di t.
\end{align}
By taking the limit $r_{-\infty}^*\to -\infty$, the second line tends to $0$, and the term on the RHS of the first line has a limit, arriving at
\begin{align}
\label{eq:sumup:fluxterm}
\hspace{4ex}&\hspace{-4ex}\limsup_{r_{-\infty}^*\to -\infty}\mathbf{F}_{r_{-\infty}^*}\notag\\
\lesssim{}&\int_{-\infty}^{\infty}\int_{\mathbb{S}^2} \left(|\partial_t \psi_{\chi}|^2 +|\partial_{r^*}\psi_{\chi}|^2
+a^2|\partial_{\phi}\psi_{\chi}|^2 \right)_{r=r_+} \di \sigma_{\mathbb{S}^2}\di t\notag\\
\lesssim{}& \int_{\mathcal{H}^+(0,\tau)} \left(|\partial_t\psi+\partial_{r^*}\psi|^2
+a^2|\partial_{\phi}\psi|^2 \right)
+\varepsilon^2
\int_{\mathcal{H}^+(0,\tau)}
|\psi|^2\notag\\
\lesssim{}&E_{0}(\tilde{\psi})
+(\veps_0^2+\varepsilon^2){E}_{\mathcal{H}^+(0,\tau)}(\psi)
+e_0\int_{\mathcal{D}(0,\tau)\cap[r_0,r_1]}
\left|\partial\psi\right|^2+\mathcal{E}_{F,e_0,0,\tau}.
\end{align}
Here, we used the fact that $\partial_t=\partial_{r^*}-\frac{a}{2Mr_+}\partial_{\phi}$ on $\mathcal{H}^+$ in the second step  and the inequality \eqref{eq:SWRWEEnerEstiII} in the third step.

\subsubsection{Remainder Term $\mathbf{R}$}

Recall that this remainder term is from the last integral in \eqref{eq:TimeDominatedFreqEstimate}, hence from \eqref{eq:PropertyofCutoffchi},
\begin{align}
\label{eq:sumup:remainderR}
\mathbf{R}\lesssim{}&\int_{-\infty}^{\infty}\int_{R_5^*}^{\infty}\int_{\mathbb{S}^2}a^2r^{-3}|\partial_t (\sqrt{r^2+a^2}\psi_{\chi})|^2 \di \sigma_{\mathbb{S}^2}\di r^*\di t \notag\\
\lesssim{}&\int_{\mathcal{D}(0,\tau)\cap [R_5,\infty)}a^2r^{-3}(|\partial_t \psi|^2 +\veps^2 |\psi|^2).
\end{align}

\subsection{An Energy and Morawetz Estimate}\label{sect:summingandBEAM}
For any $\mu_i>0$, define
\begin{subequations}\label{def:errortermfromvectorfieldstotalKerrS2}
\begin{align}\label{def:errortermfromvectorfieldstotalKerrS2Posi}
\mathcal{E}_{\text{main},+2}^i(\mu_i)={}
&
\mu_i^{-1}\bigg|\int_{\mathcal{D}(0,\tau)}\Sigma^{-1}\Re\left(F_{+2}^i \partial_t\overline{\phi_{+2}^i}\right)\bigg|\notag\\
&
+(\veps_0+\mu_i)\sum_{j=0,1,2}\bigg|\int_{\mathcal{D}(0,\tau)}\Sigma^{-1}
\Re(F_{+2}^j\partial_t\overline{\phi^j_{+2}})\bigg|
\notag\\
&+\veps_0\int_{\mathcal{D}(0,\tau)}
\left(\widetilde{\mathbb{M}}(r^{4-\delta}\phi^0_{+2})+ \widetilde{\mathbb{M}}(r^{2-\delta}\phi^1_{+2})
+\mathbb{M}_{\text{deg}}(\phi^2_{+2})\right)\notag\\
&+\mu_i\sum_{j=0}^2\int_{\mathcal{D}(0,\tau)}
\mathbb{M}_{\text{deg}}(\phi^j_{+2}),
\\
\label{def:errortermfromvectorfieldstotalKerrS2Nega}
\mathcal{E}_{\text{main},-2}^i(\mu_i)={}
&
\mu_i^{-1}\bigg|\int_{\mathcal{D}(0,\tau)}\Sigma^{-1}\Re\left(F_{-2}^i \partial_t\overline{\phi_{-2}^i}\right)\bigg|\notag\\
&
+(\veps_0+\mu_i)\sum_{j=0,1,2}\bigg|\int_{\mathcal{D}(0,\tau)}\Sigma^{-1}
\Re(F_{-2}^j\partial_t\overline{\phi^j_{-2}})\bigg|
\notag\\
&+ \veps_0\int_{\mathcal{D}(0,\tau)}
\Big({\mathbb{M}}(\widetilde{\phi^0})
+{\mathbb{M}}(\widetilde{\phi^1})+ \mathbb{M}_{\text{deg}}(\phi^2_{-2})+|\nablaslash \widetilde{\phi^0}|^2+|\nablaslash \widetilde{\phi^1}|^2\Big)\notag\\
&+\mu_i\sum_{j=0}^2\int_{\mathcal{D}(0,\tau)}
\mathbb{M}_{\text{deg}}(\phi^j_{-2}),
\end{align}
\end{subequations}
and
\begin{subequations}\label{def:extraerrortermtotalKerrS2}
\begin{align}
\label{e+20}
\mathcal{E}_{\text{ex},+2}^0=&\eps_0\int_{\mathcal{D}(0,\tau)} \widetilde{\mathbb{M}}(r^{4-\delta}\phi^0_{+2})
+{\eps_0^{-1}}
   \int_{\mathcal{D}(0,\tau)}r^{-3}|\phi^1_{+2}|^2, \\
\mathcal{E}_{\text{ex},+2}^1={}&
\eps_1\int_{\mathcal{D}(0,\tau)} \widetilde{\mathbb{M}}(r^{2-\delta}\phi^1_{+2})
+
{\eps_1^{-1}}\int_{\mathcal{D}(0,\tau)}\big(
\widetilde{\mathbb{M}}_{\text{deg}}(r^{4-\delta}\phi^0_{+2})
+r^{-3}|\phi^2_{+2}|^2\big),\\
\mathcal{E}_{\text{ex},+2}^2={}&0,\\
\label{e-20}
\mathcal{E}_{\text{ex},-2}^0={}&\eps_0\int_{\mathcal{D}(0,\tau)} {\mathbb{M}}(\widetilde{\phi^0})
+\eps_0^{-1}\int_{\mathcal{D}(0,\tau)}
r^{-3}|\widetilde{\phi^1_{-2}}|^2, \\
\label{e-21}
\mathcal{E}_{\text{ex},-2}^1={}&\eps_1\int_{\mathcal{D}(0,\tau)} {\mathbb{M}}(\widetilde{\phi^1})
+\eps_1^{-1}\int_{\mathcal{D}(0,\tau)}\left(r^{-3}|\phi^2_{-2}|^2
+r^{-2}|\widetilde{\phi^0_{-2}}|^2\right),\\
\mathcal{E}_{\text{ex},-2}^2={}&0.
\end{align}
\end{subequations}
Let
\begin{equation}\label{def:ErrorTermFromSourceF}
\mathcal{E}_{s}^i(\mu_i)=\mathcal{E}_{\text{main},s}^i(\mu_i)
+\mathcal{E}_{\text{ex},s}^i.
\end{equation}

\begin{prop}\label{prop:MoraEstiAlmostScalarWave}
Under the assumptions in Theorem \ref{thm:EneAndMorEstiExtremeCompsNoLossDecayVersion2},
$\phi^i_s$ $(i=0,1,2, s=\pm 2)$ in \eqref{eq:DefOfphi012PosiSpinS2} and \eqref{eq:DefOfphi012NegaSpinS2}  satisfies the corresponding equation in the linear wave systems \eqref{eq:ReggeWheeler Phi^012KerrS2} and \eqref{eq:ReggeWheeler Phi^012KerrNegaS2} with the inhomogeneous term $F_s^i$.
Let $\varphi_s^i$
be any of the set
\begin{equation}
\left\{r^{4-\delta}\phi_{+2}^0,r^{2-\delta}\phi_{+2}^1,\phi_{+2}^2,
\widetilde{\phi^0_{-2}},\widetilde{\phi^1_{-2}},\phi^2_{-2}\right\},
\end{equation}
and has the same upper and lower indexes as $\phi_s^i$.
Then, for any $0<\delta<1/2$ and any $\mu_i>0$, there exist universal constants $\veps_0>0$,  $R=R(M)$ and $C=C(M,\delta,\Sigma_0)=C(M,\delta,\Sigma_{\tau})$ such that for all $|a|/M\leq \veps_0$ and any $\tau>0$,
the following estimates hold true:
\begin{subequations}\label{eq:MorawetEnergyEstimateforAlmostScalarWave}
\begin{itemize}
  \item For $(s,i)=(+2,0)$ or $(+2,1)$,
  \begin{align}\label{eq:MorawetEnergyEstimateforAlmostScalarWavea}
\hspace{4ex}&\hspace{-4ex}
{E}_{\tau}(\varphi_s^i)+
{E}_{\mathcal{H}^+(0,\tau)}(\varphi_s^i)
+\int_{\mathcal{D}(0,\tau)} \widetilde{\mathbb{M}}_{\text{deg}}(\varphi_s^i)\notag\\
\leq {}&C(\mu_i^{-1})E_{0}^{\text{total}}(s)+C\mathcal{E}_{s}^i(\mu_i);
\end{align}
  \item For other combinations of $(s,i)$,
  \begin{align}\label{eq:MorawetEnergyEstimateforAlmostScalarWaveb}
\hspace{4ex}&\hspace{-4ex}{E}_{\tau}(\varphi_s^i)
+{E}_{\mathcal{H}^+(0,\tau)}(\varphi_s^i)
+\int_{\mathcal{D}(0,\tau)} \mathbb{M}_{\text{deg}}(\varphi_s^i)\notag\\
\leq {}&C(\mu_i^{-1})E_{0}^{\text{total}}(s)+C\mathcal{E}_{s}^i(\mu_i)
.
\end{align}
\end{itemize}
\end{subequations}
Here,  $E_{0}^{\text{total}}(s)$ is defined as in \eqref{def:EnergynormtotalS2Kerr}.
\end{prop}

\begin{proof}
The separation in variables in Section \ref{sect:SeparateAngAndRadialEqs} requires $\tau> 2\veps^{-1}$, hence
we separate the cases $\tau> 2\veps^{-1}$ and $0\leq \tau\leq 2\veps^{-1}$. For each $\psi=\phi_s^i$, the choice of $\veps$ is made to be different. The rest of the proof will be for each single $\psi=\phi_s^i$.

In the first case where $\tau> 2\veps^{-1}$, we have otained an identity \eqref{eq:sumup:mainidentity}. Add to the equation \eqref{eq:sumup:mainidentity} $c_2$ times the red-shift estimate near horizon in Proposition \ref{prop:RedShiftEstiInhomoSWRWE} (if $\psi\in\{\phi^i_{+2}\}_{i=0,1,2}\cap \{\phi_{-2}^2\}$) and Proposition \ref{prop:RedShiftEstiInhomoSWRWEtildephi01} (if $\psi\in\{\phi^0_{-2},\phi^1_{-2}\}$)
and $c_2$ times the
Morawetz estimate in large $r$ region in Proposition \ref{prop:ImprovedMoraEstiLargerSWRWE} (if $\psi\in \{\phi^0_{+2},\phi^1_{+2}\}$) or Proposition \ref{prop:ImprovedMoraEstiLargerSWRWE1} (if $\psi\in \{\phi^0_{-2}, \phi^1_{-2}, \phi^2_{-2}, \phi^2_{+2}\}$) applied to the spacetime region $\mathcal{D}(0,\tau)$. Then for sufficiently small $\veps_0$ and $\varepsilon$, choosing $\varepsilon_3$, $\varepsilon_5$ and $c_2$ to satisfy $a_0^2/M^2+\varepsilon^2+e_0 + \varepsilon_3+ \varepsilon_5\ll c_2\ll 1$ allows us to absorb the bulk integrals of $\psi$ over $[r_0,r_1]$ and $[R_0-M,R_0]$ on the RHS of these two estimates, the bulk integrals over $[r_0,R_0]$ in both \eqref{eq:sumup:errorterm:cnear} and \eqref{eq:sumup:errorterm:snear}, and the bulk integral over $[r_0,r_1]$ in \eqref{eq:sumup:fluxterm} by the LHS of \eqref{eq:sumup:Moraterm}, and absorb the last line of \eqref{eq:sumup:errorterm:cnear}, the last line of \eqref{eq:sumup:errorterm:snear}, the bulk integral of $\psi$ in \eqref{eq:sumup:errorterm:saway}, the horizon flux term in \eqref{eq:sumup:fluxterm} and the RHS of \eqref{eq:sumup:remainderR} by the LHS of the added red-shift estimate near horizon and Morawetz estimate in large $r$ region. The parameters $\varepsilon_3$, $\varepsilon_5$ and $c_2$ are now fixed, and the LHS of the obtained estimate thus dominates over the LHS of \eqref{eq:MorawetEnergyEstimateforAlmostScalarWave}. The remaining bulk integrals on the RHS of the added red-shift estimate near horizon and Morawetz estimate in large radius region are bounded by $\mathcal{E}_s^i$ using the Cauchy--Schwarz inequality. Note here that we actually used the following estimate from the Cauchy--Schwarz inequality for the last term in \eqref{eq:ImprovedMoraEstiLargerSWRWE}:
\begin{align}
\label{eq:sumup:errorMoraallphi}
\hspace{4ex}&\hspace{-4ex} \bigg|\int_{\mathcal{D}(0,\tau)\cap\{r\geq R_0-M\}}\Re\left(F_{s}^iX_w\overline{\phi_{s}^i}\right)\bigg|\notag\\
\lesssim{}&\veps_6 \int_{\mathcal{D}(0,\tau)} \mathbb{M}_{\text{deg}}(\varphi_s^i)
+\veps_6^{-1}\int_{\mathcal{D}(0,\tau)}r^{-3+\delta}|F_{s}^i|^2
\notag\\
\lesssim{}&\veps_6 \int_{\mathcal{D}(0,\tau)} \mathbb{M}_{\text{deg}}(\varphi_s^i)
+\veps_6^{-1}\mathcal{E}_{s}^i(\mu_i),
\end{align}
and chose $\veps_6$ small enough such that the term $\veps_6 \int_{\mathcal{D}(0,\tau)} \mathbb{M}_{\text{deg}}(\varphi_s^i)$ is absorbed by the LHS of \eqref{eq:MorawetEnergyEstimateforAlmostScalarWave}. In proving the estimate \eqref{eq:sumup:errorMoraallphi}, we also showed in the second step that both the second term on the RHS of \eqref{eq:sumup:errorterm:snear} and the last term in \eqref{eq:sumup:errorterm:saway} are  bounded by $C\mathcal{E}_{s}^i(\mu_i)$.

We now obtain an estimate in which the LHS dominates over the LHS of \eqref{eq:MorawetEnergyEstimateforAlmostScalarWave} and all the terms on the RHS are bounded by $C(\veps^{-1})E_{0}^{\text{total}}(s)+C\mathcal{E}_{s}^i$ except for terms in the following two categories:
\begin{enumerate}
\item\label{pt:1}
$\varepsilon^{-1}\big(
\int_{\Sigma_{\tau-\varepsilon^{-1}}}\left|e(\psi)\right|
+e_0E_{\tau-\varepsilon^{-1}}^{\text{total}}(s)
\big)$ in \eqref{eq:sumup:Moraterm}, $\int_{\Sigma_{\tau-\varepsilon^{-1}}}
\left|e(\psi)\right|
+e_0{E}_{\tau-\varepsilon^{-1}}^{\text{total}}(s)$ in \eqref{eq:sumup:errorterm:caway}, and $\frac{\varepsilon}{\varepsilon_3} E_{\tau-\varepsilon^{-1}}^{\text{total}}(s)$ in \eqref{eq:sumup:errorterm:cnear};
\item\label{pt:2}
$\veps^{-1}\mathcal{E}_{F,e_0,0,\tau}$ in \eqref{eq:sumup:Moraterm},
$\mathcal{E}_{F,e_0,0,\varepsilon^{-1}}$ in \eqref{eq:sumup:errorterm:caway},
$\mathcal{E}_{F,e_0,\tau-\varepsilon^{-1},\tau}$ in \eqref{eq:sumup:errorterm:caway}, $\mathcal{E}_{F,e_0,0,\tau}$ in \eqref{eq:sumup:fluxterm},
and ${\varepsilon_3}^{-1}{\varepsilon}\mathcal{E}_{F,1,0,\tau}$ in \eqref{eq:sumup:errorterm:cnear}.
\end{enumerate}
Consider first the three terms in Category \ref{pt:1}.
One adds \eqref{eq:SWRWEEnerEstiII} with $\tilde{e}=1$, $\tau_1=0$ and $\tau_2=\tau-{\varepsilon}^{-1}$ to its corresponding Morawetz estimate in large radius region (i.e. Propositions \ref{prop:ImprovedMoraEstiLargerSWRWE1} and \ref{prop:ImprovedMoraEstiLargerSWRWE}), then summing over $i\in \{0,1,2\}$ gives
\begin{align} E_{\tau-\varepsilon^{-1}}^{\text{total}}(s)
\lesssim{}& E_{0}^{\text{total}}(s)
+\sum_{i=0,1,2}\int_{\mathcal{D}(0,\tau)}
\Sigma^{-1}\Re({F_s^i}T\overline{\phi_s^i})
\notag\\
&+\sum_{i=0,1,2}\int_{\mathcal{D}(0,\tau)\cap \{[r_+,r_1]\cap [R_0-M,R_0]\}}(|\partial\varphi_s^i|^2+|F_s^i|^2).
\end{align}
For the term $\int_{\Sigma_{\tau-\varepsilon^{-1}}}\left|e(\psi)\right|$, one obtains from the estimate \eqref{eq:SWRWEEnerEstiII} with $\tilde{e}=e_0\sim \veps_0^2$, $\tau_2=\tau-\varepsilon^{-1}$ and $\tau_1=0$ that
\begin{align}
\int_{\Sigma_{\tau-\varepsilon^{-1}}}\left|e(\psi)\right|
\lesssim {}&E_0(\varphi_s^i)
+\int_{\mathcal{D}(0,\tau)}
{\Sigma^{-1}}\Re({F_s^i}T\overline{\phi_s^i})\notag\\
&+\veps_0^2\int_{\mathcal{D}(0,\tau)\cap [r_+,r_1]}(|F_s^i|^2+|\varphi_s^i|^2).
\end{align}
From the expressions of $F_s^i$,
\begin{align}
\sum_{i=0}^{2}\int_{\mathcal{D}(0,\tau)\cap \{[r_+,r_1]\cap [R_0-M,R_0]\}}(|\partial\varphi_s^i|^2+|F_s^i|^2)\lesssim
\sum_{i=0}^{2}\int_{\mathcal{D}(0,\tau)}\mathbb{M}_{\text{deg}}(\varphi_s^i).
\end{align}
We combine these three estimates and choose $\veps\sim\mu_i$, then for sufficiently small $\veps_0$, these three terms are bounded by the RHS of  \eqref{eq:MorawetEnergyEstimateforAlmostScalarWave}.

Turn to the terms in Category \ref{pt:2}. Note that for each $\psi=\phi_s^i$, $\veps$ is chosen to satisfy $\veps\sim\mu_i$. Consider one fixed $\psi=\phi_s^i$. From the expression \eqref{def:bulktermtotalS2Kerr}, the first four terms therein are bounded by
\begin{align}
\mu_i^{-1}\bigg(e_0
\int_{\mathcal{D}(0,\tau)\cap[r_+,r_1]}
\mathbb{B}(\tilde{\psi},F)+\bigg|\int_{\mathcal{D}(0,
\tau)}\Re\left(\Sigma^{-1} F T\bar{\psi}\right)\bigg|\bigg).
\end{align}
Since $e_0\sim \veps_0^2$, by choosing $\veps_0\ll \mu_i$,  this is further bounded by $C\mathcal{E}_{s}^i(\mu_i)$. Since $\veps_3$ has been fixed, $\veps \sim \mu_i$, and from the expression \eqref{def:bulktermtotalS2Kerr}, one obtains the estimate for the last term ${\varepsilon_3}^{-1}{\varepsilon}\mathcal{E}_{F,1,0,\tau}\leq C\mathcal{E}_{s}^i(\mu_i)$.

For each $\psi=\phi_s^i$, the other case $\tau\leq 2\veps^{-1}\lesssim \mu_i^{-1}$ follows from a standard well-posedness argument of a general linear wave system.
\qed
\end{proof}

\subsection{Complete the Proof of \eqref{eq:MoraEstiFinal(2)KerrRegularpsiBothSpinComp} on Slowly Rotating Kerr}\label{sect:finishpfS2Kerr}

The estimates \eqref{eq:estiphi02hatphi1kerrS2} for spin $+2$ component and \eqref{eq:estiphi02hatphi1kerrS2Nega} for spin $-2$ component are proved on slowly rotating Kerr backgrounds in this subsection.

\subsubsection{Spin $+2$ Component}\label{sec:positivespin1}
Let us treat the error terms $\mathcal{E}_{+2}^i(\mu_i)$ in the energy and Morawetz estimate \eqref{eq:MorawetEnergyEstimateforAlmostScalarWave}. An application of the Cauchy--Schwarz inequality gives
\begin{align}
\label{eq:Errorpartialtphi0Posti}
\hspace{6ex}&\hspace{-6ex}
\bigg|\int_{\mathcal{D}(0,\tau)}{\Sigma^{-1}}
\Re\left(F_{+2}^0 \partial_t\overline{\phi^0}\right)\bigg|\notag\\
\lesssim_{\veps_0} {}&\eps_0\int_{\mathcal{D}(0,\tau)} \widetilde{\mathbb{M}}(r^{4-\delta}\phi^0)
+{\eps_0^{-1}}
\int_{\mathcal{D}(0,\tau)}r^{-3}\left|\phi^1_{+2}\right|^2\notag\\
\lesssim_{\veps_0}{}&
\int_{\mathcal{D}(0,\tau)} \bigg(\eps_0 \widetilde{\mathbb{M}}(r^{4-\delta}\phi^0)
+
\frac{\hat{\eps}_1}{\eps_0} \widetilde{\mathbb{M}}(r\phi^1)+\frac{1}{\eps_0\hat{\eps}_1} \mathbb{M}_{\text{deg}}(\phi^2)\bigg),\\
\label{eq:Errorpartialtphi1Posti}
\hspace{6ex}&\hspace{-6ex}
\bigg|\int_{\mathcal{D}(0,\tau)}{\Sigma^{-1}}\Re\left(F_{+2}^1 \partial_t\overline{\phi^1}\right)\bigg|\notag\\
\lesssim_{\veps_0} {}&\int_{\mathcal{D}(0,\tau)} \left( \eps_1\widetilde{\mathbb{M}}(r^{2-\delta}\phi^1)
+\eps_1^{-1}\big(
r^{-2}|\phi^0_{+2}|^2
+r^{-3}|\phi^2_{+2}|^2\big)\right).
\end{align}
Here, we used the estimate \eqref{eq:estiphi1byphi02final1} to bound the integral ${\eps_0^{-1}}
\int_{\mathcal{D}(0,\tau)}r^{-3}\left|\phi^1_{+2}\right|^2$.
For the term $\left\vert\int_{\mathcal{D}(0,\tau)}{\Sigma^{-1}}\Re(F_{+2}^2 \partial_t\overline{\phi^2})\right\vert$, we have
\begin{align}
&\bigg|\int_{\mathcal{D}(0,\tau)}
\Sigma^{-1}\Re\left(F_{+2}^2 \partial_t\overline{\phi^2}\right)\bigg|
\lesssim{}
\bigg|\int_{\mathcal{D}(0,\tau)}a^2\Sigma^{-1}
\Re\left(\partial_t\phi^1\partial_{t} \overline{\phi^2}\right)\bigg|\notag\\
&\quad \quad \quad +\bigg|\int_{\mathcal{D}(0,\tau)}a^2\Sigma^{-1}
\Re\left(\phi^0\partial_{t}
\overline{\phi^2}\right)\bigg|
+\bigg|\int_{\mathcal{D}(0,\tau)}a\Sigma^{-1}
\Re\left(\partial_{\phi}\phi^1\partial_{t}
\overline{\phi^2}\right)\bigg|.
\end{align}
These terms, a priori, cannot be estimated in the trapped region due to the trapping degeneracy.
The sum of the first and second integrals on the RHS is
\begin{align}\label{eq:EstiIKerr1 mostannoyingterm1Neun2o}
\hspace{6ex}&\hspace{-6ex}
\bigg|\int_{\mathcal{D}(0,\tau)}\frac{a^2}{2\Sigma}
Y\left(r^2|\partial_t\phi^1|^2\right)\bigg|
+\bigg|\int_{\mathcal{D}(0,\tau)}\frac{a^2}{\Sigma}
\left(\partial_{t^*}\left(\Re(\phi^0 \overline{\phi^2})\right)
-\Re(\partial_{t}\phi^0
\overline{\phi^2})\right)\bigg|
\notag\\
 \lesssim_{\veps_0} {}&0,
\end{align}
where an integration by part is used for the first term on the LHS of \eqref{eq:EstiIKerr1 mostannoyingterm1Neun2o}.
As to the third integral, we take $\check{r}_1 \in (r_0,r_{\text{trap}}^-)$ and $\check{R}_1 > r_{\text{trap}}^+$ to be fixed and split this integral into three sub-integrals over $[r_+, \check{r}_1]$, $[\check{r}_1, \check{R}_1]$, and $[\check{R}_1,\infty)$, respectively.
The sum of the sub-integrals over $[r_+, \check{r}_1]$ and $[\check{R}_1,\infty)$ is bounded by $C\Xi_{+2}(0,\tau)$ since their integral regions are away from the trapped region.
For the remaining sub-integral over $[\check{r}_1, \check{R}_1]$, we utilize
\begin{equation}\label{eq:partialtexpressedbyYpartialphipartialr}
\partial_t \phi^2=
(\R)^{-1}\left(\Delta Y\phi^2-a\partial_{\phi}\phi^2+\Delta\partial_{r}\phi^2\right),
\end{equation}
and find this sub-integral is bounded by
\begin{align}\label{eq:ControlInTrappingRegionPosiSpin1}
&\bigg|\int_{\mathcal{D}(0,\tau)\cap[\check{r}_2,R_1]}
\left(\tfrac{2a\Delta}{r\Sigma(\R)}\Re\left(\partial_{\phi}
(\overline{r\phi^1})Y\phi^2\right)
-\tfrac{a^2}{\Sigma(\R)}Y\left(\left|
\partial_{\phi}(r\phi^1)\right|^2\right)\right)\bigg|\notag\\
& +\bigg|\int_{\mathcal{D}(0,\tau)\cap[\check{r}_2,R_1]}\tfrac{2a\Delta}{\Sigma(\R)}
\Re\left(\partial_{\phi}
(\overline{\phi^1})\partial_r\phi^2\right)\bigg|
\lesssim_{\veps_0} {}
0.
\end{align}
In the last step, integrations by parts are applied to the first line and two radius parameters $\check{r}_1$ and $\check{R}_1$ are appropriately chosen such that the boundary terms at $\check{r}_1$ and $\check{R}_1$ are bounded via an average of integration by $\tfrac{C|a|}{M} \int_{\mathcal{D}(0,\tau)}\widetilde{\mathbb{M}}(r^{4-\delta}\phi^0)
+\widetilde{\mathbb{M}}(r^{2-\delta}\phi^1)$.
Therefore, it holds true that
\begin{align}
\label{eq:EstiIKerr1 mostannoyingterm1Neun2I}
\bigg|\int_{\mathcal{D}(0,\tau)\cap [\check{r}_1,\check{R}_1]}a\Sigma^{-1}
\Re\left(\partial_{\phi}\phi^1\partial_{t}
\overline{\phi^2}\right)\bigg|
\lesssim_{\veps_0} {}&0,
\end{align}
which further implies together with the above discussions that
\begin{align}\label{eq:estiForF2partialtphi2}
\bigg|\int_{\mathcal{D}(0,\tau)}\Sigma^{-1}\Re\left(F_{+2}^2 \partial_t\overline{\phi^2}\right)\bigg|
\lesssim_{\veps_0} 0.
\end{align}
Note that the last term in \eqref{e+20} is bounded already in \eqref{eq:Errorpartialtphi0Posti} using \eqref{eq:estiphi1byphi02final1}.
The estimates \eqref{eq:estiphi02hatphi1kerrS2} are then valid from Theorem \ref{prop:MoraEstiAlmostScalarWave} and the estimates \eqref{eq:Errorpartialtphi0Posti}, \eqref{eq:Errorpartialtphi1Posti} and \eqref{eq:estiForF2partialtphi2}.

\subsubsection{Spin $-2$ Component}\label{sec:NegativeSpin1}
We now estimate the error terms $\mathcal{E}_{-2}^i(\mu_i)$ in the energy and Morawetz estimate \eqref{def:ErrorTermFromSourceF}. Note from the Cauchy--Schwarz inequality, \eqref{eq:estiphi1byphi02negafinal1}, \eqref{eq:estiphi0byphi02negafinal1} and \eqref{eq:estiphi0byphi1angunegafinal1} that
\begin{align}
\label{eq:Errorpartialtphi0Nega}
\hspace{6ex}&\hspace{-6ex}
\bigg|\int_{\mathcal{D}(0,\tau)}\Sigma^{-1}\Re\left(F_{-2}^0 \partial_t\overline{\phi^0}\right)\bigg|\notag\\
\lesssim_{\veps_0} {}&\eps_0
\int_{\mathcal{D}(0,\tau)} {\mathbb{M}}(\phi^0)
+\eps_0^{-1}
\int_{\mathcal{D}(0,\tau)}r^{-3}|\widetilde{\phi^1}|^2\notag\\
\lesssim_{\veps_0}{}&
\int_{\mathcal{D}(0,\tau)} \left(\eps_0{\mathbb{M}}(\widetilde{\phi^0})
+\eps_0^{-1} \mathbb{M}_{\text{deg}}(\phi^2)\right),\\
\label{eq:Errorpartialtphi1Nega}
\hspace{6ex}&\hspace{-6ex}
\bigg|\int_{\mathcal{D}(0,\tau)}\Sigma^{-1}\Re\left(F_{-2}^1 \partial_t\overline{\phi^1}\right)\bigg|\notag\\
\lesssim_{\veps_0} {}&\int_{\mathcal{D}(0,\tau)} \Big(\eps_1
{\mathbb{M}}(\phi^1)+
{\eps_1^{-1}}\Big(
r^{-2}{|\widetilde{\phi^0}|^2}
+r^{-3}{|\phi^2|^2}+\tfrac{|a|}{M}|\nablaslash \widetilde{\phi^0}|^2\Big)
\Big)\notag\\
\lesssim_{\veps_0} {} &
\int_{\mathcal{D}(0,\tau)} \Big(\eps_1+\eps_1^{-1}\tfrac{|a|}{ M} \Big) {\mathbb{M}}(\widetilde{\phi^1})+
\eps_1^{-1}\left(
\mathbb{M}_{\text{deg}}(\phi^2)
+\tfrac{|a|}{M}\mathbb{M}(\phi^0)
\right).
\end{align}

For the term $\left|\int_{\mathcal{D}(0,\tau)}\Sigma^{-1}\Re\left(F_{-2}^2 \partial_t\overline{\phi^2}\right)\right|$, we have
\begin{align}\label{eq:MainAnnoyingErrorTermNega}
&\bigg|\int_{\mathcal{D}(0,\tau)}\Sigma^{-1}\Re\left(F_{-2}^2 \partial_t\overline{\phi^2}\right)\bigg|
\lesssim
\bigg|\int_{\mathcal{D}(0,\tau)}\frac{a^2}{\Sigma}
\Re\left(\partial_t\phi^1\partial_{t} \overline{\phi^2}\right)\bigg|\notag\\
&\quad \quad \quad \quad \quad
+\bigg|\int_{\mathcal{D}(0,\tau)}\frac{a^2}{\Sigma}
\Re\left(\phi^0\partial_{t}
\overline{\phi^2}\right)\bigg|
+\bigg|\int_{\mathcal{D}(0,\tau)}\frac{a}{\Sigma}
\Re\left(\partial_{\phi}\phi^1\partial_{t}
\overline{\phi^2}\right)\bigg|
.
\end{align}
We split the first integral on the RHS into two sub-integrals over $[r_+,\check{r}_2]$ and $[\check{r}_2,\infty)$, with $\check{r}_2\in (r_1,r_{\text{trap}}^-)$ to be fixed, and obtain
\begin{align}\label{eq:EstiIKerr1NegaSpin}
\hspace{6ex}&\hspace{-6ex}
\bigg|\int_{\mathcal{D}(0,\tau)}\frac{a^2}{\Sigma}
\Re\left(\partial_t\phi^1\partial_{t} \overline{\phi^2}\right)\bigg|\notag\\
\leq {}&
\bigg|\int_{\mathcal{D}(0,\tau)\cap[\check{r}_2,\infty)}
\frac{a^2}{2\Sigma}
V\left(r^2|\partial_t\phi^1|^2\right)\bigg|
+\bigg|\int_{\mathcal{D}(0,\tau)\cap[r_+, \check{r}_2]}\frac{a^2}{\Sigma}
\Re\left(\partial_t\phi^1\partial_{t} \overline{\phi^2}\right)\bigg|\notag\\
\lesssim_{\veps_0} {}&\frac{|a|}{M}\int_{\mathcal{D}(0,\tau)\cap \{r=\check{r}_2\}}|\partial \phi^1|^2,
\end{align}
where in the last step an integration by parts is performed for the first term in the second line and the second term in the second line can be directed estimated vis Cauchy--Schwarz inequality.
The radius parameter $\check{r}_2$ is chosen such that the last term in \eqref{eq:EstiIKerr1NegaSpin} can be bounded, via an average of integration, by
\begin{align}\label{eq:ControlOfEnergyOnConstantrc}
\frac{|a|}{M}\int_{\mathcal{D}(0,\tau)\cap \{r=\check{r}_2\}}\left|\partial \phi^1\right|^2\lesssim \frac{|a|}{M}\int_{\mathcal{D}(0,\tau)}{\mathbb{M}}(\phi^1)
\lesssim \frac{|a|}{M}\int_{\mathcal{D}(0,\tau)}{\mathbb{M}}(\widetilde{\phi^1}).
\end{align}
We split the integral region of the terms in the second line of \eqref{eq:MainAnnoyingErrorTermNega} into two subregions $[r_+, \check{r}_3]$ and $(\check{r}_3,\infty)$ with $\check{r}_3 \in (r_0,r_{\text{trap}}^-)$ to be fixed.
 The terms integrated over $[r_+, \check{r}_3]$ are bounded by $C\Xi_{-2}(0,\tau)$ because of no trapping degeneracy in this subregion. While, for the integrals over $(\check{r}_3,\infty)$,
we rewrite $\partial_t\phi^2$ as
\begin{equation}
\partial_t \phi^2=
(\R)^{-1}\left(\Delta V\phi^2-a\partial_{\phi}\phi^2-\Delta\partial_{r}\phi^2\right),
\end{equation}
and find these integrals are dominated by
\begin{align}\label{eq:errortermdtofphi2negaS2KerrIStep8}
\hspace{6ex}&\hspace{-6ex}\left|\int_{\mathcal{D}(0,\tau)\cap[\check{r}_3,\infty)}
\left(\frac{a\Delta}{r\Sigma(\R)}\Re\left(\partial_{\phi}
(\overline{r\phi^1})V\phi^2\right)
-\frac{a^2}{2\Sigma(\R)}V\left(\left|
\partial_{\phi}(r\phi^1)\right|^2\right)\right)\right|\notag\\
\hspace{6ex}&\hspace{-6ex}
+\left|\int_{\mathcal{D}(0,\tau)\cap[\check{r}_3,\infty)}
\frac{a^2\Delta}{r\Sigma(\R)}\Re\left(
(\overline{r\phi^0})V\phi^2\right)
\right|\notag\\
\hspace{6ex}&\hspace{-6ex}
+\left|\int_{\mathcal{D}(0,\tau)\cap[\check{r}_3,\infty)}
\frac{\Delta}{\Sigma(\R)}\left(\Re\left(a\partial_{\phi}
\overline{\phi^1}\partial_r\phi^2\right)
+\Re\left(a^2
\overline{\phi^0}\partial_r\phi^2\right)\right)\right|\notag\\
\hspace{6ex}&\hspace{-6ex}
+\left|\int_{\mathcal{D}(0,\tau)\cap[\check{r}_3,\infty)}
\frac{a^2}{\Sigma(\R)}\left(\partial_{\phi}
\left(\Re\left(\phi^0\overline{\phi^2}\right)\right)
-\Re\left(\partial_{\phi}\phi^0\overline{\phi^2}\right)\right)\right|\notag\\
\lesssim_{\veps_0}  {}&\int_{\mathcal{D}(0,\tau)}\frac{|a|}{M}\left(
|\nablaslash\widetilde{\phi^1}|^2
+r^{-2}|\widetilde{\phi^0}|^2\right).
\end{align}
Here, we applied integration by parts to the first two lines and utilized the definition \eqref{eq:DefOfphi012NegaSpinS2} and similar estimates as \eqref{eq:ControlOfEnergyOnConstantrc} to control the boundary terms at $\check{r}_3$ by appropriately choosing this radius parameter.
In summary, we have
\begin{align}\label{eq:estiForF1partialtphi1Nega}
\left|\int_{\mathcal{D}(0,\tau)}\Sigma^{-1}\Re\left(F_{-2}^2 \partial_t\overline{\phi^2}\right)\right|
\lesssim_{\veps_0} {}&\int_{\mathcal{D}(0,\tau)}\frac{|a|}{M}\left(
|\nablaslash\widetilde{\phi^1}|^2
+r^{-2}{|\widetilde{\phi^0}|^2}\right)\lesssim_{\veps_0} 0,
\end{align}
where the second step comes from the estimates \eqref{eq:estiphi0byphi02negafinal1} and \eqref{eq:estiphi1byphi2angunegafinal1}. In the meanwhile, the last term in \eqref{e-20} and the last integral in \eqref{e-21} have been estimated in \eqref{eq:Errorpartialtphi0Nega} and
\eqref{eq:Errorpartialtphi1Nega}.
The estimates \eqref{eq:estiphi02hatphi1kerrS2Nega} follow from the conclusion in Proposition \ref{prop:MoraEstiAlmostScalarWave}, the estimates \eqref{eq:Errorpartialtphi0Nega},
\eqref{eq:Errorpartialtphi1Nega} and \eqref{eq:estiForF1partialtphi1Nega}, and the fact that we can choose $\veps_0$ sufficiently small such that all the terms with coefficients proportional to $|a|/M$ can be put into $\lesssim_{\veps_0}$.

\subsection{An Energy Bound}\label{sect:energyboundinitial}
The following bound on the term
$\int_{\Sigma_0}r(|\nablaslash \widetilde{\phi^0}|^2+|\nablaslash \widetilde{\phi^1}|^2)$
is used in \eqref{eq:estiMoraphi012Kerrnegacomp}.
\begin{prop}\label{prop:InitialEnergyControlsMore}
For the spin $-2$ component, we have for any real value $\tau$ that
\begin{align}\label{eq:ControlInitrPhi01}
\int_{\Sigma_{\tau}}r(|\nablaslash \widetilde{\phi^0}|^2+|\nablaslash \widetilde{\phi^1}|^2)\lesssim E_{\tau}(\widetilde{\phi^0}) +E_{\tau}(\widetilde{\phi^1})
+E_{\tau}(\phi^2_{-2}).
\end{align}
\end{prop}
\begin{proof}
Rewrite the equations \eqref{eq:ReggeWheeler Phi^0KerrNegaS2} and \eqref{eq:ReggeWheeler Phi^1KerrNegaS2} as
\begin{subequations}
\begin{align}
\label{def:Qphi0phi1S2}
0
={}&
\Big(\tfrac{1}{\sin{\theta}} \partial_{\theta}(\sin \theta \partial_{\theta})
+\tfrac{1}{\sin^2\theta}\partial_{\phi\phi}^2
-4i\tfrac{\cos\theta}{\sin^2\theta}\partial_{\phi}
-\tfrac{4}{\sin^2\theta}
+\tfrac{2r^2 -6Mr + 6a^2}{r^2}\Big)\phi^0
\notag\\
&
+\tfrac{\Delta}{r^2}Y\phi^1
+4ia\cos\theta \partial_t\phi^0
+a^2\cos^2\theta\partial_{tt}^2\phi^0
+2a\partial_{t\phi}^2\phi^0\notag\\
&
+\tfrac{2(a^2\partial_t+a\partial_{\phi})}{r}\phi^0
+\tfrac{2ar}{\R}\partial_{\phi}\phi^0
-\tfrac{3\Delta+a^2}{r^3}\phi^1
,\\
\label{def:Qphi1phi2S2}
0
={}&
\Big(\tfrac{1}{\sin{\theta}} \partial_{\theta}(\sin \theta \partial_{\theta})
+\tfrac{1}{\sin^2\theta}\partial_{\phi\phi}^2
-4i\tfrac{\cos\theta}{\sin^2\theta}\partial_{\phi}
-\tfrac{4}{\sin^2\theta}
+\tfrac{6Mr-6a^2}{r^2}\Big)\phi^1
\notag\\
&
+\tfrac{\Delta}{r^2}Y\phi^2
+4ia\cos\theta \partial_t\phi^1
+a^2\cos^2\theta\partial_{tt}^2\phi^1
+2a\partial_{t\phi}^2\phi^1
+\tfrac{6(a^2\partial_t+a\partial_{\phi})}{r}\phi^1
\notag\\
&+\tfrac{2ar}{\R}\partial_{\phi}\phi^1
-6(a^2 \partial_t + a\partial_{\phi})\phi^0-\tfrac{\Delta}{r^3}\phi^2
+\tfrac{12a^2-6Mr}{r}\phi^0.
\end{align}
\end{subequations}
Note from Remark \ref{rem:EigenvalueSpinWeightedAngular} that in each equation above, the operator on the first line of the RHS is an elliptic operator on $\mathbb{S}^2(r)$ ($r\geq r_+$).
By multiplying $r^{-1}\overline{\phi^0}$ on both sides of \eqref{def:Qphi0phi1S2}, taking the real part and integrating over $\Sigma_{\tau}\cap \{r\geq R_3\}$ with $R_3\geq 5M$ to be fixed, it follows from integration by parts that
\begin{align}\label{eq:ControlOnemorerweightenergyStep1}
\int_{\Sigma_{\tau}\cap \{r\geq R_3\}}r|\nablaslash \widetilde{\phi^0}|^2\lesssim \sum_{i=0,1}E_{\tau}(\widetilde{\phi^i}) + \bigg|\int_{\Sigma_{\tau}\cap \{r\geq R_3\}}\frac{a^2\cos^2\theta}{r}\Re(\partial_{tt}^2 \phi^0 \overline{\phi^0})\bigg|.
\end{align}
We substitute into the last integral the relation
\begin{equation}
\partial_{tt}^2 =\left(\tfrac{\Delta}{\R}V-\tfrac{a}{\R}\partial_{\phi}
-\partial_{r^*}\right)\left(\tfrac{\Delta}{\R}V-\tfrac{a}{\R}\partial_{\phi}
-\partial_{r^*}\right),
\end{equation}
make the replacement $V\phi^0= -r^{-2}\phi^1 -r^{-1}\phi^0$, and perform integration by parts, arriving at
\begin{align}\label{eq:ControlOnemorerweightenergyStep2}
\bigg|\int_{\Sigma_{\tau}\cap \{r\geq R_3\}}\frac{a^2\cos^2\theta}{r}\Re(\partial_{tt}^2 \phi^0 \overline{\phi^0})\bigg|\lesssim \sum_{i=0,1}E_{\tau}(\phi^i) +\int_{\Sigma_{\tau}\cap \{r=R_3\}} |\partial \phi^0|^2.
\end{align}
We can appropriately choose $R_3$ such that the last term is bounded by $CE_{\tau}(\phi^0)$, and obtain
$\int_{\Sigma_{\tau}\cap \{r\geq R_3\}}r|\nablaslash \widetilde{\phi^0}|^2\lesssim E_{\tau}(\widetilde{\phi^0}) +E_{\tau}(\widetilde{\phi^1})$.
By adding the integral over $\Sigma_{\tau}\cap \{r\leq R_3\}$ on both sides, we conclude
\begin{align}\label{eq:ControlInitrPhi0}
\int_{\Sigma_{\tau}\cap \{r\geq R_3\}}r|\nablaslash \widetilde{\phi^0}|^2\lesssim E_{\tau}(\widetilde{\phi^0}) +E_{\tau}(\widetilde{\phi^1}).
\end{align}

Similarly, we can obtain from \eqref{def:Qphi1phi2S2} that
\begin{align}\label{eq:ControlInitrPhi1}
\int_{\Sigma_{\tau}}r|\nablaslash \widetilde{\phi^1}|^2\lesssim E_{\tau}(\widetilde{\phi^0})+E_{\tau}(\widetilde{\phi^1})
+E_{\tau}(\widetilde{\phi^2})+\int_{\Sigma_{\tau}}r|\nablaslash \widetilde{\phi^0}|^2.
\end{align}
The inequality \eqref{eq:ControlInitrPhi01} then follows from combining \eqref{eq:ControlInitrPhi0} with \eqref{eq:ControlInitrPhi1}.
\qed
\end{proof}

\subsection{Proof of Theorem \ref{thm:EneAndMorEstiExtremeCompsNoLossDecayVersion2} for $n\geq 1$}\label{sect:highorderS2}

Let
\begin{align}
\mathbb{X}_0={}&\{\chi_0rYr, \partial_t,\partial_{\phi}\},\\
\mathbb{X}_1={}&\{\chi_1rYr, \partial_t,\partial_{\phi}\},\\
\mathbb{X}_2={}&\{rYr, \partial_t,\partial_{\phi}\},
\end{align}
and define for a multi-index $\mathrm{m}=(\mathrm{m}_1, \mathrm{m}_2, \mathrm{m}_3)$ with $\mathrm{m}_i\geq 0$ $(i=1,2,3)$ that
\begin{align}
\mathbb{X}_0^{\mathrm{m}}\psi= {} (\chi_0 rYr)^{\mathrm{m}_1}\partial_{t}^{\mathrm{m}_2}
\partial_{\phi}^{\mathrm{m}_3}\psi.
\end{align}
Similarly define $\mathbb{X}_1^{\mathrm{m}}\psi$ and $\mathbb{X}_2^{\mathrm{m}}\psi$.
The length of such a multi-index $\mathrm{m}$ is $|\mathrm{m}|=\mathrm{m}_1+\mathrm{m}_2+\mathrm{m}_3$.
Let
\begin{align}
\mathbf{\Phi}_{+2}={}&\{r^{4-\delta}\phi^0_{+2}, r^{2-\delta}\phi^1_{+2}, \phi^2_{+2}\},&
\mathbf{\Phi}_{-2}={}&\{\widetilde{\phi^0_{-2}}, \widetilde{\phi^1_{-2}}, \phi^2_{-2}\}, \notag\\
\mathbf{\Phi}_{+2}'={}&\{\phi^0_{+2}, \phi^1_{+2}, \phi^2_{+2}\},& \mathbf{\Phi}_{-2}'={}&\{\widetilde{\phi^0_{-2}}, r^2\widetilde{\phi^1_{-2}}, \phi^2_{-2}\}.
\end{align}
We will show the estimates \eqref{eq:estiMoraphi012Kerrposicomp} and \eqref{eq:estiMoraphi012Kerrnegacomp} still hold true if we replace all $\varphi_s^i\in \mathbf{\Phi}_{s}$ therein by $\mathbb{X}_0^{\mathrm{m}}\varphi_s^i$ and then sum over all $|\mathrm{m}|\leq n$. For simplicity, we write these estimates required to prove in a unified way
\begin{align}
\label{eq:highordertoprovemain}
\hspace{4ex}&\hspace{-4ex}\sum_{\varphi_s^i\in\mathbf{\Phi_{s}}}\sum_{|\mathrm{m}|\leq n}\left(\mathbf{E}_{\Sigma_{\tau}}(\mathbb{X}_0^{\mathrm{m}}\varphi_s^i)
+\mathbf{E}_{\mathcal{H}^+(0,\tau)}(\mathbb{X}_0^{\mathrm{m}}\varphi_s^i)
+\mathbf{M}_{\mathcal{D}(0,\tau)}(\mathbb{X}_0^{\mathrm{m}}\varphi_s^i)\right)\notag\\
\lesssim{}&\sum_{\varphi_s^i\in\mathbf{\Phi_{s}}}\sum_{|\mathrm{m}|\leq n}\mathbf{E}_{\Sigma_0}(\mathbb{X}_0^{\mathrm{m}}\varphi_s^i),
\end{align}
where $\mathbf{E}$ terms are energy terms and $\mathbf{M}$ are Morawetz terms. We claim that to prove the estimate \eqref{eq:MoraEstiFinal(2)KerrRegularpsiBothSpinComp} for general $n\geq 0$, it suffices to prove \eqref{eq:highordertoprovemain} for general $n\geq 0$.
The reason is as follows. Consider $n=1$ first. Note that the span of $\chi_1 Y$, $\partial_t$, and $\partial_{\phi}$ contains a vector field which is timelike everywhere in $\mathcal{D}(0,\tau)$, and we denote it by $U$. Therefore, by elliptic estimates,
we have for any $\tau\geq 0$,
\begin{align}
\hspace{4ex}&\hspace{-4ex}
\sum_{j=0,1}\sum_{\varphi_s^i\in\mathbf{\Phi_{s}}}
\left(\mathbf{E}_{\Sigma_{\tau}}(U^j\varphi_s^i)
+\mathbf{E}_{\mathcal{H}^+(0,\tau)}(U^j\varphi_s^i)
+\mathbf{M}_{\mathcal{D}(0,\tau)}(U^j\varphi_s^i)\right)
\notag\\
\gtrsim{}&
\sum_{\varphi_s^i\in\mathbf{\Phi_{s}}}\sum_{|\mathrm{k}|\leq1}
\left(
\mathbf{E}_{\Sigma_{\tau}}(\partial^{\mathrm{k}}\varphi_s^i)
+\mathbf{E}_{\mathcal{H}^+(0,\tau)}(\partial^{\mathrm{k}}\varphi_s^i)
+\mathbf{M}_{\mathcal{D}(0,\tau)}(\partial^{\mathrm{k}}\varphi_s^i)
\right).
\end{align}
The LHS of this estimate is bounded by the LHS of \eqref{eq:highordertoprovemain}, and the RHS of \eqref{eq:highordertoprovemain} is manifestly controlled by $\sum_{\varphi_s^i\in\mathbf{\Phi_{s}}}\sum_{|\mathrm{k}|\leq1}
\mathbf{E}_{\Sigma_{\tau}}(\partial^{\mathrm{k}}\varphi_s^i)$. Hence we arrive at
\begin{align}
\hspace{4ex}&\hspace{-4ex}\sum_{\varphi_s^i\in\mathbf{\Phi_{s}}}\sum_{|\mathrm{k}|\leq1}
\left(
\mathbf{E}_{\Sigma_{\tau}}(\partial^{\mathrm{k}}\varphi_s^i)
+\mathbf{E}_{\mathcal{H}^+(0,\tau)}(\partial^{\mathrm{k}}\varphi_s^i)
+\mathbf{M}_{\mathcal{D}(0,\tau)}(\partial^{\mathrm{k}}\varphi_s^i)
\right)\notag\\
\lesssim{}&\sum_{\varphi_s^i\in\mathbf{\Phi_{s}}}\sum_{|\mathrm{k}|\leq1}
\mathbf{E}_{\Sigma_{\tau}}(\partial^{\mathrm{k}}\varphi_s^i).
\end{align}
This is equivalent to the $n=1$ case of the estimate \eqref{eq:MoraEstiFinal(2)KerrRegularpsiBothSpinComp}. The general $n\geq 1$ cases follow by induction.

To prove the estimate \eqref{eq:highordertoprovemain} for any nonnegative $n$, we prove it by induction in $n$.  Assume the estimate \eqref{eq:highordertoprovemain} is valid for $n=n_0-1\geq 0$, we prove it for $n=n_0$. From this assumption and the fact that $\partial_{t}$ and $\partial_{\phi}$ commute with systems \eqref{eq:ReggeWheeler Phi^012KerrS2} and \eqref{eq:ReggeWheeler Phi^012KerrNegaS2}, one obtains
\begin{align}
\label{eq:highordertoprovemain:trivial}
\hspace{4ex}&\hspace{-4ex}\sum_{\varphi_s^i\in\mathbf{\Phi_{s}}}
\sum_{|\mathrm{m}|\leq n_0, \mathrm{m}_1\leq n_0-1}\left(\mathbf{E}_{\Sigma_{\tau}}(\mathbb{X}_0^{\mathrm{m}}\varphi_s^i)
+\mathbf{E}_{\mathcal{H}^+(0,\tau)}(\mathbb{X}_0^{\mathrm{m}}\varphi_s^i)
+\mathbf{M}_{\mathcal{D}(0,\tau)}(\mathbb{X}_0^{\mathrm{m}}\varphi_s^i)\right)\notag\\
\lesssim{}&\sum_{\varphi_s^i\in\mathbf{\Phi_{s}}}\sum_{|\mathrm{m}|\leq n_0, \mathrm{m}_1\leq n_0-1}\mathbf{E}_{\Sigma_0}(\mathbb{X}_0^{\mathrm{m}}\varphi_s^i),
\end{align}
To close the induction, one needs to prove that the LHS of \eqref{eq:highordertoprovemain} with $|\mathrm{m}|=\mathrm{m}_1=n_0$ is bounded by the RHS of \eqref{eq:highordertoprovemain} with $n=n_0$, i.e. to show that one can use $\sum_{\varphi_s^i\in\mathbf{\Phi_{s}}}\sum_{|\mathrm{m}|\leq n_0}\mathbf{E}_{\Sigma_0}(\mathbb{X}_0^{\mathrm{m}}\varphi_s^i)$ to bound the following terms for all $i=0,1,2$:
\begin{align}
\label{eq:highorderphi+22tobound}
\mathbf{E}_{\Sigma_{\tau}\cap [r_+,r_0]}((rYr)^{n_0}\varphi_s^i)
+\mathbf{E}_{\mathcal{H}^+(0,\tau)}((rYr)^{n_0}\varphi_s^i)
+\mathbf{M}_{\mathcal{D}(0,\tau)\cap [r_+,r_0]}((rYr)^{n_0}\varphi_s^i).
\end{align}

Recall in the proofs of Propositions \ref{prop:RedShiftEstiInhomoSWRWE} and \ref{prop:RedShiftEstiInhomoSWRWEtildephi01} that the equation of $\varphi_s^i\in \mathbf{\Phi}_{+2}'\cup\mathbf{\Phi}_{-2}'$ can be put into the following form
\begin{align}
\Sigma\widetilde{\Box}_g\varphi_s^i=G_s^i.
\end{align}
See \eqref{eq:phi+20toRS},
\eqref{eq:eqPsi[-2]} and
\eqref{eq:r2tildephi1toRS} for the source term $G_s^i$ for $\varphi_s^i\in\{\phi_{+2}^0, \widetilde{\phi^0_{-2}}, r^2\widetilde{\phi^1_{-2}}\}$, and the source term $G_s^i$ for $\varphi_s^i\in \{\phi_{+2}^1, \phi_{+2}^2,\phi_{-2}^2\}$ is
$4ias\cos\theta\partial_t\phi^i_{s}+4(\Delta+a^2)/{r^2})\phi^i_{s}+F_{s}^i$.
Note from \eqref{eq:SigmaTildeBoxpsi} and \eqref{def:Lsoperator} that $\Sigma\widetilde{\Box}_g=\mathbf{L}_s+s^2+2ias\cos\theta\partial_t$, which implies the commutator $[\Sigma\widetilde{\Box}_g, rYr]\psi$ is exactly the RHS of \eqref{eq:wavecommutatorwithrYr}. In fact, by commuting with $rYr$ operator $n_0$ times, one arrives at
\begin{align}
[\Sigma\widetilde{\Box}_g, (rYr)^{n_0}]\varphi_s^i
={}&-\tfrac{2n_0(r^2-3Mr+2a^2)}{r}Y((rYr)^{n_0}\varphi_s^i)\notag\\
&
+\sum_{|\mathrm{m}|\leq n_0+1, \mathrm{m}_1\leq n_0} h_{\mathrm{m}}(r) \mathbb{X}_2^{\mathrm{m}}\varphi_s^i.
\end{align}
Therefore, by commuting the equations of $\varphi_s^i\in \mathbf{\Phi}_{+2}'\cup\mathbf{\Phi}_{-2}'$ which are in the form of \eqref{eq:RewrittenFormofSWRWEOpeForm}  with $(\chi_1rYr)^{n_0}$, we get
\begin{align}
\label{eq:phi+22commutewithrYr}
\Sigma \widetilde{\Box}_g((\chi_1rYr)^{n_0}\varphi_s^i)=
{}&f_{n_0}(r)Y((\chi_1rYr)^{n_0}\varphi_s^i)\notag\\
&+\sum_{|\mathrm{m}|\leq n_0+1, \mathrm{m}_1\leq n_0}h_{\mathrm{m},2}(r)\mathbb{X}_2^{\mathrm{m}}\varphi_s^i
\notag\\
&+\sum_{j=0,1}\sum_{|\mathrm{m}|\leq n_0+1, \mathrm{m}_1\leq n_0}h_{\mathrm{m},j}(r)\mathbb{X}_2^{\mathrm{m}}\phi_{+2}^i.
\end{align}
Here, for $i=0,1,2$, $h_{\mathrm{m},i}(r)$ are bounded functions supported in $[r_+,r_1]$, and
\begin{align}
f_{n_0}(r)=
\left\{
  \begin{array}{ll}
    -\tfrac{2n_0(r^2-3Mr+2a^2)}{r}, & \quad \text{for}\  \varphi_s^i\in \{\phi_{+2}^0, \phi_{+2}^1, \phi_{+2}^2,\phi_{-2}^2\},  \\
 \frac{4(r-M)r-5\Delta}{r}-\tfrac{2n_0(r^2-3Mr+2a^2)}{r}, & \quad \text{for}\  \varphi_s^i=\widetilde{\phi^0_{-2}},\\
 \frac{4(r-M)r-9\Delta}{2r}-\tfrac{2n_0(r^2-3Mr+2a^2)}{r}, & \quad \text{for}\  \varphi_s^i=r^2\widetilde{\phi^1_{-2}}.
  \end{array}
  \right.
\end{align}
 Denote the first line on the RHS \eqref{eq:phi+22commutewithrYr} by $G_{s,n_0,m}^i$ and the extra two lines on the RHS by $G_{s,n_0,e}^i$. We apply Lemma \ref{lem:Redshiftspinweightedwavegeneral} with $\tau_1=0$, $\tau_2=\tau$, $\psi=(\chi_1rYr)^{n_0}\phi_{+2}^2$ and $G=G_{+2,n_0,m}^2+G_{+2,n_0,e}^2$ to this equation, and since $f_{n_0}(r)\geq 0$ in $r\leq r_0$, one utilizes the Cauchy--Schwarz inequality to obtain for each $\varphi_s^i\in \mathbf{\Phi}_{+2}'\cup\mathbf{\Phi}_{-2}'$ that
\begin{align}
\hspace{4ex}&\hspace{-4ex}\int_{\mathcal{D}(0,\tau)\cap \{r\leq r_1\}}\Re(\overline{G_{s,n_0,m}^i}\cdot N_{\chi_0}(\chi_1rYr)^{n_0}\varphi_s^i)\notag\\
\lesssim{}& \veps\mathbf{M}_{\mathcal{D}(0,\tau)\cap [r_+,r_1]}((rYr)^{n_0}\varphi_s^i)
+\veps^{-1}\sum_{|\mathrm{m}|\leq n_0, \mathrm{m}_1\leq n_0-1}\mathbf{M}_{\mathcal{D}(0,\tau)\cap [r_+,r_0]}(\mathbb{X}_2^{\mathrm{m}}\varphi_s^i).
\end{align}
This leads to an estimate for any $\varphi_s^i\in \mathbf{\Phi}_{+2}'\cup\mathbf{\Phi}_{-2}'$:
\begin{align}
\label{eq:highorder:laststep1}
\hspace{4ex}&\hspace{-4ex}\mathbf{E}_{\Sigma_{\tau}\cap [r_+,r_0]}((rYr)^{n_0}\varphi_s^i)
+\mathbf{E}_{\mathcal{H}^+(0,\tau)}((rYr)^{n_0}\varphi_s^i)
+\mathbf{M}_{\mathcal{D}(0,\tau)\cap [r_+,r_0]}((rYr)^{n_0}\varphi_s^i)\notag\\
\lesssim{}&\mathbf{E}_{\Sigma_0\cap[r_+,r_1]}((rYr)^{n_0}\varphi_s^i)
+\mathbf{M}_{\mathcal{D}(0,\tau)\cap [r_+,r_0]}((rYr)^{n_0}\varphi_s^i)\notag\\
&+\veps\mathbf{M}_{\mathcal{D}(0,\tau)\cap [r_0,r_1]}((rYr)^{n_0}\varphi_s^i)\notag\\
&
+\veps^{-1}\sum_{j=0}^2\sum_{|\mathrm{m}|\leq n_0, \mathrm{m}_1\leq n_0-1}\mathbf{M}_{\mathcal{D}(0,\tau)\cap [r_+,r_0]}(\mathbb{X}_2^{\mathrm{m}}\varphi_s^j).
\end{align}
Here, the Cauchy--Schwarz inequality is applied to the other bulk integral term $\int_{\mathcal{D}(0,\tau)\cap \{r\leq r_1\}}\Re(\overline{G_{s,n_0,e}^i}\cdot N_{\chi_0}(\chi_1rYr)^{n_0}\varphi_{s}^i)$. Taking $\veps$ small enough allows us to absorb the second term on the RHS of \eqref{eq:highorder:laststep1} by the LHS.
Recall that the span of $\partial_t$ and $\partial_{\phi}$ contains a timelike vector which is timelike in the interior of domain of outer communication, hence by elliptic estimates, the second line on the RHS of \eqref{eq:highorder:laststep1} is bounded by the LHS of \eqref{eq:highordertoprovemain:trivial}. By adding a large multiple of the estimate \eqref{eq:highordertoprovemain:trivial} to \eqref{eq:highorder:laststep1}, the terms in \eqref{eq:highorderphi+22tobound} for any $\varphi_s^i\in\mathbf{\Phi_{s}}$ are bounded by $\sum_{\varphi_s^i\in\mathbf{\Phi_{s}}}\sum_{|\mathrm{m}|\leq n_0}\mathbf{E}_{\Sigma_0}(\mathbb{X}_0^{\mathrm{m}}\varphi_s^i)$, hence the estimate \eqref{eq:highordertoprovemain} follows in the case of $n=n_0$.
\qed

\subsection*{Acknowledgment}
The author is grateful to Lars Andersson,  Pieter Blue and Claudio Paganini for many helpful discussions and comments.


\appendix
\section{Commutators of a Spin-weighted Wave Operator and $rYr$ (or $rVr$)}\label{sect:commutatorwaveandYV}

\begin{prop}
\label{prop:commutatorwaveandYV}
Let $\mathbf{L}_s$ be a spin-weighted wave operator
\begin{align}
\label{def:Lsoperator}
\mathbf{L}_s={}&\Sigma \Box_g+\tfrac{2is\cos\theta}{\sin^2 \theta}\partial_{\phi}-s^2\cot^2 \theta-s^2.
\end{align}
For any scalar $\psi$ with spin weight $s$, we have the following commutators
\begin{align}
[\mathbf{L}_s, -rVr]\psi={}&-\tfrac{2(r^2-3Mr+2a^2)}{r^3}r^2V(rV(r\psi))
+\tfrac{4}{r}(a^2\partial_t +a\partial_{\phi})(rV(r\psi))\notag\\
&
-\tfrac{2(r^2-Mr+3a^2)}{r^2}rV(r\psi)
-2(a^2\partial_t +a\partial_{\phi})\psi
+\tfrac{2Mr-4a^2}{r}\psi,
\label{eq:wavecommutatorwithrVr}\\
[\mathbf{L}_s, rYr]\psi={}&-\tfrac{2(r^2-3Mr+2a^2)}{r^3}r^2Y(rY(r\psi))
+\tfrac{4}{r}(a^2\partial_t +a\partial_{\phi})(rY(r\psi))\notag\\
&
+\tfrac{2(r^2-Mr+3a^2)}{r^2}rY(r\psi)
+2(a^2\partial_t +a\partial_{\phi})\psi
+\tfrac{2Mr-4a^2}{r}\psi.
\label{eq:wavecommutatorwithrYr}
\end{align}
\end{prop}

\begin{proof}
Expand $\mathbf{L}_s \psi$ into the form of
\begin{align}
\mathbf{L}_s\psi
={}&\left(\tfrac{1}{\sin{\theta}} \partial_{\theta}(\sin \theta \partial_{\theta})+\tfrac{\partial_{\phi\phi}^2}{\sin^2\theta}
+2a\partial_{t\phi}^2+a^2 \sin^2 \theta\partial_{tt}^2
+\tfrac{2is\cos\theta}{\sin^2 \theta}\partial_{\phi}-\tfrac{s^2}{\sin^2 \theta}\right)\psi\notag\\
&
-rY\left(\tfrac{\Delta}{r^2} V(r\psi)\right)
+\tfrac{a^2\Delta}{r^2(r^2+a^2)}(V+Y)(r\psi)
+\tfrac{2ar}{\R}\partial_{\phi}\psi
-2ias\cos\theta\partial_t \psi
\notag\\
&
-\left[\tfrac{2Mr^3+a^2r^2-4a^2Mr+a^4}{(\R)^2}+r\sqrt{\R}\partial_r\left(\tfrac{a^2 \Delta}{r^2 (\R)^{3/2}}\right)\right]\psi.
\label{eq:expandformofLswaveop}
\end{align}
We prove the commutator relation \eqref{eq:wavecommutatorwithrVr} below, and the commutator \eqref{eq:wavecommutatorwithrYr} is manifest from \eqref{eq:wavecommutatorwithrVr} by letting $t\to -t$ and $\phi \to -\phi$ (hence $\partial_t\to -\partial_t$, $\partial_{\phi} \to -\partial_{\phi}$ and $V \to -Y$). We calculate the commutators between each term and $-rVr$. The first line of \eqref{eq:expandformofLswaveop} commutes with $-rVr$, and hence their commutators vanish. The last term on the second line commutes with $-rVr$, and for the other terms on the second line, we have
\begin{align}
\hspace{4ex}&\hspace{-4ex}[rY\left(\tfrac{\Delta}{r^2} Vr\right), -rVr]\psi\notag\\
={}&r^3 [Y,V]\left(\tfrac{\Delta}{r^2}V(r\psi)\right)
-2r^2 V\left(\tfrac{\Delta}{r^2}V(r\psi)\right)\notag\\
&- 2r\tfrac{\Delta}{r^2}V(r\psi)
-rY\left(r^2 \partial_r\left(\tfrac{\Delta}{r^2}\right)V(r\psi)\right),
\label{eq:secondcommutatorrela}\\
\hspace{4ex}&\hspace{-4ex}
[\tfrac{a^2\Delta}{r^2(r^2+a^2)}(V+Y)r,-rVr]\psi\notag\\
={}&-\tfrac{a^2\Delta}{\R}[Y,V](r\psi)
+\tfrac{2a^2r\Delta}{r^2(\R)} V(r\psi)\notag\\
&
+r\partial_r\left(\tfrac{a^2\Delta}{r(\R)}\right)Y(r\psi)
+r\partial_r\left(\tfrac{a^2\Delta}{r^3(\R)}\right)r^2V(r\psi),
\label{eq:thirdcommutatorrela}\\
\hspace{4ex}&\hspace{-4ex}
[\tfrac{2ar}{\R}\partial_{\phi},-rVr]\psi
={}-\tfrac{2ar^2(r^2-a^2)}{(\R)^2}\partial_{\phi}\psi.
\end{align}
The commutator of the last line with $-rVr$ equals
\begin{align}
\hspace{4ex}&\hspace{-4ex}-r^2\partial_r\left(\tfrac{2Mr^3+a^2r^2-4a^2Mr+a^4}{(\R)^2}+r\sqrt{\R}\partial_r\left(\tfrac{a^2 \Delta}{r^2 (\R)^{3/2}}\right)\right)\psi\notag\\
={}&-r^2 \partial_r\left(\tfrac{2(Mr-a^2)}{r^2}\right)\psi\notag\\
={}&\tfrac{2Mr-4a^2}{r}\psi.
\end{align}
It remains to calculate the commutator $[Y,V]$ which is present in both \eqref{eq:secondcommutatorrela} and \eqref{eq:thirdcommutatorrela}. For a general field $\psi$,
\begin{align}
[Y,V]\psi={}&-2\partial_r\left(\tfrac{\R}{\Delta}\right)\partial_t \psi
-2\partial_r\left(\tfrac{a}{\Delta}\right)\partial_{\phi}\psi\notag\\
={}&\tfrac{4M(r^2-a^2)}{\Delta^2}\partial_t\psi
+\tfrac{4a(r-M)}{\Delta^2}\partial_{\phi}\psi.
\end{align}
Collecting the above discussions and calculations, we arrive at
\begin{align}
\hspace{4ex}&\hspace{-4ex}[\mathbf{L}_s, -rVr]\psi\notag\\
={}&-\left(\tfrac{2\Delta}{r}-\tfrac{2(Mr-a^2)}{r}\right)V(rV(r\psi))
+\tfrac{4}{r}(a^2\partial_t +a\partial_{\phi})(rV(r\psi))\notag\\
&-\tfrac{2(r^2-Mr+3a^2)}{r^2}rV(r\psi)
-2a^2\partial_t \psi
-\tfrac{2a^3(3r^2+a^2)+2ar^2(r^2-a^2)}{(\R)^2}\partial_{\phi}\psi\notag\\
&+\tfrac{2Mr-4a^2}{r}\psi.
\end{align}
The relation \eqref{eq:wavecommutatorwithrVr} follows by calculating the coefficient of each single term in the above equation.
\qed
\end{proof}

\newcommand{\mnras}{Monthly Notices of the Royal Astronomical Society}
\newcommand{\prd}{Phys. Rev. D}
\newcommand{\apj}{Astrophysical J.}
\providecommand{\MR}{\relax\ifhmode\unskip\space\fi MR }
\bibliographystyle{amsplain}

\end{document}